\documentclass[aps,pre,twocolumn,groupedaddress,longbibliography,showkeys,showpacs]{revtex4-1}

\usepackage{graphicx}
\usepackage{bm}
\usepackage{amssymb,amsfonts}
\usepackage{amsmath}
\usepackage{amsthm}
\usepackage{graphicx}
\usepackage[usenames,dvipsnames]{xcolor}
\usepackage{soul}
\usepackage{url}
\graphicspath{{Figures/}}
\newtheorem{prop}{Proposition}

\newcommand{\hmb}[1]{{\color[rgb]{0,0,1}{#1}}}

\begin{document}

\title{Critical Phenomena in  Complex Networks: from Scale-free to Random Networks}


\author{Alexander I. Nesterov}
   \email{nesterov@cencar.udg.mx}
\affiliation{Departamento de F{\'\i}sica, CUCEI, Universidad de Guadalajara,
 Guadalajara, CP 44420, Jalisco, M\'exico}

\author{Pablo H\'ector Mata Villafuerte}
   \email{themata@hotmail.com}
\affiliation{Departamento de F{\'\i}sica, CUCEI, Universidad de Guadalajara,
 Guadalajara, CP 44420, Jalisco, M\'exico}

\date{\today}

\begin{abstract}

Within the conventional statistical physics framework, we study critical phenomena in a class of configuration network models with hidden variables controlling links between pairs of nodes.  We find analytical expressions for the average node degree, the expected number of edges, and the Landau and Helmholtz free energies, as a function of the temperature and number of nodes. We show that the network's temperature is a parameter that controls the average node degree in the whole network and the transition from unconnected graphs to a power-law degree (scale-free) and random graphs. With increasing temperature, the degree distribution is changed from power-law degree distribution, for lower temperatures, to a Poisson-like distribution for high temperatures. We also show that phase transition in the so-called {\em Type A} networks leads to fundamental structural changes in the network topology. Below the critical temperature, the graph is completely disconnected. Above the critical temperature, the graph becomes connected, and a giant component appears.

\end{abstract}

 \keywords{complex networks; statistical mechanics; graph ensembles; phase transitions; hidden variables;  graph temperature}

\maketitle

\section{Introduction}

Network science has contributed to very diverse fields in both the natural and human sciences due to its intrinsic interdisciplinary nature. The phenomena and processes in networks belonging to nature's fundamental structures are quite different from those in lattices and fractals. That is why studying these intriguing effects will lead to a new understanding of a broad class of natural, artificial, and social systems \cite{BB1, DSMF, GCal, BABM, BAL, MN2018}.

Current research in complex networks (or graphs) focuses on three main classes of models: random graphs, small-world, and scale-free networks \cite{ARB}.  {In contrast to regular networks, in the random graph, some properties' values are fixed. Still, others, such as the number of nodes, edges, and connections between them, are determined randomly \cite{BB1, BABM, BAL, MN2018}. }

In general, small-world networks are characterized by the property that the typical distance between any pairs of nodes is short; it depends logarithmically on the number of nodes. A subclass of the small-world networks, the so-called Watts-Strogatz networks, are additionally characterized by a relatively high clustering level compared to a random graph with the same node-edge density \cite{WDSS}. This subclass is often identified with the whole class of \st{the} small-world networks. On the whole, small-world networks are intermediate between highly clustered regular lattices and random graphs. The coexistence of small path length and clustering may be taken as the main feature of small-world networks \cite{ARB,KJSD}.
 
Scale-free networks, having a power-law degree distribution, are characterized by large hubs, i.e., a few nodes highly connected to other network nodes. Nowadays, scale-free models are of significant interest since many real networks such as social networks, airline networks, the World Wide Web, computer networks, the Internet, and others, can be treated as scale-free networks \cite{BB1,DSMF, BABM, GCal, BAL1, BAL, ARB,KJSD, MN2018, GMNM, VIHP}. 

During the two last decades, methods of statistical mechanics applied to complex networks became a powerful tool for the study and explanation of the properties of real-world networks 
\cite{NMSW, MEJN1, RPMR, ARB, MN2018, PJNM, CDLM, CGST, BAL1, BAL}. The idea that a 
statistical approach is adequate to study complex networks is a natural one, since networks are 
large complex systems, and a deterministic approach cannot describe their collective behavior. 
The recent development of these methods has revealed new and unexpected challenges in the 
statistical physics of networks. One of them is the concept of network temperature and its function in the formation and dynamics of complex networks \cite{MDC}.

Usually, the temperature of a network is considered as a dummy variable. However, to complete 
the analogy with statistical physics, this concept should be more meaningful. One of the 
successful attempts to introduce the {\em graph temperature} as a parameter which controls 
clustering and/or the degree of topological optimization of a network, was made in Refs. 
\cite{CDAS,KDPF1,KDPF2}. Further progress in this approach has been achieved by determining 
the network temperature in terms of empirical data, such as the number of nodes, average node 
degree, and exponent of the degree distribution \cite{ANHM}.

Critical phenomena in networks result in drastic changes in the networks' topological properties, 
such as cluster and community structure of the network, the emergence of percolation and a 
giant connected component, etc. \cite{DCNMS, DSGAM}. This leads to another challenge: is it 
possible to treat the critical phenomena in complex networks using statistical physics methods 
in terms of thermodynamic potentials?

 { The purpose of this paper is twofold. First, to reveal the role of temperature. Second, to employ methods of statistical mechanics to study the critical phenomena in detail. In this paper, we study the statistical properties of complex networks with hidden variables, assigned to each edge $\langle i,j \rangle$ of the network \cite{PJNM,BMPSR, SMKD1}. Motivated by the importance of scale-free models for real networks, we concentrate on class configuration models consistent with the scale-free networks. 

 We use essentially analytical methods in our work, in contrast to the numerical simulations that form most studies' core. While our approach works equally well for undirected and directed graphs, we restrict ourselves to considering the undirected case only for the sake of simplicity. The generalization to directed graphs is straightforward. It's worth mentioning that in our work we consider only unweighted networks. The extension of our approach to the weighted network models requires a more profound development of the analytical methods.}

The paper is organized as follows. In Sec. II, we discuss the statistical properties of complex networks. In particular, we show how a network's temperature can be determined in terms of the empirical data available, such as the average node degree, the number of nodes, and degree exponent. In Sec. III, we explore phase transitions in complex networks with hidden variables. The transition from an unconnected to a connected network, the formation of a giant component, and other essential aspects of network phase transitions are discussed. In the Conclusion, we summarize our results and discuss possible generalizations of our approach.  {In Appendix A we study analytical properties of the global clustering coefficients in  the limit of low temperatures.} In Appendix B, we estimate the critical exponents of the phase transition.

\section{ Statistical description of complex networks} 

\subsection{General formalism }

A network is a set of $N$ nodes (or vertices) connected by $L$ links (or edges). One can describe the network by an adjacency matrix, $a_{ij}$, where each existing or non-existing link between pairs of nodes ($ij$) is indicated by a 1 or 0 in the $i,j$ entry. Individual 
nodes possess local properties such as node degree $k_i = \sum_j a_{ij}$, and clustering coefficient $c_i = \sum_{jk} a_{ij}a_{jk}a_{ki}/k_i(k_i -1)$ 
\cite{WDSS,NMSW,BSLV,ARB}. The network as a whole can be described quantitatively by its 
degree distribution $P_k$ and connectivity. The connectivity is characterized by the connection 
probability $p_{ij}$, i.e. the probability that a pair of nodes $(ij)$ is connected.  

Unlike the conventional approach to the statistical mechanics, where the Gibbs distribution is derived by considering a system in weak interaction with the environment, the statistical description of complex networks is based on the informational Shannon-Gibbs entropy, subject to certain constraints \cite{PJNM,STCD,HPLG}. For a particular graph $G \in \mathcal G$, belonging to ensemble of graphs, $\mathcal G$, we denote by  $P(G)$ the probability of obtaining this graph. Consider the entropy of the graph, $S = - \sum P(G) \ln P(G)$, and assume that the ollowing additional constraints are imposed: $\sum P(G) = 1$ and $E = \sum H(G)P(G)$, where $H(G)$ is the graph Hamiltonian and $E$ is the ``energy'' of the graph. 

 {From the principle of the maximum entropy, that taking into account the imposed constraint can be written as
 \begin{align}
 & \frac{\partial}{\partial  P(G)}\Big (- \sum P(G) \ln P(G) + \lambda \big ( \sum P(G) - 1\big ) \nonumber \\
   & + \beta \big(E -  \sum H(G) P(G)\big) \Big)=0,
    \end{align}
 where $\lambda$ and $\beta $ are the Lagrange multipliers,  we obtain the convenient Gibbs distribution:
   \begin{align}
    P(G) = \frac{1}{Z}e^{-\beta  H(G) }.
   \end{align}
 The expected value of any function, $X (G)$, is calculated as follows:
\begin{align}
 \langle X \rangle= \sum X(G) P(G)	
\end{align}
  }

The computation of the Shannon-Gibbs entropy yields
  \begin{align}
  S = -\sum 	P(G) \ln P(G) = \ln Z + \beta E.
  \end{align}
Employing the relation $\partial S/\partial E = 1/T $, we find that the Lagrange multiplier $\beta = 1/T$  is the inverse ``temperature'' of the graph. Next, using the thermodynamic relation $F= E -TS$, where $F$ is the Helmholtz free energy, we find $F = - T \ln Z$.
  
 To fit the model to the real empirical network, $G^\ast$\hmb{,} one can apply the maximum-likelihood principle, which specifies the parameter choice for the given set of constraints \cite{STCD}. 
  The log-likelihood
  \begin{align}
  	\mathcal L =\ln P(G^\ast) = - \beta H(G^\ast)  - \ln Z (\beta)
  \end{align}
  is maximized by a particular choice of $\beta^\ast$ yielding $E^\ast = -\partial \ln Z/\partial 
  \beta^\ast$, where $E^\ast$ is the empirical value measured on the real network:
  \begin{align}
  	E(G^\ast) = E^\ast = \sum H(G)P(G|\beta^\ast).
  \end{align}
  Thus, the temperature of the graph becomes a parameter that can be defined from the empirical data.  

The most general statistical description of an undirected network in equilibrium, with a fixed number of vertices $N$ and a varying number of links, is given by the grand canonical ensemble \cite{PJNM,CDLM,CDAS,CGST}. For an undirected network with fixed number of vertices $N$ and varying number of links, the probability of obtaining a graph, $A$, with energy $E_A$ can be written as \cite{PJNM,CDLM,CDAS,CGST}
\begin{align}
P_{A}=\frac{1}{\mathcal{Z}} \exp \big (\beta(\mu L_{A}-E_{A})\big),
\end{align}
where  $\beta =1/T$, with $T$ being the network temperature, $\mu$ is the chemical potential, and $L_{A}=\sum_{i j} a_{i j}$ is number of links in the graph $A$. The partition function reads 
\begin{align}
\mathcal{Z} =\sum_{A} \exp \big (\beta(\mu L_{A}-E_{A})\big).
\end{align}

The temperature is a parameter that controls clustering, and the chemical potential controls the link density and the connection probability in the network. It's worth noting that $P_{A}=2^{-N(N-1) / 2}$ for all graphs when $T \rightarrow \infty$, and when $T \rightarrow 0$ we have $P_A =1$ for the graph with the maximum value of $\mu L_{A}-E_{A}$, and $P_A = 0$ for all other graphs \cite{CDAS}.

 {To obtain the grand potential, $\Omega$, which we will refer to as the Landau free energy, we use the relation $\Omega = - \beta^{-1}\ln \mathcal Z$. Next, one can recover the Helmholtz free energy $F$, internal energy $ E $, entropy $S$, and 
heat capacity $C_N$, using the following relations:
\begin{align}   \label{Eq3a}
     F  &=\Omega + \mu L, \\
      E  &  =F+ TS , \label{EqE}\\
      S  &  =- \frac{\partial \Omega}{\partial T}\bigg |_{\mu} , \label{EqS}\\
      C_N & =T \frac{\partial S}{\partial T}.
     \label{Eq3b}
\end{align}
Finally, having the Landau free energy, one can find the expected number of links as follows: $L =-\partial \Omega/\partial \mu$.

{\em Remark.} If we have empirical information about the number of nodes, average node degree, chemical potential, and other parameters that can be included in the density of states, i.e., the exponent of the degree distribution, etc., then we can define the temperature of a given network employing the equation of state,  $L = - \partial \Omega/\partial \mu$. 

\subsection{Fermionic and bosonic graphs}

Let us assign to each edge $\langle i,j \rangle$ the `energy' $\varepsilon_{ij}$. Then the energy of the graph can be written as $E_{A} =\sum_{i<j} \varepsilon_{i j} a_{i j}$, and the partition function and the graph probability are given by \cite{CDAS}
\begin{align} \label{Z}
    \mathcal{Z}=&\sum_{\{A\}} \prod_{i<j} e^{\beta\left(\mu-\varepsilon_{i j}\right)  a_{i j} }.
\end{align}
The connection probability of the existing link between nodes $i$ and $j$ is
\begin{align}
	 p_{i j} =- \frac{\partial \ln \mathcal{Z}}{\partial (\beta \varepsilon_{ij})}.
	 \label{EqP}
\end{align}
The probability $p_{i j}$ is equivalent to the expected number of edges $\bar a_{i j}$ between vertices $i$ and $j$, namely, $\bar a_{i j}= p_{i j}$.

Consider two sets of undirected graphs $G$: (1) simple ({\em fermionic}) graphs with only one edge allowed between any pair of vertices; (2) multi-edge ({\em bosonic}) graphs with any number of edges between any pair of vertices (excluding self-edges) \cite{PJNM}. The computation of the partition function yields
\begin{itemize}
\item Fermionic graph: 	
\begin{align} \label{Z}
    \mathcal{Z}=\prod_{i<j}\big (1+e^{\beta\left(\mu-\varepsilon_{i j}\right)  }\big ).
\end{align}
\item Bosonic graph: 	
\begin{align} \label{Z}
    \mathcal{Z} =\prod_{i<j}\frac{1}{  \big (1-e^{\beta\left(\mu-\varepsilon_{i j}\right)  }\big )} .
\end{align}
\end{itemize}
Employing Eq. \eqref{EqP}, we obtain
\begin{align}
    p_{i j}=\frac{1}{e^{\beta \left(\varepsilon_{i j}-\mu\right) }\pm1},
    \label{FDBE}
\end{align}
where the upper sign corresponds to the Fermi-Dirac distribution and lower sign to the Bose-Einstein distribution.

Using the relation $\Omega = -\beta^{-1} \ln \mathcal Z$, we obtain
\begin{align}
\Omega = \mp \beta^{-1} \sum_{i<j}\ln \big(1 \pm e^{\beta\left(\mu-\varepsilon_{ij } \right)  }\big).
\label{Eq7f}    
\end{align}
Having the Landau free energy, one can find the expected number of links as follows: $L =-\partial \Omega/\partial \mu$. The computation yields
\begin{align}
L = \sum _{i < j } \frac{1}{e^{\beta \left(\varepsilon_{i j}-\mu\right) }\pm1}.
\label{Eq.3_1}
\end{align}

The expected degree of a vertex $i$ reads
\begin{align}
\bar k_i= \sum_j p_{ij}.	
\label{DV1a}
\end{align}
Denoting the average node degree in the whole network with $\langle  k \rangle = (1/N)\sum_i  \bar k_i $ , we obtain
\begin{align}
\langle  k \rangle  =\frac{2}{N}\sum _{i < j } \frac{1}{e^{\beta \left(\varepsilon_{i j}-\mu\right) }\pm1},
\label{Eq.4_1}
    \end{align}
  Comparing this expression with \eqref{Eq.3_1}, we obtain the following relation between the expected number of links and average node degree: $L  =  N \langle  k \rangle/2$.

\subsection{Configuration fermionic model}

We will focus now on the model introduced in \cite{PJNM} with the graph Hamiltonian given 
by
\begin{align}
	H(G) = \sum_i \varepsilon_i k_i,
	\label{SGH}
\end{align}
where  $\varepsilon_i $ is an``energy"   assigned to each vertex $i$. Noting that $k_i =  \sum_j a_{ij}$, one can rewrite \eqref{SGH} as follows:
\begin{align}
H(G) = \sum_{i<j}\varepsilon_{ij} a_{ij}
 = \sum_{i<j}(\varepsilon_{i}   + \varepsilon_{j}  )a_{ij}  . 
\label{SGH1}
\end{align}

The Landau free energy \eqref{Eq7f} takes the form
\begin{align}
\Omega = -\beta^{-1} \sum_{i<j}\big (1+e^{\beta\left(\mu-\varepsilon_{i } - \varepsilon_{j }\right)  }\big ).
\label{Eq7f}    
\end{align}
For $N \gg 1$ one can replace the sum by an integral, $\frac{2}{N(N-1)}\sum_{i<j} \rightarrow \iint$, and recast \eqref{Eq7f} as
\begin{align}
\Omega = -\frac{N(N-1)}{2 \beta} \iint\ln \big (1 + e^{\beta (\mu- \varepsilon - \varepsilon') } \big) \rho (\varepsilon) \rho (\varepsilon')d \varepsilon  d \varepsilon',
\label{Eq7h}    
\end{align}
where $\rho(\varepsilon)$ denotes the density of states, with the standard normalization $\int \rho (\varepsilon) d \varepsilon =1$. Next, employing the relation $L =-\partial \Omega/\partial \mu$, we obtain
\begin{align}
L = \frac{N(N -1)}{2}\iint \frac{\rho (\varepsilon ) d \varepsilon \rho (\varepsilon') d \varepsilon'}{e^{\beta \left(\varepsilon + \varepsilon' -\mu\right) }+1}.
\label{Eq.3}
\end{align}

In the same limit, the expected degree of a vertex $i$ with energy $\varepsilon_i$ is
\begin{align}
	 \bar k (\varepsilon_i ) = (N - 1)\int p(\varepsilon_i, \varepsilon' ) \rho(\varepsilon' ) d \varepsilon' ,
\end{align}
 where
\begin{align}
    p(\varepsilon_i, \varepsilon') = \frac{1}{e^{\beta \left(\varepsilon_i + \varepsilon' -\mu\right) }+1}.
    \label{FD}
\end{align}
For the average node degree in the whole network this yields\begin{align}
\langle  k \rangle  =\int \bar k(\varepsilon) \rho (\varepsilon) d \varepsilon = (N -1) \iint \frac{\rho (\varepsilon ) \rho (\varepsilon') d \varepsilon d \varepsilon'}{e^{\beta \left(\varepsilon + \varepsilon' -\mu\right) }+1}.
\label{Eq.4}
    \end{align}

{\em Remark.} Unlike the soft configuration models, where the expected degree sequence is constrained to a given sequence, expected degrees are random variables in our model. In recent terminology these models are called {\em hypersoft configuration models } \cite{KAKD,HPLG,VIHPKM}. 

The density of a simple network is characterized by the {\em connectance} or {\em density} $\varrho $ of the network, which is defined as the fraction of those edges that are actually present \cite{MN2018}
\begin{align}
	\varrho = \frac{L}{L_{\max}} = \frac{2L}{N(N-1)}.
	\label{DN}
\end{align}
As one can see, the range of density values is $0 \leq \varrho \leq 1$. Using the relation $L  =  N \langle  k \rangle/2$, one can rewrite the expression for the network's density \eqref{DN} in the equivalent form $ \varrho =\kappa $, where $ \kappa =\langle k \rangle/(N-1) $ is the average node degree per node. 

There are two important classes of networks to consider. One is when the network remains connected ($\varrho$ is non-zero) in the limit of large $N$ ({\em dense} network). The opposite case is when the network becomes empty ({\em sparse}) for large $N$, i.e. $\varrho\rightarrow 0$.

 We will say that a network is {\em asymptotically sparse} if for large $N$ the density $\varrho\rightarrow 0$ as $T\rightarrow T_c$, where $T_c$ is a critical temperature. It's worth noting that for both networks the density converges to $  \varrho _\infty= 1/2$ as $T\rightarrow \infty$.
}

\subsubsection*{ Clustering coefficient} 

The  clustering coefficient is defined as the probability that two 
nodes, connected to a third node, will also be connected to each other \cite{WDSS}. For a given node $i$ with energy $\varepsilon_i$, the local clustering coefficient, $ c_i=
c(\varepsilon_i ) $, can be calculated as follows \cite{BMPSR}:
\begin{align}
c_i= \frac{(N -1)^2}{\bar k^2(\varepsilon_i )} \iint p(\varepsilon_i, \varepsilon'  ) p(\varepsilon', \varepsilon''  ) p(\varepsilon_i, \varepsilon''  )\rho(\varepsilon') \rho(\varepsilon'')  d \varepsilon' d \varepsilon''.
\label{C1}
\end{align}
The clustering coefficient for vertices of degree $k$ is given by \cite{BMPSR}
\begin{align}
\bar c_k=\frac{1}{P_k} \int g(k | \varepsilon) c (\varepsilon) \rho(\varepsilon) d\varepsilon , \, k=2,3, \dots ,
\label{EqP1}
\end{align}
where $g(k | \varepsilon)$ is the propagator, yielding
\begin{align}
P_k= \int g(k | \varepsilon) \rho(\varepsilon) d\varepsilon.
\end{align}

 {There are two possible characterizations of the global clustering coefficient. First, and most used, is by averaging of $c_k$ over the whole network, i.e. writing $C_1 =\sum_k c_k P_k$. In the continuous limit we obtain 
\begin{align}
 C_1=  \int c(\varepsilon) \rho(\varepsilon )d \varepsilon.
 \label{C1}
\end{align}
In the computation of this expression, we used the normalization condition $\sum_k g(k | \varepsilon) =1$. Substitution of $c(\varepsilon)$ yields
\begin{align}
    C_1= \int d\varepsilon\rho(\varepsilon)  \frac{\iint  p(\varepsilon, \varepsilon'  ) p(\varepsilon', \varepsilon''  ) p(\varepsilon, \varepsilon''  )\rho(\varepsilon' )\rho(\varepsilon'')  d \varepsilon' d \varepsilon'' }{\big (\int p(\varepsilon, \varepsilon' ) \rho(\varepsilon' ) d \varepsilon' \big )^2  }.
    \label{C1a}
\end{align}
}
The second definition is as follows \cite{BAWM,PJNM}:
\begin{align}
    C_2=\frac{3\times\text { number of triangles }}{\text { number of connected triples }}.
    \label{C2}
    \end{align}
This can be recast as \cite{BBOR}
\begin{align}
    C_2=\frac{\bar \sum_i \bar k(\varepsilon_i)(\bar k(\varepsilon_i) -1) c(\varepsilon_i ) }{\sum_i  \bar k(\varepsilon_i)(\bar k(\varepsilon_i) -1)}.
    \label{C2a}
    \end{align}
In the continuous limit we have
\begin{align}
    C_2=\frac{\int {\bar k}^2(\varepsilon)  c(\varepsilon) \rho(\varepsilon )d \varepsilon }{\int {\bar  k}^2(\varepsilon)  \rho(\varepsilon )d \varepsilon }.
    \label{C2a}
\end{align}
 {This can be rewritten as
 \begin{align}
    C_2= \frac{\iiint  p(\varepsilon, \varepsilon'  ) p(\varepsilon', \varepsilon''  ) p(\varepsilon, \varepsilon''  )\rho(\varepsilon) \rho(\varepsilon' )\rho(\varepsilon'') d \varepsilon d \varepsilon' d \varepsilon'' }{\int \big (\int p(\varepsilon, \varepsilon' ) \rho(\varepsilon' ) d \varepsilon' \big )^2 \rho(\varepsilon )d \varepsilon }.
    \label{C2b}
\end{align}
}
In the limit of high temperatures, the connection probability $p(\varepsilon, \varepsilon') \rightarrow 1/2$, as $T \rightarrow \infty$.  Using this result, one can show that the average node degree in the whole network, $\langle  k \rangle$, converges to $N/2$, and for both definitions the average clustering coefficient $C_{1,2}\rightarrow 1/2$, when $T \rightarrow \infty$.

{\em Comment.} The definitions for the global clustering coefficient introduced above are highly non-equivalent, i.e., in some situations, one can obtain $C_1=1$ and $C_2=0$ (see, for instance, the discussion in  Refs. \cite{BBOR}). In Section III, we will show that the first definition leads to the wrong results for asymptotically sparse networks.

\subsubsection*{Generating function formalism}

Critical phenomena in networks with arbitrary degree distribution can be treated successfully using the generating function formalism \cite{WHS,NMSW}. Following  \cite{NMSW}, we define a generating function as
\begin{align}
G_0(z)	= \sum_k z^k P_k,
\label{Eq12}	
\end{align}
where $P_k$ is the degree distribution (the probability that any given vertex has degree $k$). Further, all calculations will be confined to the region $0\leq z\leq1$. 

To compute $G_0(z)$, we use the generating function formalism for networks with hidden variables developed in \cite{BMPSR}. If we consider $\varepsilon$ as a hidden variable, then the degree distribution can be written
\begin{align}
P_k= \int g(k | \varepsilon)\rho(\varepsilon)  d\varepsilon,
\end{align}
where $g(k | \varepsilon)$ denotes the propagator,  with the normalization condition $\sum_k g(k | \varepsilon) =1$. Substituting $P_k$ in Eq. \eqref{Eq12}, we obtain
\begin{align}
G_0(z)	= \int  d\varepsilon\rho(\varepsilon)\sum_k z^k g(k | \varepsilon).
\label{Eq12p}	
\end{align}
As shown in \cite{BMPSR},
\begin{align}
\ln \sum_k z^k g(k | \varepsilon)= N\int d\varepsilon^{\prime}\rho(\varepsilon^{\prime}) \ln \big (1-(1-z) p(\varepsilon, \varepsilon^{\prime})\big).	
\end{align}
Using this result in Eq. \eqref{Eq12p}, we obtain
\begin{align}
G_0(z)	=\int d\varepsilon\rho(\varepsilon)\exp \Big (N\int d\varepsilon^{\prime}\rho(\varepsilon^{\prime}) \ln \big (1-(1-z) p(\varepsilon, \varepsilon^{\prime})\big)\Big ).
\label{Eq13f}
\end{align}

Having the generating function, one can easily calculate the degree distribution and its moments:
\begin{align} \label{Eq13a}	
P_k =  \frac{1}{k!}\frac{d^k}{dz^k}G_0(z)\Big|_{z=0}, \\
	\langle k^n \rangle =  \Big(z \frac{d}{dz}\Big)^nG_0(z)\Big|_{z=1}.
\label{Eq13b}	
\end{align}
In particular, this yields
\begin{align} \langle k \rangle =G'_0(1), \,\,
	\langle k^2 \rangle = G''_0(1) + G'_0(1).
	\label{Eq13c}
\end{align}

Further, it is convenient to introduce the abbreviation for derivatives of the generating function:
\begin{align}
	z_n =\frac{d^n}{dz^n}G_0(z)\Big|_{z=1}.
	\label{Eq13d}
\end{align}
Then, using Eq. \eqref {Eq13b},  we obtain $ \langle k \rangle = z_1$, $ \langle k^2 \rangle = z_2  + z_1$, etc.

 Using \eqref{Eq13f} and the results of Ref. \cite{BMPSR},  we find that the generating function can be written as
\begin{align}
G_0(z)    =\int e^{(z-1) \bar k (\varepsilon ) }   \rho (\varepsilon) d \varepsilon ,
\label{Eq13p}    
\end{align}
where  $ \bar k (\varepsilon ) = N \int p(\varepsilon, \varepsilon' ) \rho(\varepsilon' ) d \varepsilon' $ is the expected degree of the node with the hidden variable $\varepsilon$. Straightforward computation leads to the following relation: $z_n =\langle \bar  k^n \rangle$,
where we denote by $ \langle \bar  k^n \rangle$ the 
$n$-th moment of the node degree with the hidden variable $\varepsilon$,
 \begin{align}
 \langle \bar  k^n \rangle	 = \int  \bar k^n(\varepsilon) \rho (\varepsilon) d \varepsilon.
 \label{Eq.4a}
 \end{align}
Using  the relation $z_n =\langle \bar  k^n \rangle$ and Eq. \eqref{Eq13c} , one can calculate the $n$-th moment of the degree distribution, $\langle  k^n \rangle$, if we know the corresponding moments for the hidden variables, $\langle \bar  k^n\rangle$. In particular, we obtain
\begin{align}
	 \langle k^2 \rangle = \langle \bar k^2\rangle  + \langle  k \rangle.
\end{align}

\subsection{Configuration models consistent with scale-free networks}

Scale-free networks are characterized by a power-law degree distribution,  $P(k)\propto k^{-\gamma}$, where $1\leq k \leq k_0$, and the exponent of the distribution is $\gamma > 1$. We will focus now on the class of configuration models consistent with the scale-free networks. The simplest example is the soft configuration model, whose expected degree sequence is constrained to a given (observed) sequence  $\{k^\ast_i \}$  \cite{PJNM,CFLL,CFLL1,CGCA, SVCG,TSGD},
\begin{align}
k^\ast_i=\bar k(\varepsilon_i) = \sum _j p_{ij}.	
\end{align}

In the limit of $N \gg 1 $, the expected degree $ \bar  k (\varepsilon)$ of a vertex with energy $\varepsilon$ can be written as 
\begin{align}
    \bar k (\varepsilon)  = (N- 1) \int \frac{\rho (\varepsilon' )d \varepsilon' }{e^{\beta \left(\varepsilon+ \varepsilon' -\mu\right) }+1}.
    \label{ED1}
\end{align}
If $ k =\bar k (\varepsilon) $ is a monotonous function of $\varepsilon$, then
the probability distribution $P(k)$ is given by \cite{CGCA, SVCG}
\begin{align}
P(k) =\rho(\varepsilon(k)) \varepsilon'(k).
\label{PK}
\end{align}

 To specify the model one should define a density of state, $\rho(\varepsilon ) $.  {Let us assume that it is distributed according to $\rho(\varepsilon 
)  \propto \beta_c(\gamma -1)e^{\beta_c(\gamma- 1)\varepsilon}$, where $0 \leq 
\varepsilon \leq \varepsilon_0$; $\beta_c $ is a constant with dimension of inverse 
temperature. Substituting $\rho(\varepsilon)$ in \eqref{ED1}, we find that in the 
low-temperature limit the leading term is $\bar  k (\varepsilon) \propto e^{-\beta \varepsilon}$. Using this result, we assign to each node with  the energy $\varepsilon$  a degree $ k (\varepsilon) =e^{\beta ( \varepsilon_0 - \varepsilon)}$. The parameter $k$ is bounded by $k_0 =e^{\beta \varepsilon_0 }$.  Substituting $k(\varepsilon)$ in Eq. \eqref{PK}, we obtain a power-law degree distribution,  $P(k)\propto k^{-\tau}$ with $\tau = \beta_c (\gamma -1)/\beta+1$}. Thus, in the low-temperature regime, our model describes a scale-free network. 

As pointed out in Ref. \cite{CDAS}, the soft configuration model considered above can also be regarded as a particular case of the so-called {\em fitness model} introduced in \cite{CGCA}. In the fitness model, each node $i$ is characterized by a {\em fitness}   $x_i = e^{-\beta \varepsilon_i}$, and the connection probability takes the form $p_{i j} = f(x_ix_j)$, with $f(x_i,x_j)$ being a symmetric function of the variables $x_i$ and $x_j$.  In terms of the hidden variables,  $x_i$, we  can rewrite Eq. \eqref{FD} as
\begin{equation}
p_{i j}=\frac{1}{e^{\beta \left(\varepsilon_{i }+ \varepsilon_j-\mu\right) 
}+1}=\frac{z x_{i} x_{j}}{1+z x_{i} x_{j}},
\end{equation}
where $z= e^{\beta \mu}$.  In Ref. \cite{PYNM2} it has been shown that the choice of $\rho (x) 
\sim x^{-\gamma}$ leads to a scale-free degree distribution with the same exponent, 
$-\gamma$. A similar measure, $\rho(x) \sim e^{\alpha x}$, was considered in sparse 
hypersoft configuration models \cite{KAKD,HPLG,VIHPKM}.

The alternative, yielding a power-law degree distribution in the low-temperature limit, can be made by choosing the density with exponential growth, $\rho(\varepsilon )  \propto e^{q\varepsilon}$. A similar consideration, as has been made in the case of the exponential decay of $\rho(\varepsilon)$, leads to $ k (\varepsilon) \propto e^{\beta \varepsilon }$. The parameter $k$ is bounded by $k_0 =k(\varepsilon_0)$. We obtain a power-law degree distribution $P(k)\propto k^{-\tau}$ with $\tau = q/\beta+1$. 

{\em Comments.} In contrast to Ref. \cite{CDAS}, where the temperature of a network was directly related to the exponent of its degree distribution by setting $T= \gamma -1$, we allow the temperature to be a free parameter. It can be calculated if the empirical data are known. Next, since the re-scaling of $\varepsilon$ can adjust $ T_c$, without loss of generality, one can take its value to be $T_c =1$.  Throughout the paper, we choose $T_c =1$, and in most of numerical simulations we take $\gamma = 2.1$. (Note the chosen value of $\gamma$ is typical for many real networks.)

 {Further, the networks with exponential growth (decay) of the density of states, $\rho (\varepsilon)$, we  will denote as {\em Type A} ({\em Type B}), respectively.} Imposing the standard normalization condition, $\int_0^{\varepsilon_0 }\rho (\varepsilon) d \varepsilon =1$, we obtain 
\begin{align}\label{Eq5a}
    \rho_g(\varepsilon) =&\frac{\alpha\beta e^{\alpha \beta (\varepsilon- 
\varepsilon_0/2)}}{2\sinh(a\beta\varepsilon_0/2)},  \quad 
    \rm Type \, A \\
    \rho_d(\varepsilon) =&\frac{\alpha\beta e^{-a \beta (\varepsilon- 
\varepsilon_0/2)}}{2\sinh(a\beta\varepsilon_0/2)}, \quad \rm Type \, B
\label{Eq5b}
\end{align}
where  $\alpha = \beta_c (\gamma -1)/\beta$.

We are now able to calculate the expected vertex degree $ \bar k (\varepsilon)$, expected 
number of links $ L $, and the Landau free energy $\Omega$. After some algebra we obtain 
rather long expressions:  
\begin{widetext}
\begin{itemize}
	\item Type A
	\begin{align}\label{Eq8}	
& \bar k_g(\varepsilon) =  \frac{N-1}{2\sinh(\alpha\beta\varepsilon_0/2)} \Big ( e^{\alpha\beta 
\varepsilon_0/2 }{}_{2}F_{1} \big (1, \alpha ;  1+\alpha  ;-e^{ \beta (\varepsilon + \varepsilon_0 
- \mu)} \big )  -   e^{-\alpha \beta \varepsilon_0 /2}{}_{2}F_{1} \big (1, \alpha ;  1+\alpha  ;- 
e^{ \beta(\varepsilon -\mu)} \big ) \Big ) ,\\
&L_g=  \frac{N(N-1) }{8\sinh^2(\alpha\beta\varepsilon_0/2)} \Big ( e^{\alpha \beta \varepsilon_0 }{}_{3}F_{2} \big (1, \alpha, \alpha ; 1+\alpha , 1+\alpha  ;- e^{ \beta (2 \varepsilon_0 - \mu)} \big )  - 2 {}_{3}F_{2} \big (1, \alpha, \alpha ; 1+\alpha , 1+\alpha  ;- e^{ \beta (\varepsilon_0 - \mu)}\big) \nonumber \\
&	+ e^{ -\alpha\beta \varepsilon_0} {}_{3}F_{2}\big (1, \alpha, \alpha ; 1+\alpha , 1+\alpha  ;- 
e^{ -\beta \mu} \big)\Big ), \label{Eq8a}	\\
	&\Omega_g=   -\frac{1}{a \beta}L_g -  \frac{N(N-1) }{8 
	\beta\sinh^2(\alpha\beta\varepsilon_0/2)} \Big ( e^{\alpha \beta  
				\varepsilon_0}\ln \big (1 + e^{\beta (\mu- 2\varepsilon_0)} \big ) -2\ln \big (1 + e^{\beta (\mu- \varepsilon_0)} \big ) +e^{-\alpha \beta  
						\varepsilon_0} \ln \big (1 + e^{\beta \mu} \big ) \nonumber \\
		&+ e^{\alpha \beta\varepsilon_0} \Phi(-e^{\beta (2\varepsilon_0-\mu)},1,\alpha) \big ) -2 \Phi(-e^{\beta (\varepsilon_0-\mu )},1,\alpha) \big ) + e^{-\alpha \beta  \varepsilon_0}\Phi(-e^{-\beta \mu},1,\alpha) \Big) ,
				\label{Eq8b}	
	\end{align}
	\item Type B 
\begin{align}\label{Eq9a}	
& \bar k_d (\varepsilon)=  \frac{(N-1) \alpha e^{\beta(\mu - \varepsilon)  }  e^{\alpha\beta 
\varepsilon_0/2 } }{2(1 +\alpha)\sinh(\alpha\beta\varepsilon_0/2)} \Big ({}_{2}F_{1} \big(1, 
1+\alpha ;  2+\alpha  ;-e^{ \beta (\mu -\varepsilon )} \big)  -   e^{-(1+\alpha\beta 
)\varepsilon_0 }{}_{2}F_{1} \big(1, 1+\alpha ;  2 +\alpha  ;- e^{ \beta( \mu - \varepsilon - 
\varepsilon_0 )} \big) \Big ) , \\
	&L_d = \frac{N(N-1)\alpha^2 e^{\beta( \mu - \varepsilon_0 )}}{8(1+\alpha)^2\sinh^2(\alpha\beta\varepsilon_0/2)} \Big ( e^{-(1+\alpha)  \beta \varepsilon_0}{}_{3}F_{2} \big(1, 1+\alpha, 1+\alpha ; 2+\alpha , 2+\alpha  ;- e^{\beta (\mu- 2\varepsilon_0 )} \big) \nonumber \\
	& - 2  {}_{3}F_{2} \big (1,1+ \alpha,1+ \alpha ; 2+\alpha , 2+\alpha  ;- e^{\beta (\mu- \varepsilon_0 )} \big ) 
	+e^{(1+\alpha)  \beta \varepsilon_0} {}_{3}F_{2} \big (1, 1+\alpha, 1+\alpha ; 2+\alpha , 2+\alpha  ;- e^{ \beta \mu} \big )\Big ).	\label{Eq9b}	 {OK!}  \\
	&\Omega_d=   \frac{1}{\alpha \beta}L_d -  \frac{N(N-1) }{8 
	\beta\sinh^2(\alpha\beta\varepsilon_0/2)} \Big ( e^{-\alpha \beta\varepsilon_0}
	\ln \big (1 + e^{\beta (\mu- 2\varepsilon_0)} \big ) 
	-2\ln \big (1 + e^{\beta (\mu- \varepsilon_0)} \big ) +e^{\alpha \beta \varepsilon_0} \ln 
		\big (1 + e^{\beta \mu} \big ) \nonumber \\
		&+ e^{-\alpha \beta\varepsilon_0} \Phi(-e^{\beta (2\varepsilon_0-\mu)},1,-\alpha \big ) -2 \Phi \big (-e^{\beta (\varepsilon_0-\mu )},1,-\alpha \big ) + e^{\alpha\beta  \varepsilon_0}\Phi\big (-e^{-\beta \mu},1,-\alpha \big) \Big) ,
				\label{Eq9c}	
	\end{align}
\end{itemize}
\end{widetext}
where ${}_{p}F_{q}(a_1, \dots, a_p; b_1, \dots, b_q; z)$ is the generalized hypergeometric function, and $\Phi(z,a,b)$ denotes the Lerch transcendent \cite{AEWM,NIST}.  

 {\subsubsection*{Thermodynamic limit}

To study the thermodynamic limit, one should specify the model by determining the parameters $\varepsilon_0$ and $\mu$.  As one can see the graph density and the clustering coefficients do not depend explicitly on the network's size, $\varrho = \varrho (\mu, T, \varepsilon_0) $ and  $C_{1,2} =  C_{1,2}(\mu, T, \varepsilon_0) $. Since the chemical potential and temperature are intensive variables, the graph density and the correlation coefficient would depend on the system's size if the cut-off $\varepsilon_0$ depends on  $N$. Therefore it is instructive to analyze their asymptotics for large $\varepsilon_0$. Before proceeding, it's worth mentioning that, from the relation $k_0 = e^{\beta \varepsilon_0}$, it follows $k_0 \rightarrow \infty$ when $\varepsilon_0 \rightarrow \infty$. 

We  make use of the asymptotic properties of the generalized hypergeometric functions \cite{AEWM,NIST,abr} to obtain
\begin{align} \label{Eq9c}	
&\bar k_g(\varepsilon) =  \frac{\alpha e^{ \beta ( \mu - \varepsilon  )} }{\alpha -1} e^{ -\beta \varepsilon_0 } \big( 1+ {\mathcal O}(e^{- \beta  \varepsilon_0} )\big) , \\
&\bar k_d (\varepsilon)=  {}_{2}F_{1} \big(1, 1+\alpha ;  2+\alpha  ;-e^{ \beta (\mu -\varepsilon )}  \big ) + {\mathcal O}(e^{- \beta  \varepsilon_0} ), \label{Eq9d}	
\\
	&\varrho_s  = \frac{\alpha^2}{(\alpha -1 )^2}e^{- \beta (2 \varepsilon_0 - \mu)} \big( 1+ {\mathcal O}(e^{- \beta  \varepsilon_0} )\big), \\
	&\varrho_d =  \frac{\alpha ^2 e^{\beta  \mu}}{(1+\alpha )^2}\,
	 {}_{3}F_{2} \big (1, 1+\alpha , 1+\alpha  ; 2+\alpha  , 2+\alpha   ;- e^{ \beta \mu}\big ) \nonumber \\
	 &+{\mathcal O}(e^{- \beta  \varepsilon_0} ).
\end{align}

The asymptotics of the global clustering coefficient can be found by substituting \eqref{Eq9c}	and \eqref{Eq9d} in Eqs. \eqref{C1a} and \eqref{C2b},
\begin{align} 
	&C_{1,2} \propto e^{- \beta  \varepsilon_0}, \quad \text {A-graph}, \\
		&C_{1,2} = \tilde C_{1,2}+ {\mathcal O}(e^{- \beta  \varepsilon_0} ), \quad \text{B-graph},
\end{align}
where $ \tilde C_{1,2} = \lim_{\varepsilon_0 \rightarrow \infty} C_{1,2}(\mu,T,\varepsilon_0)$.
It follows that in the limit of $\varepsilon_0 \rightarrow \infty$, the A-graph becomes sparse, and its properties do not depend on the temperature. However, for the B-graph, this is not true. We consider its features below in detail.
}

\subsection*{Fitness model}

As the first application of the developed approach, we consider the fitness model proposed in Ref. \cite{CGCA}. In this model for each vertex of the random network, a real non-negative hidden variable, called the {\em fitness}, is assigned. The fitness model belongs to the class of dense soft configuration 
networks described above, with $\varepsilon$ being the hidden variable, such that $\varepsilon_0 \rightarrow \infty$ and the chemical potential $\mu = \rm const$ \cite{CGCA, SVCG}. The linking probability is described by Eq. \eqref{EqP} with the density of states given by
\begin{align}\label{Eq11}
	\rho_f(\varepsilon) = \alpha\beta e^{-\alpha\beta \varepsilon }= \beta_c(\gamma - 1) e^{-\beta_c (\gamma -1)\varepsilon }.
\end{align}

The computation of the expected degree of a vertex,  the expected value of links and the Landau free energy yields: 
\begin{widetext}
	\begin{align}\label{Eq12a}	
&\bar  k_f (\varepsilon) =  \frac{  \alpha (N-1)e^{ \beta (\mu -\varepsilon )}}{1 +\alpha } \,{}_{2}F_{1} \big(1, 1+\alpha  ;  2+\alpha   ;-e^{ \beta (\mu -\varepsilon )} \big) ,  \\
&L_f = \frac{N(N-1)\alpha ^2 e^{\beta  \mu}}{2(1+\alpha )^2}\,
	 {}_{3}F_{2} \big (1, 1+\alpha , 1+\alpha  ; 2+\alpha  , 2+\alpha   ;- e^{ \beta \mu}\big ),\\
	 	 &\Omega_f=   \frac{1}{\alpha  \beta}L_f -  \frac{N(N-1) }{2 \beta } \Big (  \ln 
		\big (1 + e^{\beta \mu} \big ) +\Phi\big (-e^{-\beta \mu},1,-\alpha  \big) \Big) .
				\label{Eq12b}	
\end{align}
\end{widetext}

\begin{figure}[tbh]
\includegraphics[width=0.9\linewidth]{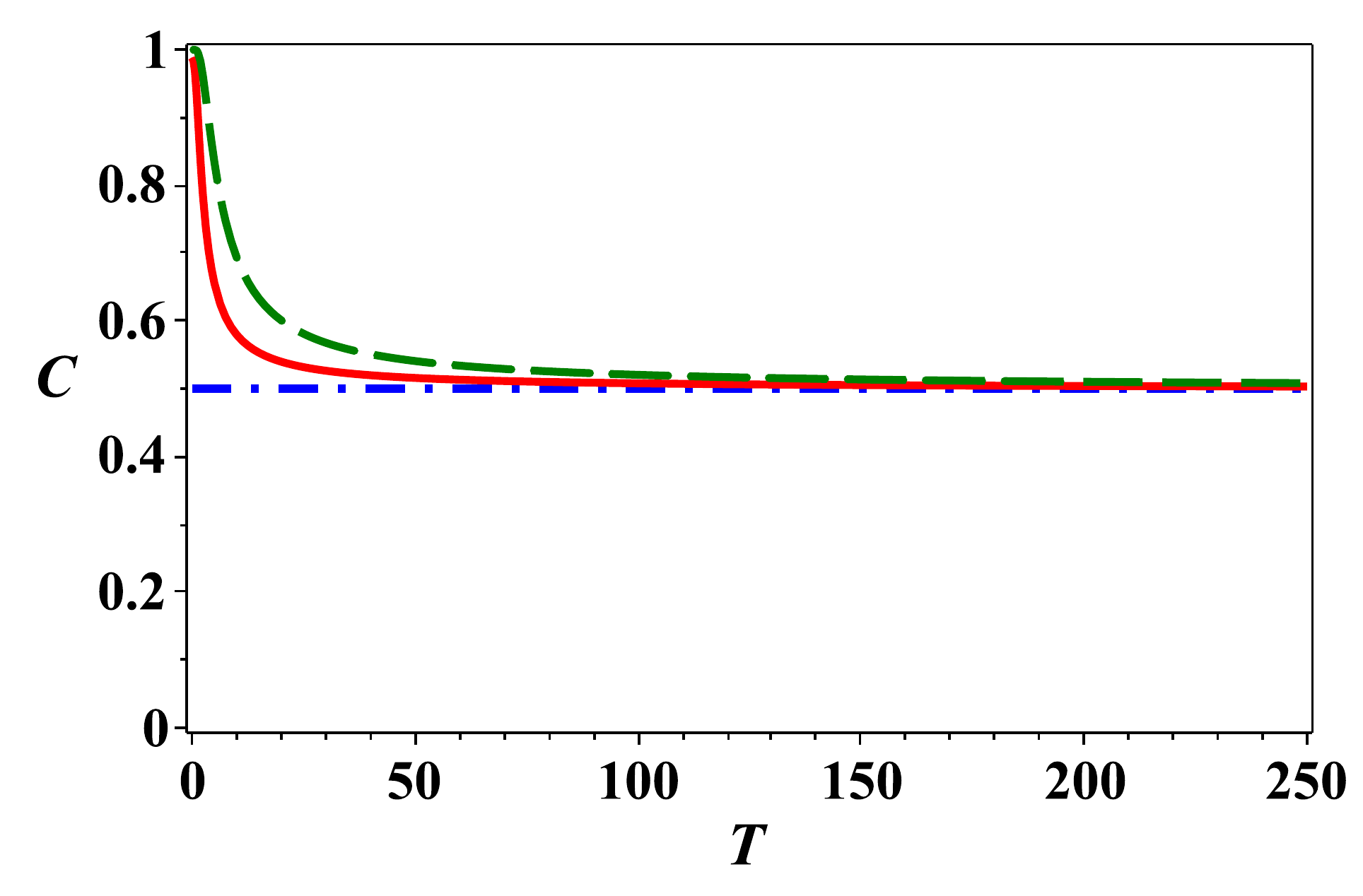} 
	\caption{Fitness model. The average clustering coefficient, $C$, in the whole network as a function of $T$ is depicted. Green dashed line: $\mu =10$, red solid line: $\mu =5$. Blue dash-dotted line presents the asymptotic value of the clustering coefficient.}
	\label{fig3b}
\end{figure}

\begin{figure}[tbh]
            \includegraphics[width=1\linewidth]{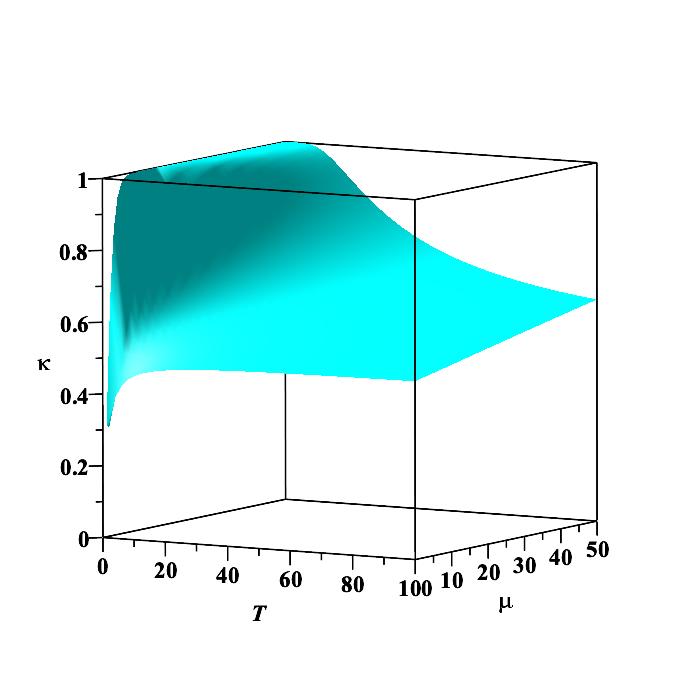} 
            	\caption{Fitness model. The average node degree per node, $ \kappa =\langle k \rangle/N$, as a function of  the chemical potential $\mu$ and temperature $T$.}
            	\label{fig3c}
            \end{figure}

Our numerical simulations show that in the fitness model both definitions of the clustering coefficient yield the same result; therefore, we omit indices, writing $ C $ instead of $ C_{1,2} $. In Figs \ref{fig3b},  \ref{fig3c} the  average clustering coefficient $C$, in the whole network and the average node degree per node, $ \kappa =\langle k \rangle/N$, are depicted.  As shown in Fig. \ref{fig3b}, the clustering coefficient $C\simeq 1$ for low temperatures,  $T \lesssim T_c$, and $C\simeq 1/2$ when $T \gg T_c$. The behavior of the average node degree is typical for Type B graphs: for low temperatures $\kappa \simeq 1$, and $\kappa \rightarrow 1/2$ in the limit of high temperatures, $T \gg T_c$ (Fig. \ref{fig3c}).  	

\section{Critical phenomena}

One of the important specific cases of configuration models is when $\varepsilon_0 = \mu$. Our choice of cut-off, $\varepsilon_0$, leads to the following modifications of Eqs. \eqref{Eq5a} -- \eqref{Eq9c}. The density of states for a Type A/Type B networks is given now by 
\begin{align}\label{Eq10h}
	\rho_g(\varepsilon) =&\frac{\alpha\beta e^{a \beta (\varepsilon- 
\mu/2)}}{2\sinh(a\beta\mu/2)},  \quad 
	\text{ Type A }\\
	\rho_d(\varepsilon) =&\frac{\alpha\beta e^{-\alpha \beta (\varepsilon- 
\mu/2)}}{2\sinh(a\beta\mu/2)}, \quad \text{ Type B }
\label{Eq10g}
\end{align}

Further we assume that the number of nodes $N \gg$1.  Then the expressions for the expected vertex degree, number of links and Landau free  energy are modified as follows:
	\begin{widetext}
\begin{itemize}
	\item Type A
	\begin{align}	\label{Eq10a}	
	 &\bar k_g(\varepsilon)=  \frac{N-1}{2\sinh(\alpha\beta\mu/2)} \Big ( e^{\alpha \beta \mu/2 }{}_{2}F_{1} \big (1, \alpha;  1+\alpha  ;-e^{ \beta \varepsilon } \big )  -   e^{-\alpha \beta \mu /2}{}_{2}F_{1} \big (1, \alpha ;  1+\alpha  ;- e^{ \beta(\varepsilon -\mu)} \big ) \Big ), \\
	&L_g=  \frac{N(N-1) }{8\sinh^2(\alpha\beta\mu/2)} \Big ( e^{\alpha \beta \mu} {}_{3}F_{2} \big (1, \alpha, \alpha ; 1+\alpha , 1+\alpha  ;- e^{ \beta \mu } \big )  - 2 {}_{3}F_{2} \big (1, \alpha, \alpha ; 1+\alpha , 1+\alpha  ;- 1\big ) \nonumber \\
&	+ e^{ -\alpha \beta \mu} {}_{3}F_{2} \big (1, \alpha, \alpha ; 1+\alpha , 1+\alpha  ;- e^{ -\beta \mu} \big )\Big ), 	\label{Eq10b}\\
	&\Omega_g=  - \frac{1}{\alpha \beta}L_g -  \frac{N(N-1)  }{8 
	\beta\sinh^2(\alpha\beta\mu/2)} \Big ( e^{\alpha \beta  
				\mu}\ln \big (1 + e^{-\beta \mu} \big ) -2\ln 2 +e^{-\alpha \beta\mu} \ln 
		\big (1 + e^{\beta \mu} \big ) + e^{\alpha \beta\mu} \Phi \big(-e^{\beta \mu},1,\alpha\big )\nonumber \\
		&-2 \Phi (-1,1,\alpha  ) + e^{-\alpha \beta  \mu}\Phi \big (-e^{-\beta \mu},1,\alpha \big)\Big) .
				\label{Eq10c}	
	\end{align}
	\item Type B 
\begin{align}\label{Eq11a}	
&\bar k_d (\varepsilon)=  \frac{\alpha (N -1)}{2(1 +\alpha)\sinh(\alpha\beta\mu/2)} \Big ( e^{\alpha\beta \mu/2 } {}_{2}F_{1} \big(1, 1+\alpha ;  2+\alpha  ;-e^{ \beta (\mu -\varepsilon )} \big)  -   e^{-\alpha\beta \mu /2}{}_{2}F_{1} \big(1, 1+\alpha ;  2 +\alpha  ;- e^{ -\beta \varepsilon } \big) \Big ),  \\
	&L_d = \frac{N(N-1) \alpha^2}{8(1+\alpha)^2\sinh^2(\alpha\beta\mu/2)} \Big ( e^{-(1+\alpha)  \beta \mu}{}_{3}F_{2}\big (1, 1+\alpha, 1+\alpha ; 2+\alpha , 2+\alpha  ;- e^{-\beta \mu} \big ) \nonumber \\
	& - 2 {}_{3}F_{2}(1,1+ \alpha,1+ \alpha ; 2+\alpha , 2+\alpha  ;- 1) 
	+ e^{(1+\alpha) \beta \mu}  {}_{3}F_{2} \big (1, 1+\alpha, 1+\alpha ; 2+\alpha , 2+\alpha  ;- e^{ \beta \mu}\big )\Big ),	\label{Eq11b} \\
		&\Omega_d=   \frac{1}{a \beta}L_d -  \frac{N(N-1) }{8 
		\beta\sinh^2(\alpha\beta\mu/2)} \Big ( e^{-\alpha\beta  
					\mu}\ln \big (1 + e^{-\beta \mu} \big ) -2\ln 2 +e^{\alpha \beta\mu} \ln \big (1 + e^{\beta \mu} \big ) + e^{-\alpha\beta\mu} \Phi \big(-e^{\beta \mu},1,-\alpha \big )\nonumber \\
			&-2 \Phi(-1,1,-\alpha)  + e^{a \beta  
					\mu}\Phi \big(-e^{-\beta \mu},1,-\alpha \big) \Big) .
					\label{Eq11c}	
		\end{align}
	\end{itemize}
\end{widetext}

\begin{figure}[tbh]
            \includegraphics[width=0.9\linewidth]{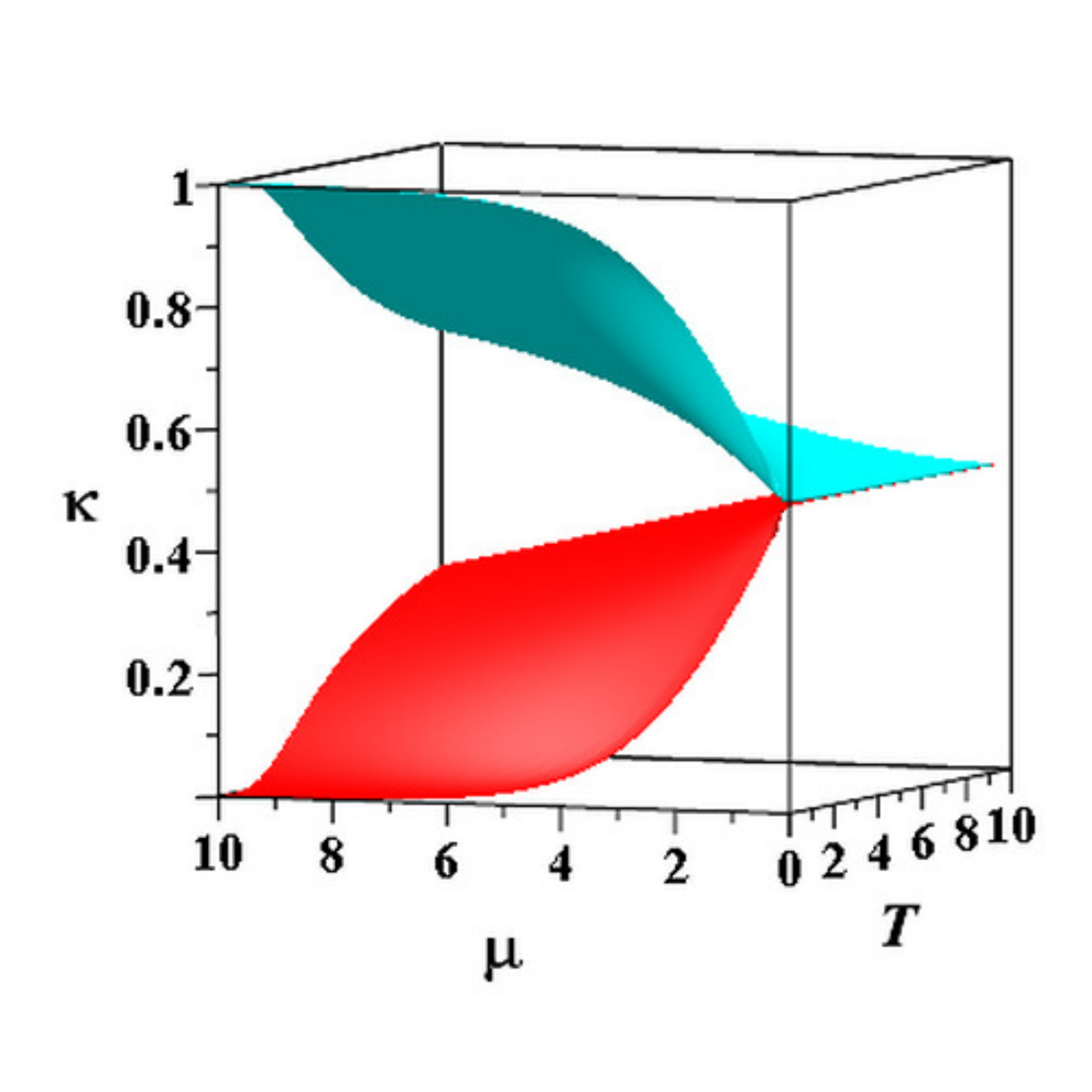} 
            	\caption{The average node degree per node, $ \kappa =\langle k \rangle/N$, as a function of chemical potential and temperature ($\gamma = 2.1$). Upper (cyan) surface depicts $ \kappa$ for the Type B graph. Lower (red) surface presents the results for the Type A graph.}
            	\label{fig3}
            \end{figure}     
            \begin{figure}[htb]
\includegraphics[width=0.9\linewidth]{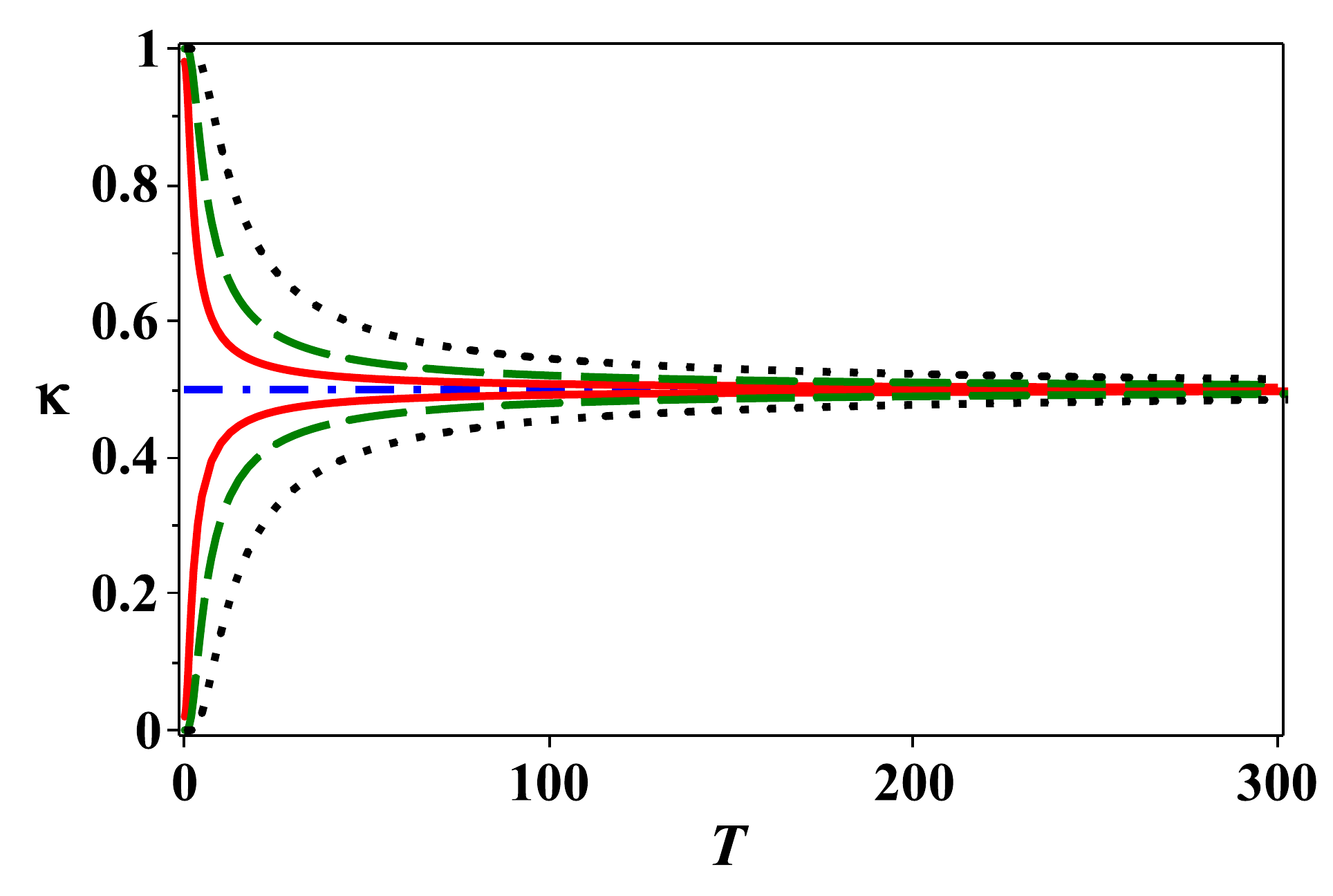} 
	\caption{ The average node degree per node, $ \kappa =\langle k \rangle/N$, as 
            	a function of temperature.  Upper curves: Type B graph. Lower curves: Type A graph. Black dotted line: $\mu =20$, green dashed line: $\mu =10$, red solid line: $\mu =5$. Blue dash-dotted line presents the asymptotic value of the clustering coefficient. }
	\label{fig1q}
\end{figure}

\begin{figure}[tbh]
\includegraphics[width=0.4925\linewidth]{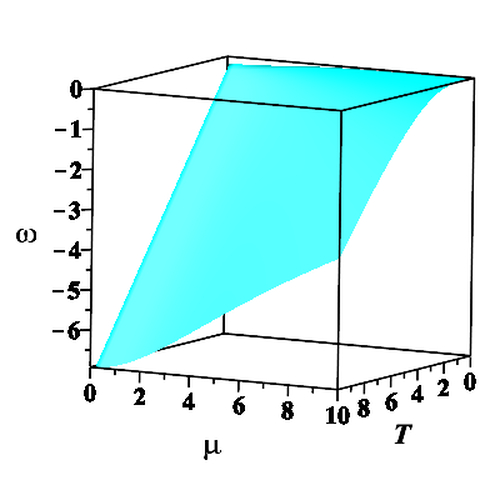} 
\includegraphics[width=0.4925\linewidth]{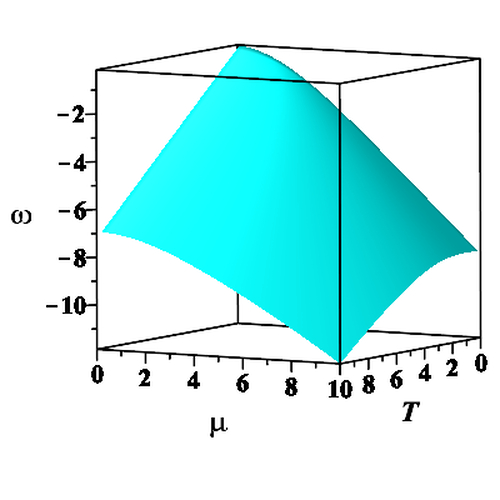} 
\caption{ Landau free energy per link as 
	a function of $T$ and $\mu$. Left: Type A graph. Right: Type B graph. }
	\label{fig1p}
\end{figure}
  \begin{figure}[tbh]
  \includegraphics[width=0.495\linewidth]{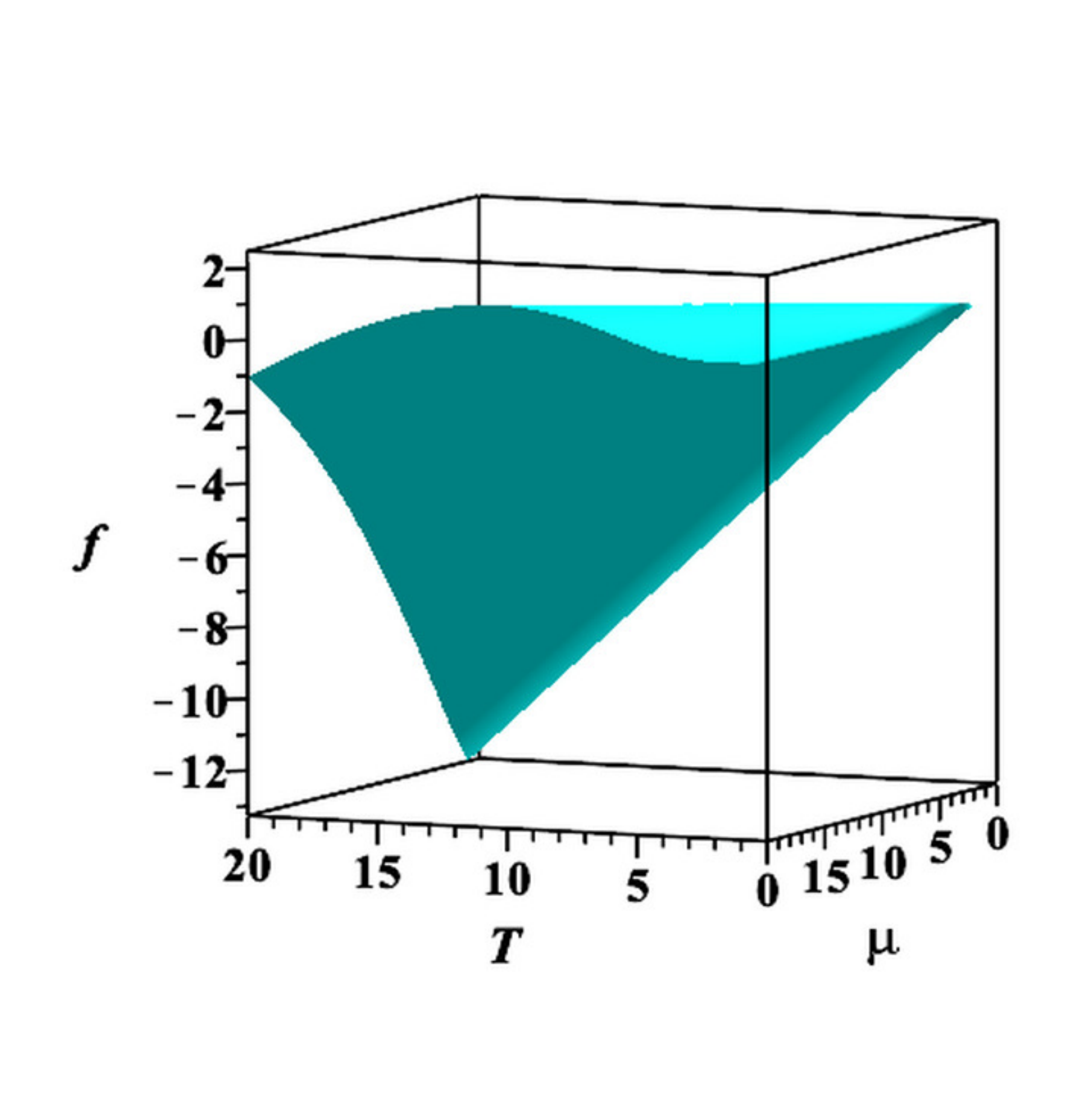} 
\includegraphics[width=0.49\linewidth]{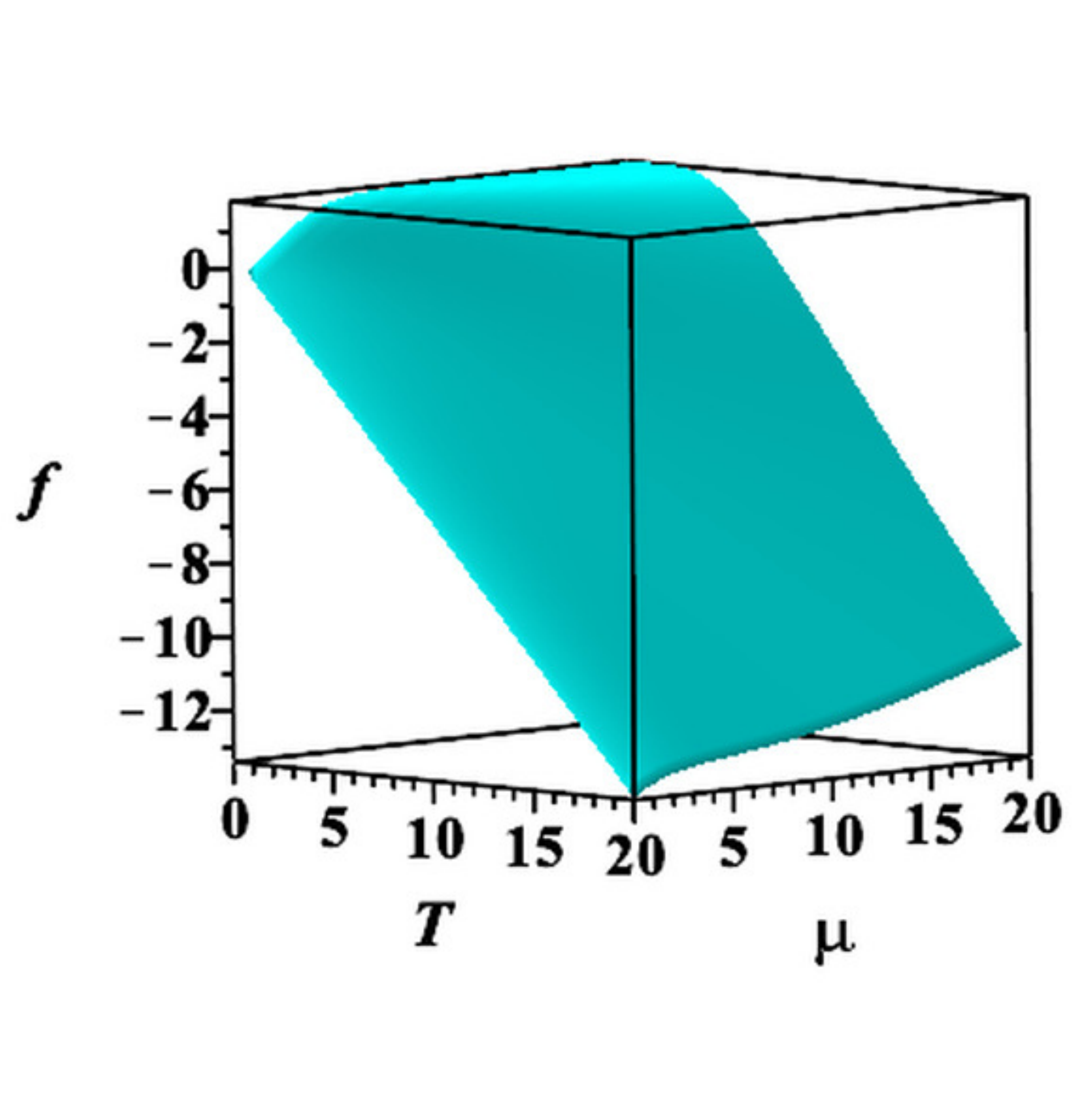} 
\caption{ Helmholtz free energy per link as 
a function of $T$ and $\mu$. Left: Type A graph. Right: Type B graph.}
	\label{fig1r}
\end{figure}

Numerical simulations presented in Figs. \ref{fig3}, \ref{fig1q}  validate our analytical predictions for the behavior of average node degree as a function of temperature for Type A and Type B graphs. As one can see, $\kappa \rightarrow 1/2$ for $T \gg \mu$. The behavior of the average node degree near zero temperature coincides for both graphs in the limit of $\mu \rightarrow 0$ yielding $\kappa \rightarrow 1/2$ (Fig. \ref{fig3}).  {Note, a {\em true sparse (dense) graph}, that implies $\langle k \rangle =0$ ($\langle k \rangle =1$) at $T =0$, only exists in the limit of $\mu \rightarrow \infty$ when $T \rightarrow 0$.}
  \begin{figure}[tbh]
\includegraphics[width=0.9\linewidth]{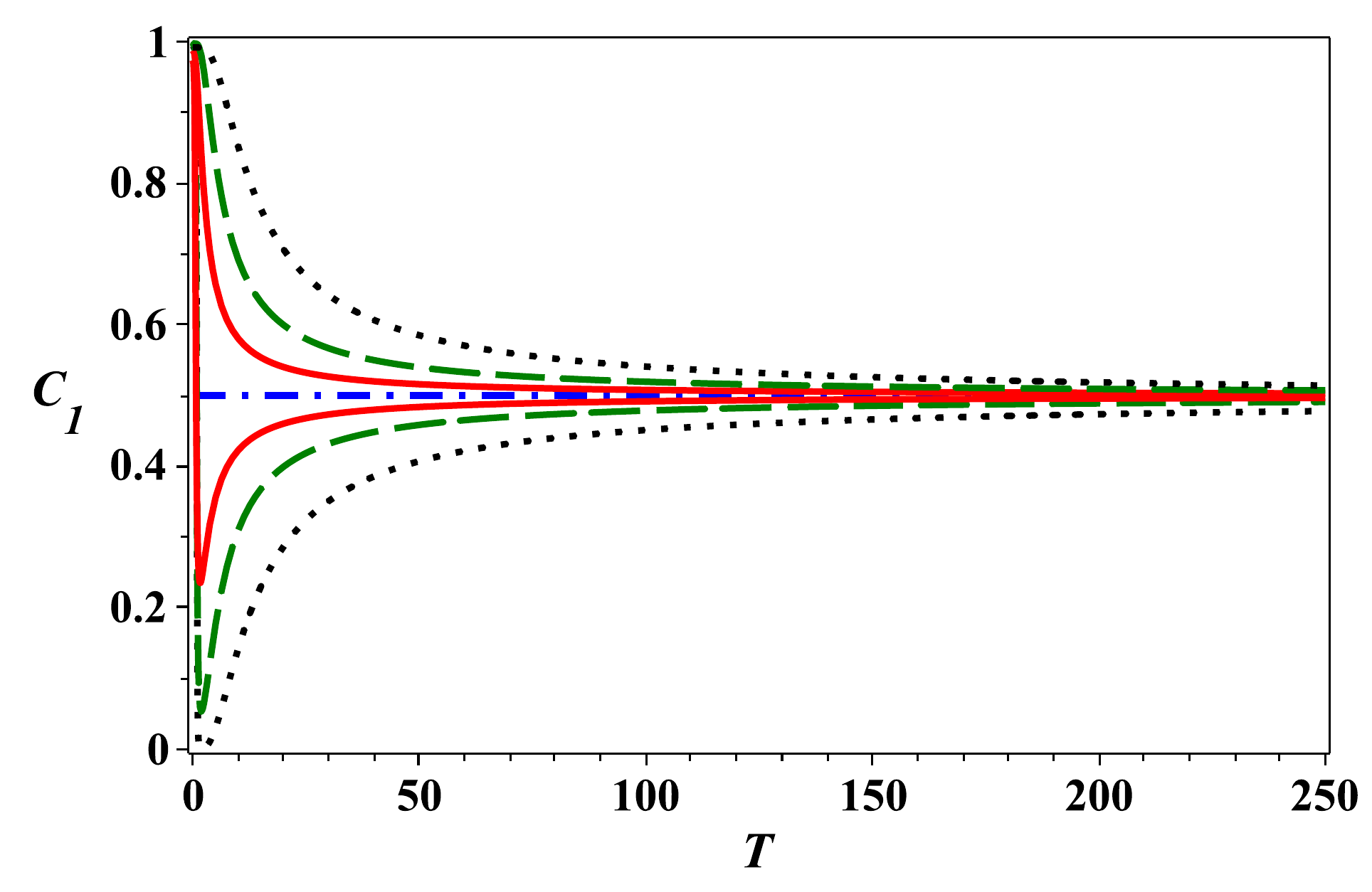} 
	\caption{ The average clustering coefficient, $C_1$, in the whole network as 
	a function of $T$. Upper curves: Type B graph. Lower curves: Type A graph. Black dotted line: $\mu =20$, green dashed line: $\mu =10$, red solid line: $\mu =5$. Blue dash-dotted line presents the asymptotic value of the clustering coefficient. }
	\label{fig3d}
\end{figure}
 \begin{figure}[tbh]
\includegraphics[width=0.9\linewidth]{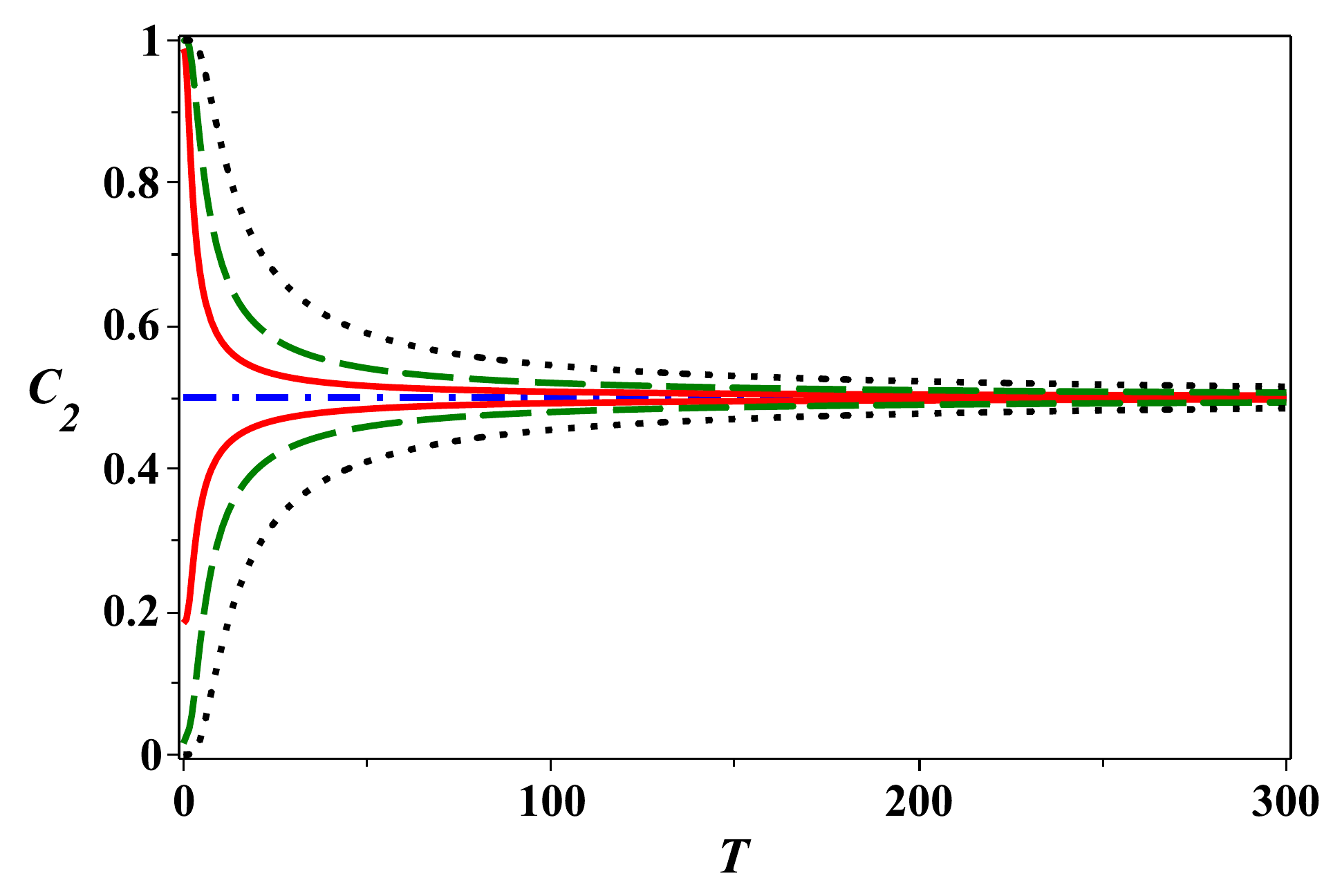} 
	\caption{ The average clustering coefficient, $C_2$, in the whole network as 
	a function of $T$. Upper curves: Type B graph. Lower curves: Type A graph. Black dotted line: $\mu =20$, green dashed line: $\mu =10$, red solid line: $\mu =5$. Blue dash-dotted line presents the asymptotic value of the clustering coefficient. }
	\label{fig3e}
\end{figure}

In Figs. \ref{fig1p}, \ref{fig1r}, the Landau and Helmholtz free energies are depicted as functions of temperature and chemical potential. For low temperatures and high values of the chemical potential, one can observe a slightly pronounced minimum in the behavior of the Helmholtz free energy and a flat Landau free energy for the Type A graph (Fig. \ref{fig1r}).

In Figs. \ref{fig3d}, \ref{fig3e} the average clustering coefficients in the whole network, $C_{1,2}$ are presented for different magnitudes of the chemical potential. For high temperatures, both clustering coefficients behave according to the theoretical 
predictions, $C_{1,2} \simeq1/2$. However, for low temperatures, the results are quite 
different. For Type B graphs, both definitions lead to the clustering coefficient's correct behavior in the limit of the low $T$-regime: $C \rightarrow 1$ when $T \rightarrow 0$. 

 {The low density of the graph characterizes type A in the limit of low temperatures and $\mu \gg 1$, i.e., $\langle k \rangle \ll 1$ as  $T\rightarrow 0$. Therefore, one expects the clustering coefficient to behave in the same way, $ C\ll 1$ as  $T\rightarrow 0$. While the first definition yields the wrong result, $ C\rightarrow 1$ when $T\rightarrow 0$, the behavior of $C_{2}$ is in agreement with the behavior of the average node degree (Fig. \ref{fig1q}).

 Our numerical simulations are confirmed by Appendix A's analytical results, where the comparison between the two measures is made. We show that in a low-temperature regime, the clustering coefficients behave as
\begin{align}
	 C_1&=	1-  e^{-\alpha\beta \mu/2}  +{\mathcal O}(\alpha ^2), \\
	 C_2 &\propto  e^{- \alpha\beta \mu/2} + {\mathcal O}(\alpha^2) .
\end{align}
 Based on these findings, we will use the definition of $C_{2}$ for the global clustering coefficient.}

\subsection{Phase transitions in asymptotically sparse networks}

 {We are interested in the asymptotically sparse network model with large number of nodes, $N \gg 1 $, and vanishing average node degree in the range of temperature $0 \leq T \leq T_c$. Since $N \gg 1 $, we will replace $N-1$ by $N$ in all formulas having this factor. We consider a particular model with the chemical potential defined as $\mu = T_c\ln(\nu N/\langle k \rangle) $, where $\nu$ is a temperature-independent parameter \cite{KDPF1,KDPF2,ANHM}. As follows from our previous analysis, the chemical potential should be infinite at $T=T_c$.

To determine $\nu $, we use the relation $L_g = \langle 
k  \rangle N/2$. After substitution of $\langle k \rangle=N\nu e^{- \beta_c\mu}$ in Eq. \eqref{Eq10b}, we obtain
\begin{widetext}
\begin{align}
	\nu= & \frac{e^{ \beta_c\mu}}{4\sinh^2(\alpha\beta\mu/2)} \Big ( e^{\alpha \beta \mu} {}_{3}F_{2} \big (1, \alpha, \alpha ; 1+\alpha , 1+\alpha  ;- e^{ \beta \mu } \big )  + 2 \alpha^2 \beta'(\alpha)	+ e^{ -\alpha \beta \mu} {}_{3}F_{2} \big (1, \alpha, \alpha ; 1+\alpha , 1+\alpha  ;- e^{ -\beta \mu} \big )\Big ), 
		\label{Eq10g}
\end{align}
	\end{widetext}
where we use the relation \cite{PBM3}
\begin{align}
{}_{3}F_{2}(1, \alpha, \alpha ; 1 +\alpha, 1+ \alpha ;-1) = - \alpha^2 \beta'(\alpha),
\end{align}
to replace the second term in Eq. \eqref{Eq10b}. Here
\begin{align}
	\beta(z)=\frac{1}{2}\left[\psi\left(\frac{z+1}{2}\right)- \psi\left(\frac{z}{2}\right)\right],
\end{align}
and $\psi (z)$ denotes the digamma function \cite{NIST}. 

We  make use of the asymptotic properties of the generalized hypergeometric functions \cite{AEWM,NIST,abr} to get
\begin{align}
	\nu = \bigg (\frac{\gamma -1}{\gamma -2}\bigg)^2 e^{(\beta_c - \beta)\mu}+ {\mathcal O} (e^{-(\gamma -2 )\beta_c \mu} ).
\end{align}
As seen, the  asymptotic series converges when $\gamma > 2$. Still supposing $\gamma  > 2$ and, in addition, assuming that $ \mu(T) \rightarrow \infty$ when $T \rightarrow T_c$, we obtain
\begin{align}
	\nu = \bigg (\frac{\gamma -1}{\gamma -2}\bigg)^2. 
\end{align}

For high temperatures, $T \gg T_c$, similar consideration yields
\begin{align}
	\nu e^{-\beta_c \mu} = - \alpha^2 \beta'(\alpha) +  {\mathcal O} (1 -e^{-\beta \mu} ).
\end{align}
Substituting  $\alpha = \beta_c (\gamma -1)/\beta$ and taking the limit of $T \rightarrow \infty$, we get
  \begin{align}
  \mu \rightarrow \mu_0 = \ln  \Big (\frac{2(\gamma -1)^2}{(\gamma -2)^2} \bigg )\quad {\rm and } 
  \quad \frac{\langle k \rangle}{N}  \rightarrow \frac{1}{2}.
  \end{align}
}
Now that we have obtained the constant $\nu$, we can find the dependence of the chemical potential on temperature for Type A and Type B networks employing Eqs. \eqref{Eq10b}, \eqref{Eq11b}. We obtain
\begin{align}
	\nu e^{- \beta_c\mu} = \frac{2L(N,T,\mu)}{N^2},
\end{align}
where $L(N,T,\mu)$ is the expected number of links. (Hereafter we omit the subindices $g/d$ in all calculations.)

Since an analytical solution of this equation does not exist, we solve it numerically. In Fig. \ref{fig2c} the chemical potential for Type A (blue curve) and Type B (red curve) networks is depicted. For the Type A graph we have $\mu(T) \rightarrow \infty$ as $T\rightarrow T_c +$. The Type B network's chemical potential is a continuous and bounded function of temperature and for both networks $\mu(T) \rightarrow \mu_0$ when $T\rightarrow \infty$.
\begin{figure}[tbh]
     \includegraphics[width=0.9\linewidth]{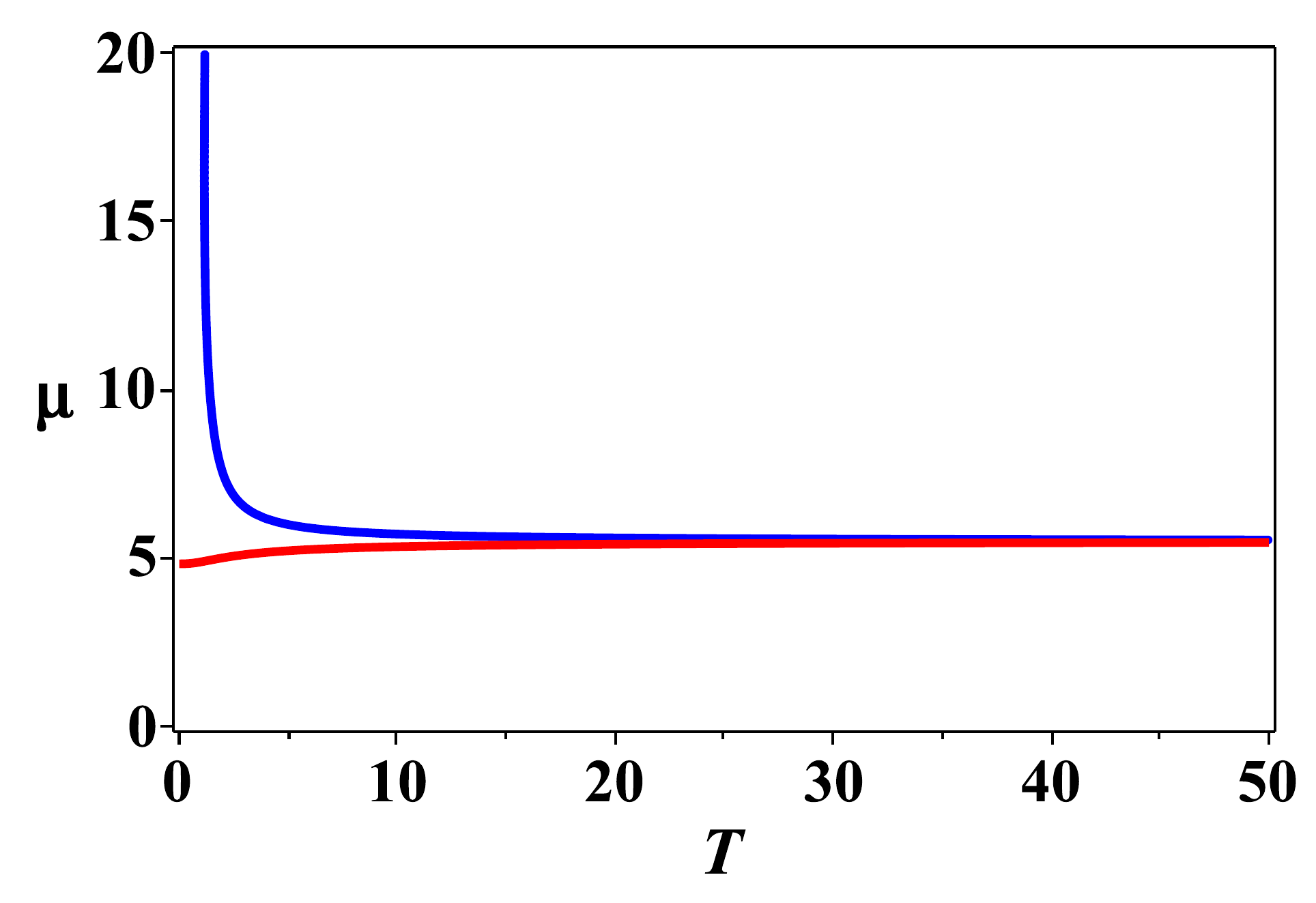} 
     	\caption{ Graph of the chemical potential $\mu$ as a function of temperature ($T_c =1$). Upper (blue) line depicts the behavior of the chemical potential for Type A graph.  Lower (red) line presents the Type B graph.}
     	\label{fig2c}
     \end{figure}
    
     \begin{figure}[tbh]
       \includegraphics[width=0.9\linewidth]{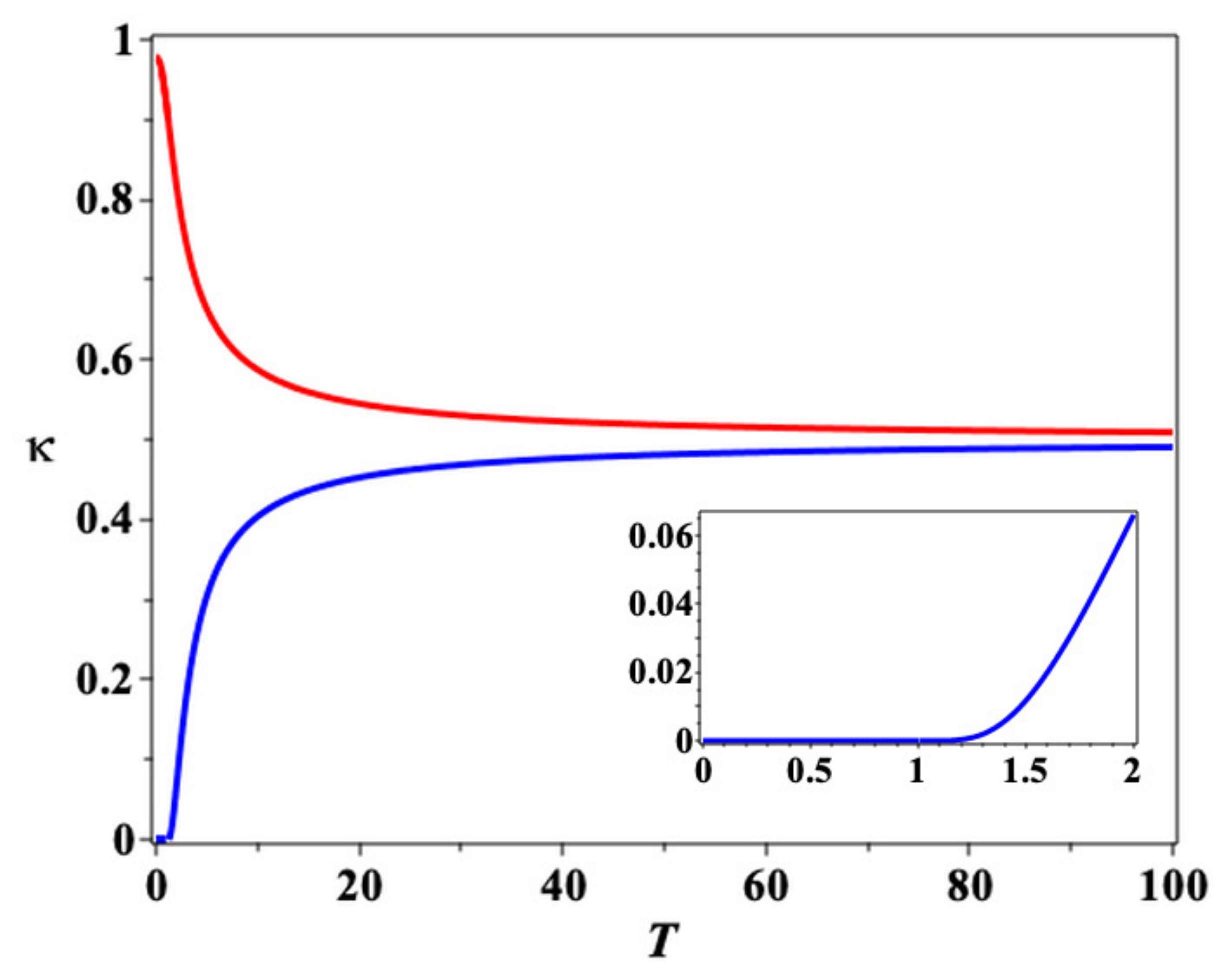} 
       	\caption{The average node degree per node, $ \kappa =\langle k \rangle/N$, as a function of temperature. Upper (red) line: Type B graph.  Lower (blue) curve: Type A graph. Inset: Zoom of the main plot for the Type A graph.}
       	\label{fig1a}
       \end{figure}
    
      \begin{figure}[tbh]
      \includegraphics[width=0.9\linewidth]{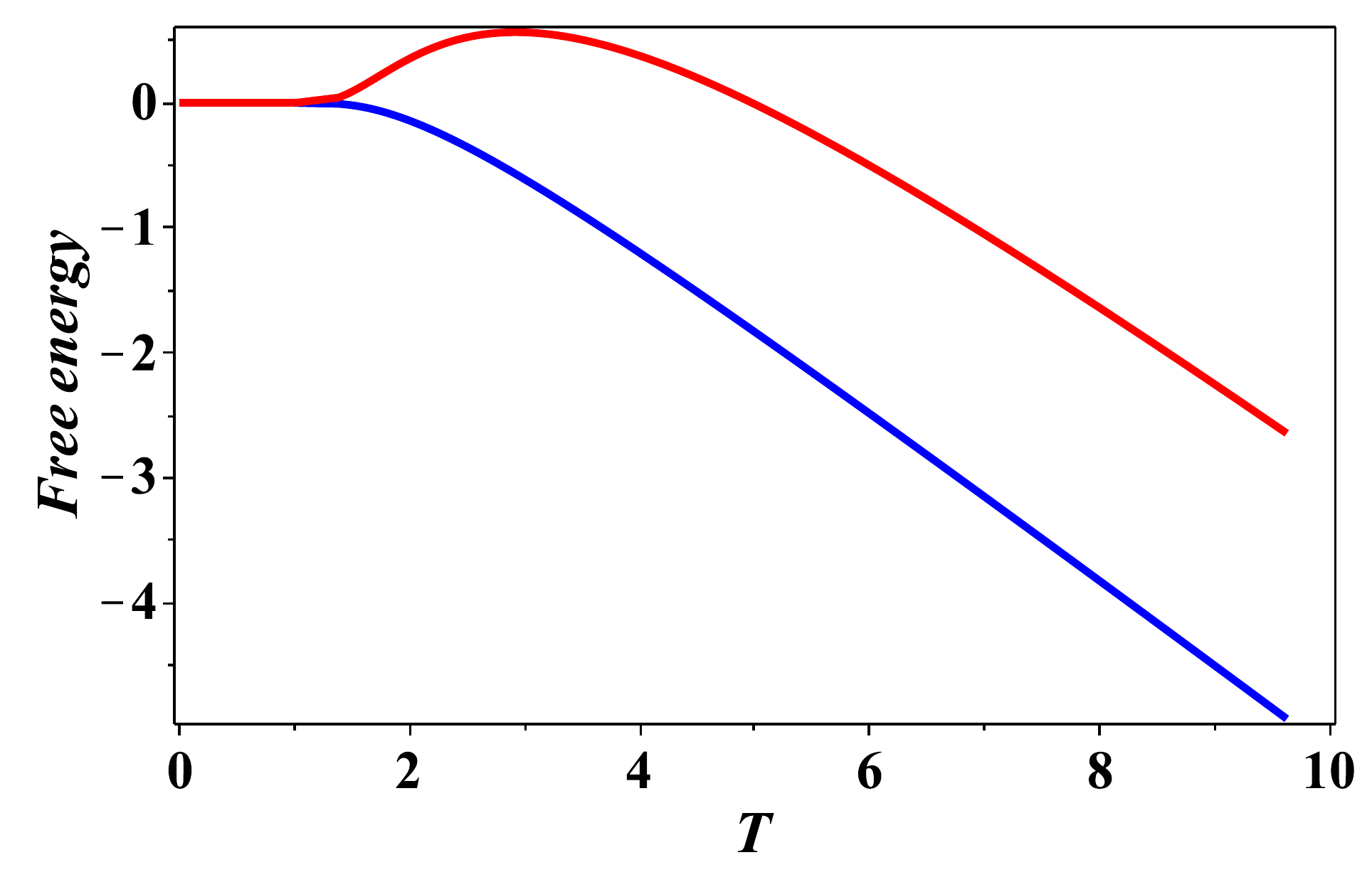} 
      	\caption{Type A graph. Free energy per link as a function of temperature. Upper (red) line depicts the Helmholtz free energy. Lower (blue) line presents the Landau free energy. }
      	\label{fig5}
      \end{figure}  
      
In Fig. \ref{fig1a} the average node degree per node, $ \kappa =\langle k \rangle/N$, is presented. For both graphs, $\kappa \rightarrow 1/2$ when $T \rightarrow \infty$, as expected. However, the average node degree's behavior for Type A and Type B graphs is highly different for low temperatures.  While $\kappa$ is a smooth function of temperature for the Type B graph, for the Type A graph, this is not true. At the point $T=T_c$, the system experiences a phase transition. Below the critical temperature, the graph is completely disconnected, $ \kappa =0$. In Fig. \ref{fig5} the Landau and Helmholtz free energies of Type A graphs are depicted. Both energies are continuous functions of temperature; however, they lost their analytical properties at the critical point . 

Near the critical temperature the chemical potential behaves as 
\begin{align}
	\mu \sim -\lambda\ln \tau, 
	\label{EqCP}
	\end{align}
where $\tau=(T-T_c)/T_c$ is the reduced temperature. The constant, $\lambda$, is calculated by  performing the numerical simulations (see SM for details). We obtain
\begin{align}
	\lambda = \left \{
	\begin{array} {ll}
		1/(\gamma - 2), &  2 < \gamma   <3 \\
		1, &  \gamma   \geq 3
	\end{array}
	\right .
	\label{EqL}
\end{align} 

To describe the phase transition, we introduce the order parameter, $\eta = 2 \langle k \rangle /N $, which ranges between zero and one.  We find that near the critical temperature the thermodynamic potentials behave as (see SM for details)
\begin{align}		
	\Omega & \approx -\frac{N^2}{4\nu\beta_c} \Big ( A - \tau \ln\Big ( \frac{\eta}{2\nu}\Big ) \Big ) \eta\label{LFEa}, \\
		F & \approx -\frac{N^2}{4\nu\beta_c} \Big ( A + (1-\tau) \ln\Big ( \frac{\eta}{2\nu}\Big ) \Big )\eta ,
	\label{HFEb} \\
	E & \approx -\frac{N^2}{4\nu\beta_c} \Big ( A + (2-\tau) \ln\Big ( \frac{\eta}{2\nu}\Big ) \Big )\eta ,
	\label{HFEc} \\
	S& \approx -\frac{N^2}{2\nu} \ln\Big ( \frac{\eta}{2\nu}\Big ) \eta  ,
	\label{EqS} \\
	C_N & \approx -\frac{N^2}{2\nu\beta_c}\Big (1 + \ln\Big ( \frac{\eta}{2\nu}\Big ) \Big )\frac{d\eta}{dT} ,
	\label{EqC}
	\end{align}
where
\begin{align}
	A= \frac{\gamma^2 -3\gamma +3}{(\gamma -2)^2}.
\end{align}

Next, substituting  $\eta = 2 \langle k \rangle /N $ in Eqs. \eqref{LFEa} -- \eqref{EqC}, we obtain the dependence of the thermodynamic potentials on the average node degree in the whole network 
\begin{align}		
	\Omega & \approx -\frac{N \langle k \rangle }{2\nu\beta_c} \Big ( A - \tau \ln\Big ( \frac{\langle k \rangle }{\nu N}\Big )  \Big ) \label{LFr}, \\
		F & \approx -\frac{N \langle k \rangle }{2\nu\beta_c}  \Big ( A + (1-\tau) \ln\Big ( \frac{\langle k \rangle }{\nu N}\Big )\Big ) ,
	\label{HFb} \\
	E & \approx -\frac{N \langle k \rangle }{2\nu\beta_c} \Big ( A + (2-\tau)\ln\Big ( \frac{\langle k \rangle }{\nu N}\Big ) \Big ) ,
	\label{HFd} \\
	S& \approx -\frac{N \langle k \rangle }{2\nu\beta_c} \ln\Big ( \frac{\langle k \rangle }{\nu N}\Big )  ,
	\label{EqSa} \\
	C_N & \approx -\frac{N^2}{2\nu\beta_c}\Big (1 + \ln\Big ( \frac{\langle k \rangle }{\nu N}\Big ) \Big )\frac{d \langle k \rangle }{dT} ,
	\label{EqCd}
	\end{align}
Finally, employing Eq. \eqref{EqCP} and keeping the leading terms in Eqs.\eqref{LFr} -- \eqref{EqCd}, one can obtain the dependence of thermodynamic functions on the reduced temperature. In Table I, we summarize our results.

\begin{center}
{Table I. Critical exponents for $0<\tau <1$} \\

\begin{tabular}{ll}
  \hline \hline
Thermodynamic functions\quad \quad \quad &  Relation \quad\\
\hline
Chemical potential & $\mu \propto- \lambda \ln \tau $   \\
Order parameter & $\eta \propto \tau^\lambda $    \\
Landau free energy  &$\Omega \propto -\tau^\lambda $ \\
Helmholtz  free energy  &$ F \propto -\tau^\lambda \ln \tau$    \\ 
Internal energy  &$ E \propto -\tau^\lambda \ln \tau$    \\             
Entropy & $ S \propto -\tau^{\lambda }\ln \tau$  \\
Heat capacity & $ C_N \propto -\tau^{\lambda -1}\ln \tau$  \\
  \hline \hline
\end{tabular}
\end{center}

 \subsection{Degree distribution}

We are now ready to analyze the topological properties of the network. First, we are interested in the degree distribution, $P_k$. To proceed, we use generating functions approach presented in Sec. II. The generating function for sparse networks is given by Eq. \eqref{Eq13p}, written as 
\begin{align}
G_0(z)    =\int^{\mu}_0e^{(z-1 ) \bar k (\varepsilon ) }   \rho (\varepsilon) d \varepsilon ,
\label{Eq14}    
\end{align}
where  $ \bar k (\varepsilon ) = N \int_0^\mu p(\varepsilon, \varepsilon' ) \rho(\varepsilon' ) d \varepsilon' $ is the expected degree of the node with the hidden variable $\varepsilon$ (see Eq. \eqref{Eq10a}). 

The computation of the degree distribution yields
\begin{align}
P_k =  \frac{1}{k!}\int^{\mu}_0  e^{-\bar k (\varepsilon ) } 
\big (\bar  k (\varepsilon )\big)^k  \rho (\varepsilon) d \varepsilon.
\label{Eq15}	
\end{align}
\begin{figure}[htb]
\includegraphics[width=0.9\linewidth]{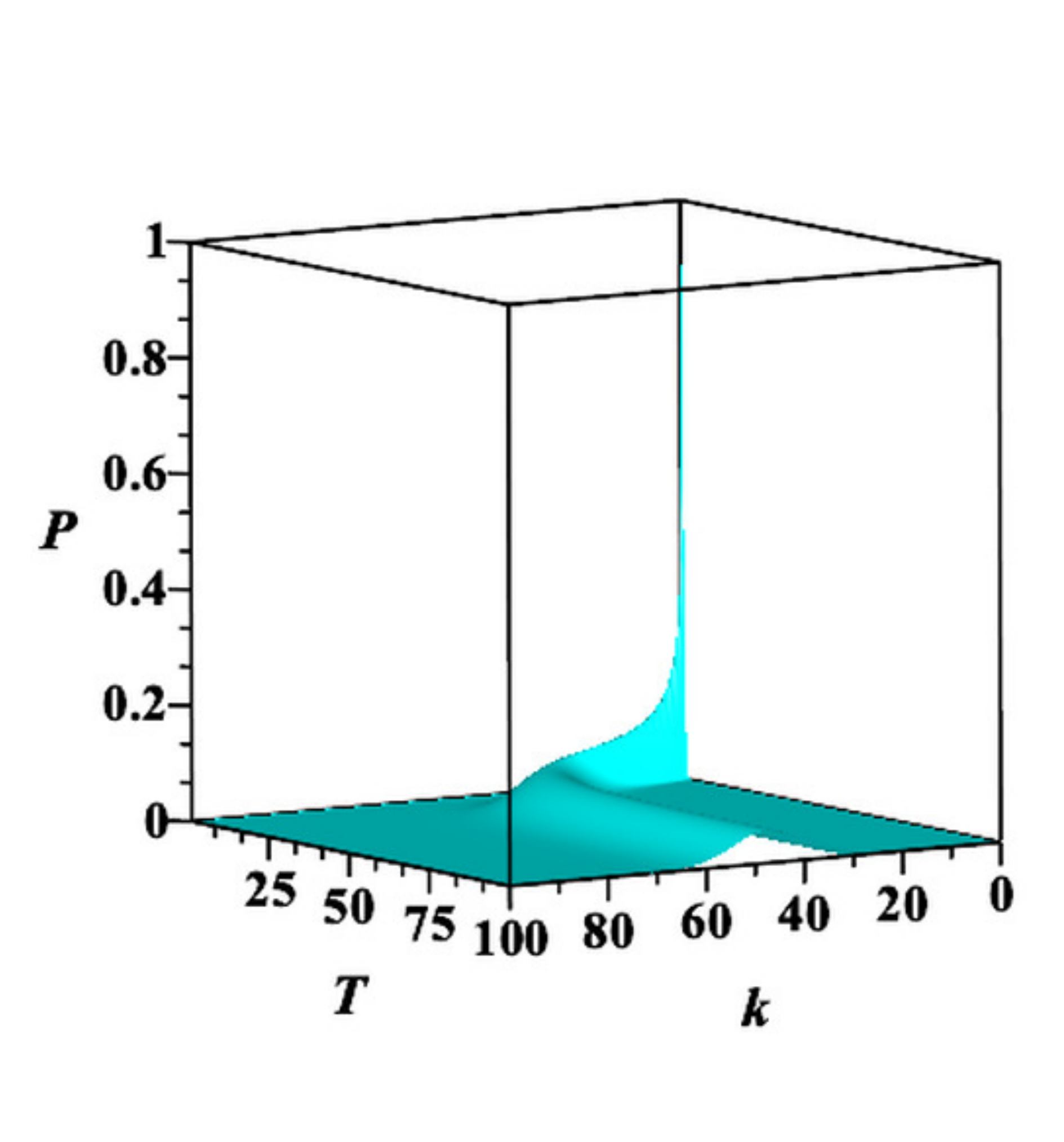} 
	\caption{Degree distribution as a function of $k$ and $T$. Number of nodes is $N=10^2$. }
	\label{fig6b}
\end{figure}

In Fig. \ref{fig6b} the degree distribution, $P_k$, is depicted as a function of $k$ and temperature. With increasing temperature, the dependence of  $P_k$  on $k$ is changed from power-degree, for lower temperatures, to a Poisson-like distribution for high temperatures. To prove this conjecture, below we will consider two limited cases, $T\simeq T_c$ and $T \gg T_c$, and derive the approximate formulas for their degree distribution.

\subsubsection*{Low temperatures}

 {Near the critical temperature, taking into account that $\beta \mu \gg 1$,  we obtain
\begin{align}		\label{Eq16a}	
&\rho (\varepsilon) \approx \alpha \beta e^{\alpha\beta (\varepsilon - \mu )}, \\
&	\bar  k(\varepsilon) \approx  N {}_{2}F_{1} \big (1, \alpha ;  1+\alpha  ;-e^{ \beta \varepsilon } \big ),\\
	&\langle k \rangle \approx N {}_{3}F_{2} \big (1, \alpha, \alpha ; 1+\alpha , 1+\alpha  ;- e^{ \beta \mu } \big ).
	 	 \label{Eq16b}	
\end{align}
As one can see, the main contribution in computation of integrals \eqref{Eq14},  \eqref{Eq15} yield high energies,  $\beta \varepsilon \gg 1 $. Using the asymptotic properties of the hypergeometric functions \cite{abr,NIST},  we obtain
\begin{align}			
&	\bar  k(\varepsilon) \approx N \xi e^{ -\beta \varepsilon } , \, \beta \varepsilon \gg 1 \\
	&\langle k \rangle 	\approx N \xi^2 e^{-\beta \mu},
\end{align}
where $ \xi  = {\alpha}/{(\alpha -1)}$. Note that at the critical point $ \xi(T_c) = \sqrt{\nu}$.
}

Using these results, we rewrite  Eq. \eqref{Eq14} as
\begin{align}
G_0(z)	=\alpha x_0^\alpha\int^1_{x_0}e^{-(1-z)\langle k \rangle x/(\xi  x_0) }   x^{-\alpha -1}dx,
\label{Eq18}	
\end{align}	 
 {where $x = e^{-\beta \varepsilon}$ and $x_0 = e^{-\beta \mu} $.}  Performing the integration, we obtain
\begin{align}
G_0(z)	=\alpha y^\alpha\big(\Gamma (-\alpha, y) - \Gamma (-\alpha, y/ x_0) \big ),
\label{Eq19}	
\end{align}	 
where $y = (1-z) \langle k \rangle /\xi $. To verify our results, we derive $G_0(1)$ and $G'_0(1)$. We find that $G_0(1)=1$ and $G'_0(1) =\langle k \rangle$, as expected.

In order to get the degree distribution, we use expression \eqref{Eq18} for the generating function. Taking the derivatives at the point $z=0$, we find
\begin{align}
	\frac{d^k}{dz^k}G_0(0)= \alpha \bigg(\frac{\langle k \rangle}{\xi} \bigg )^k x_0^{\alpha -k}\int^1_{x_0}e^{-\langle k \rangle x/(\xi  x_0) }   x^{k-\alpha -1}dx.
	\label{Eq19c}	
\end{align}
Performing the integration, we obtain
\begin{align}
P_k =& \alpha \bigg(\frac{\langle k \rangle}{\xi} \bigg )^\alpha\frac{\Gamma (k - \alpha,\langle k \rangle/\xi )}{k!} \nonumber \\
&-\alpha \bigg(\frac{\langle k \rangle}{\xi x_0} \bigg )^\alpha\frac{\Gamma (k - \alpha,\langle k \rangle /(\xi  x_0) }{k!} .
\label{Eq19b}	
\end{align}
where we have used Eq. \eqref{Eq13a} in computing the degree distribution.

Since $\langle k \rangle /(\xi  x_0) \gg 1$, one can neglect the last term and write
\begin{align}
P_k = \alpha \bigg(\frac{\langle k \rangle}{\xi} \bigg )^\alpha\frac{\Gamma (k - \alpha,\langle k \rangle/\xi )}{k!} \sim  k^{- \alpha -1}.
\label{Eq20a}	
\end{align}
Thus, near the critical point the degree distribution scales as $P_k \sim k^{-\tau}$, with $\tau = \alpha +1 =(T/T_c)(\gamma -1) +1 \approx \gamma$. At the critical point we have $\tau = \gamma$.

\subsubsection*{High temperatures} 

 Now let us consider the case of high temperatures,  $T \gg T_c$. Proceeding as above, we obtain
\begin{align}\label{Eq20a}	
\rho (\varepsilon) =&\frac{\alpha \beta e^{\alpha \beta(\varepsilon- 
\mu/2)}}{2\sinh(a \beta\mu/2)}, \\
 \bar k(\varepsilon) \approx &  \frac{N}{2} e^{ -\beta\varepsilon }, \\
 \langle k \rangle \approx  &  \frac{Ne^{ -\beta\mu/2} \sinh\big((\alpha - 1)\beta\mu/2 \big)  }{2 \sinh(\alpha \beta\mu/2)}.
	 	 	 \label{Eq20b}	
\end{align}
 {When $\beta \rightarrow 0$ we have $\langle k \rangle \rightarrow N/2$, as was predicted by the model.}

Now we can repeat the same procedure that we did in the case of of  low temperatures ($T \simeq T_c$).  Straightforward computation of the generating function yields
\begin{align}
G_0(z)	= e^{- (1-z)\langle k \rangle}.
\label{Eq21}	
\end{align}
In the calculation of this expression we have used the asymptotic properties of the incomplete gamma function $\Gamma(a,z)$ for large $a$ \cite{NIST}. This leads to a Poisson distribution for the degree distribution:
\begin{align}
P_k = \frac{e^{- \langle k \rangle}\langle k \rangle^k}{k!} .
\label{Eq22}	
\end{align}
Thus, we find that in the limit of high temperatures the degree distribution does not depend on $\gamma$, and the graph becomes a random graph. This is a universal behavior of scale-free networks. 
\begin{figure}[htb]
\includegraphics[width=0.9\linewidth]{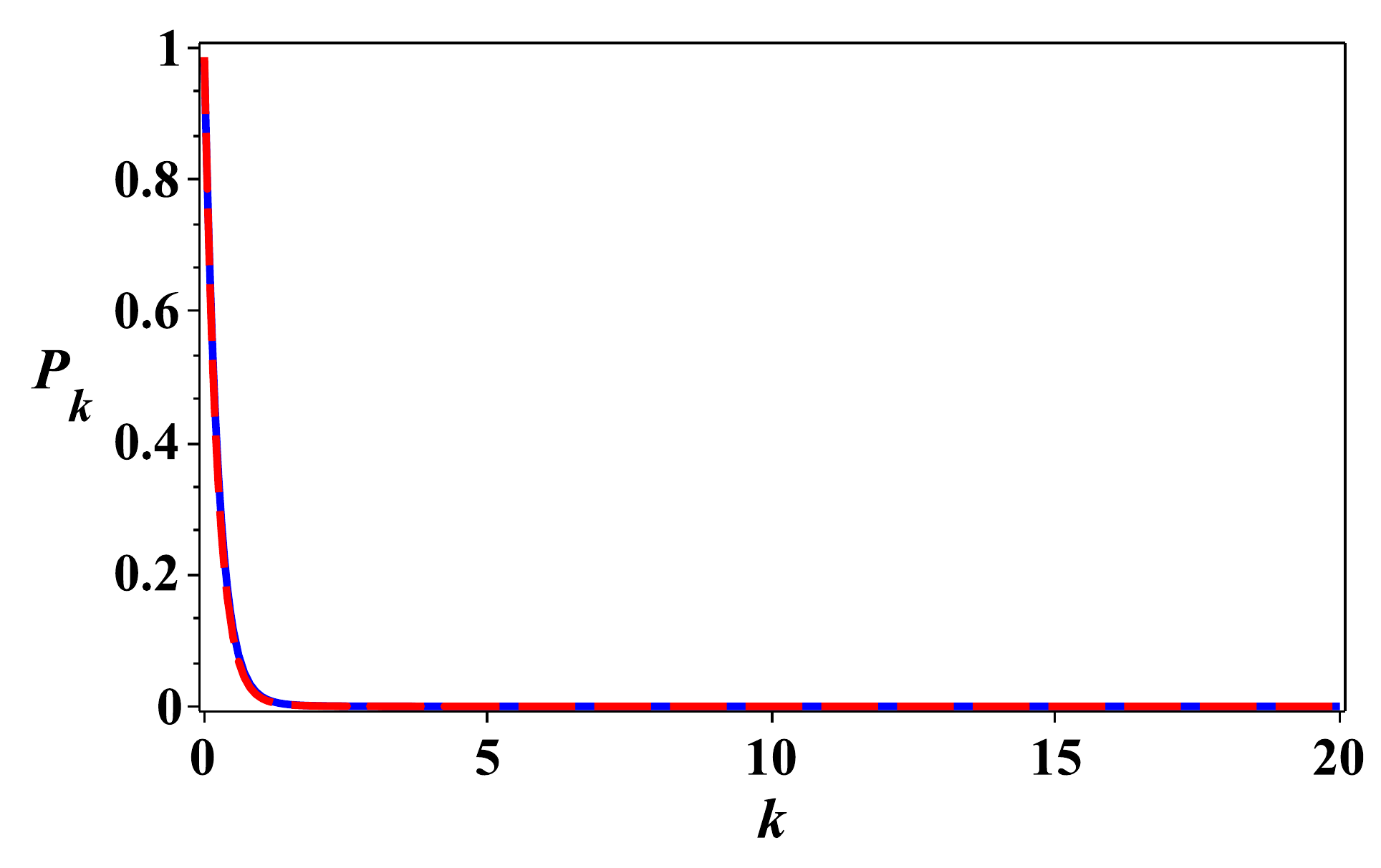} 
(a)\\
\includegraphics[width=0.9\linewidth]{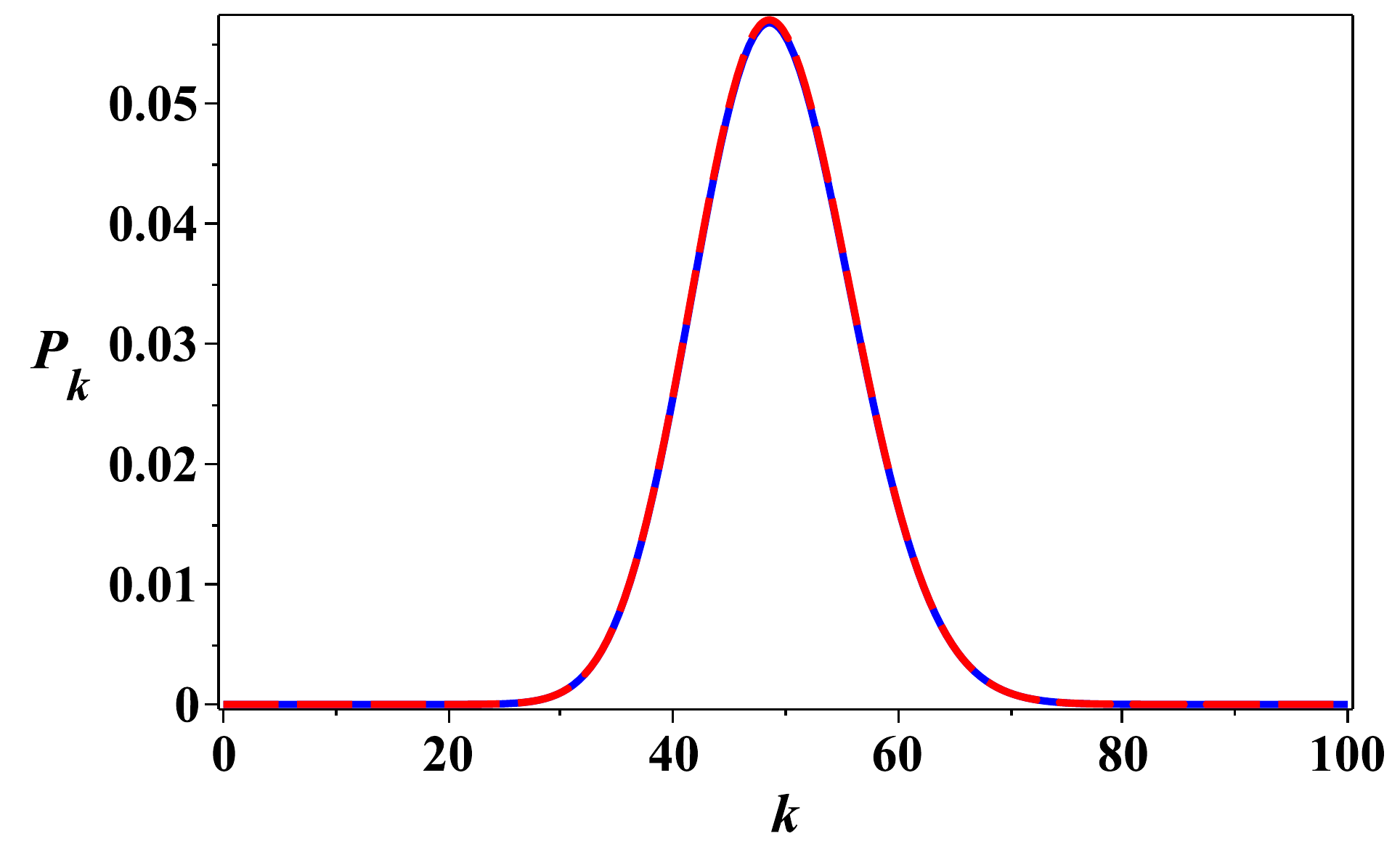} 
(b)
	\caption{Degree distribution $P_k$. Number of nodes is $N=10^2$.  Left: $T=	1.	1$. Right: $T= 100$. Blue curves depict $P_k$ defined by exact formula \eqref{Eq15}. Red dashed curves present the approximate expression \eqref{Eq20a} (a) and the Poisson distribution \eqref{Eq22} (b). }
	\label{fig6a}
\end{figure}

In Fig. \ref{fig6a} we compare the approximate (red dashed curves) and exact (solid blue curves) expressions for the degree distribution. One can see an excellent agreement between both the approximate and exact results for low and high temperatures.    
	
\subsection{Formation of a giant component}

Most real networks exhibit inhomogeneity in their link distribution leading to the natural clustering of the network into groups or communities. Within the same community, node-node connections are dense, but between groups, connections are less dense. A group of nodes forms a {\em component} when all of them are connected, directly or indirectly \cite{NMSW,GMNM}. 

A ``giant component" contains a significant part of the total number of nodes. For instance, if the degree connection $k=N$, the whole network is a giant component. In particular, this is valid for a Type B graph in the limit of $T \rightarrow 0$. In Sec. II, we have shown that, for both dense and Type A graphs, $k \rightarrow N/2$ in the limit of high temperatures. Thus, one can expect that with increasing temperature the dense network should be fragmented in a finite number of giant components, and in a sparse network, giant components should arise. In what follows, we will show that only one giant component arises in our model.

A giant component is formed in the network when the following condition holds \cite{GMNM}:
\begin{align}
	\langle k^2\rangle - 2\langle k\rangle \geq 0.
\label{Eq22}	
\end{align}
Using the relations $ \langle k \rangle = z_1$ and $ \langle k^2 \rangle = z_2  + z_1$, one can recast \eqref{Eq22} as
\begin{align}
z_2 - z_1       \geq 0.
\label{Eq23}
\end{align}

The computation of $z_n$ employing the generating function \eqref{Eq18} yields
\begin{align}
z_n \approx \frac{\gamma -1}{n+1 - \gamma }	\bigg (\frac{\langle k \rangle}{\sqrt{\nu}} \bigg )^n \Big ( e^{(n+1 -\gamma)\beta_c \mu}-1\Big ).
\end{align}
Using this in combination with Eq. \eqref{Eq23} leads to
\begin{align}
z_2 - z_1  \sim  \left \{
	\begin{array} {ll}
		\langle k \rangle, &  2 < \gamma   <3 \\
		\langle k \rangle  - k_0, &  \gamma   \geq 3
	\end{array}
	\right .
	\label{EqZ}
\end{align} 
where $k_0 =(\gamma -3)(\gamma -1) /(\gamma -2)^2$ is the threshold.

The structural phase transition, leading to percolation, results in the giant component emerging 
and occurs at the temperature $T_0$, where $z_2 - z_1=0$. Employing \eqref{EqZ} we find that 
the critical temperature $T_0 =T_c$, if $ 2 < \gamma < 3$, and $T_0 > T_c$ when  $\gamma 
>3$. We see that for networks with $2<\gamma < 3$, the transition associated with 
percolation still exists, though at a vanishing threshold this result was reported before in 
\cite{RAS}.).

Let $S$ be the fraction of the graph occupied by the
giant component. Then the size of the giant component is defined by \cite{NMSW}
\begin{align}
    S= 1-G_0(u),
    \label{Eq23c}
\end{align}
where $u$ is the smallest non-negative real solution of the equation
\begin{align}
    u = G_1(u) \equiv \frac{G'_0(u)}{\langle k\rangle}.    
    \label{Eq25}
\end{align}

Fig. \ref{fig9} shows the giant component's size for a network with $N=10^2$ and $N=10^4$ nodes. As one can see, the size of the giant component is a rapidly increasing function. The saturation occurs for comparatively low temperatures and values of the average node degree.
 \begin{figure}[htb]
\includegraphics[width=0.9\linewidth]{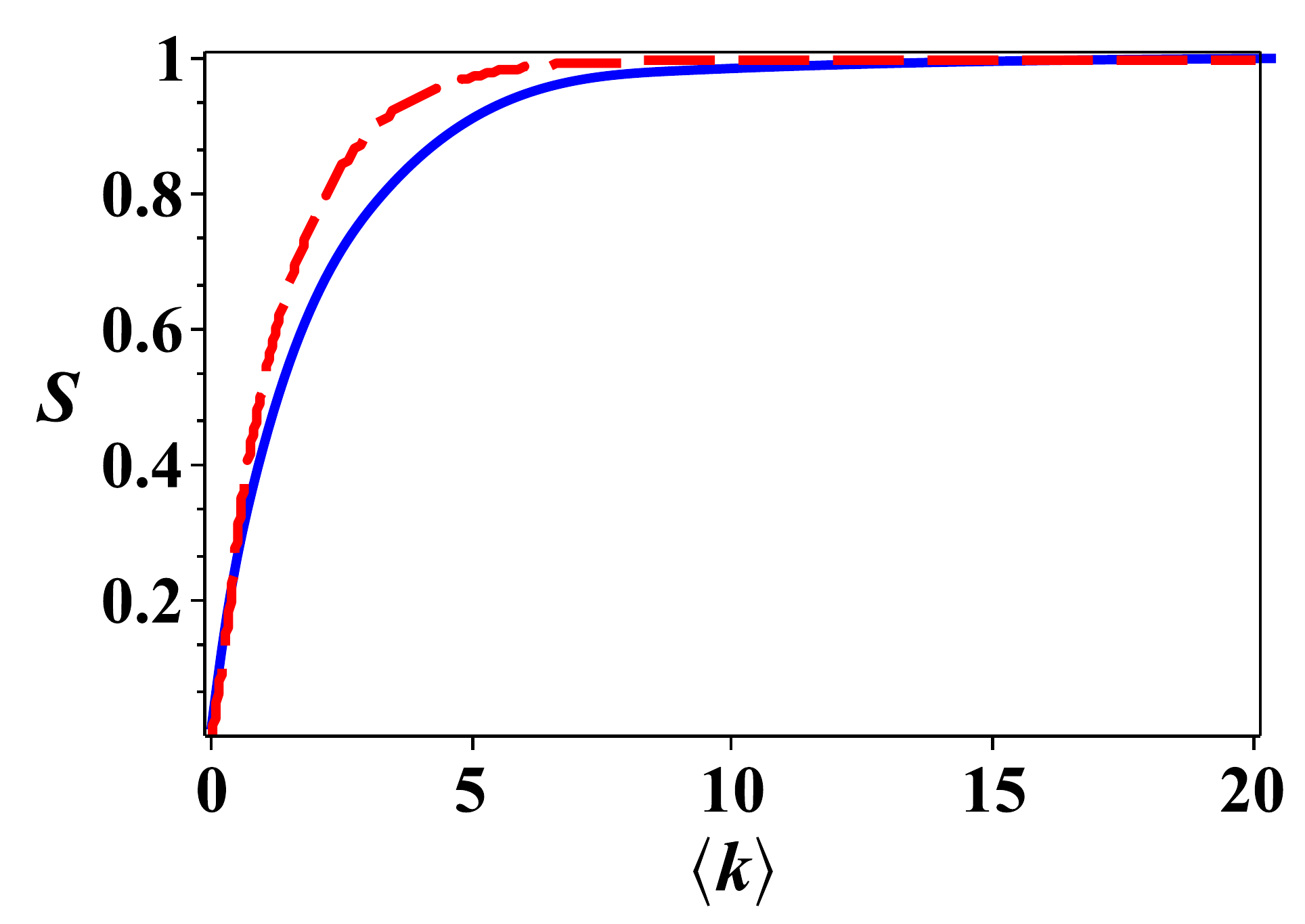} 
(a) \\
\includegraphics[width=0.9\linewidth]{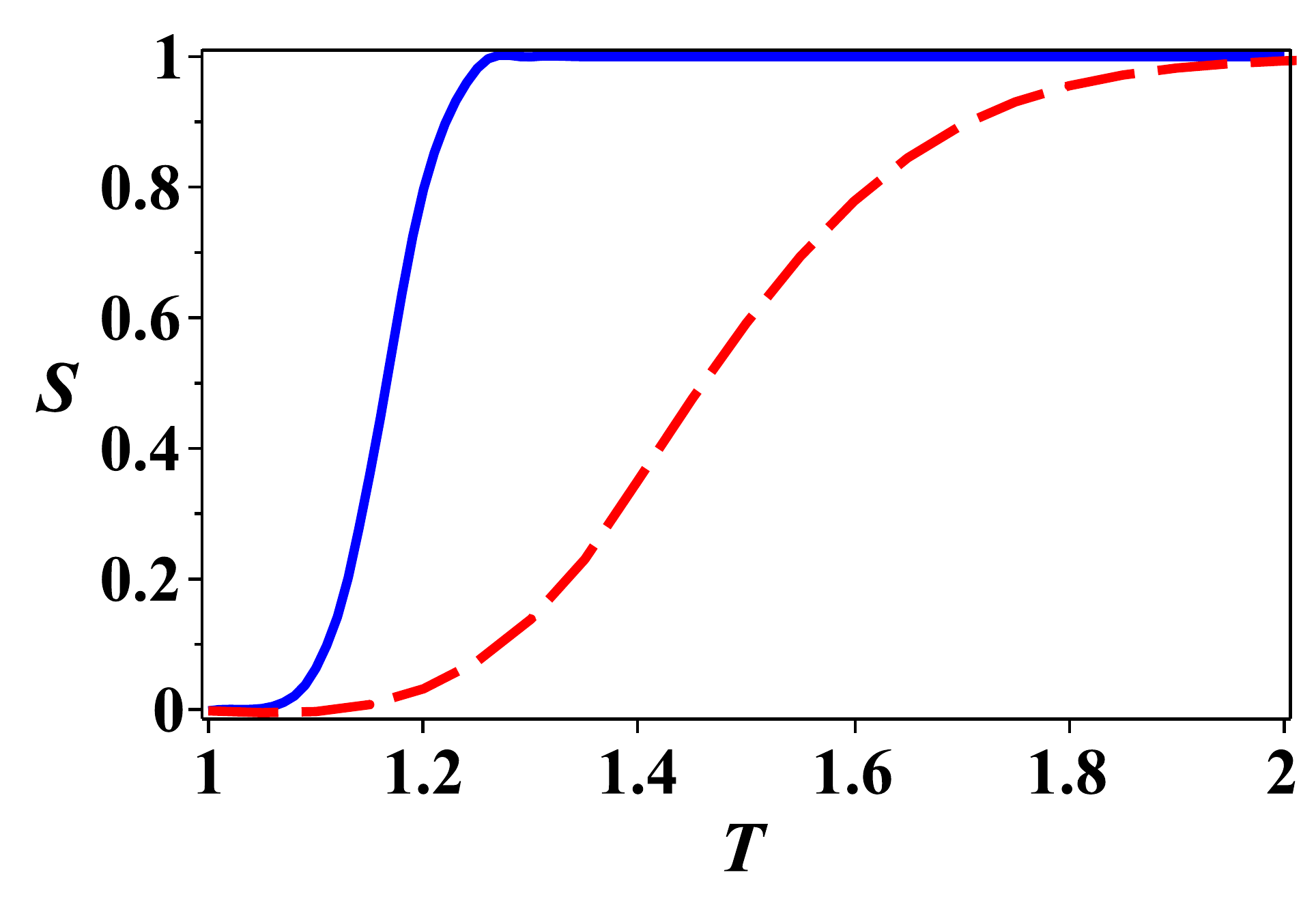} 
(b)
	\caption{ Size of a giant component as a function of average node degree (a), and temperature (b). Blue curves: $N=10^2$. Red dashed lines: $N=10^4$. }
	\label{fig9}
\end{figure}
Below we study two important cases in detail: low $(T \simeq T_c)$ and high temperature limits $(T \gg T_c)$. 

\subsubsection*{Low temperatures}

As has been shown above, for low temperatures the generating function can be approximated as
\begin{align}
G_0(z)	=\alpha y^\alpha\big(\Gamma (-\alpha, y) - \Gamma (-\alpha, y/ x_0) \big ),
\label{Eq19a}	
\end{align}	 
where $y = (1-z) \langle k \rangle /\xi $ and $x_0 =  \langle k \rangle /(N\xi^2) $.	
In the thermodynamic limit, one can neglect the last term and write
\begin{align}
G_0(z)	=\alpha y^\alpha \Gamma (-\alpha, y) .
\label{Eq19c}	
\end{align}	 
 
We now examine {Eq. \eqref {Eq25} in detail. Using the series expansion for the incomplete 
gamma function \cite{NIST},
 \begin{equation}
\Gamma(a, z)=\Gamma(a)-\sum_{k=0}^{\infty} \frac{(-1)^{k} z^{a+k}}{k !(a+k)},
\end{equation}
 we recast Eq.\eqref {Eq25} as
 \begin{align}
\epsilon = -\bigg ( \frac{\epsilon \langle k \rangle}{\xi}\bigg )^{\gamma- 2}\Gamma(3-\gamma)+(\gamma -2) \sum_{k=1}^{\infty} \frac{(-1)^{k} (\epsilon  \langle k \rangle)^{k}}{k !(k+2 - \gamma) \xi^k},
\label{Eq19e}
\end{align}
 where $\epsilon = 1-u$.
 
First we consider the case with $2 < \gamma <3$. Keeping only dominant terms as $\epsilon \rightarrow 0$, we obtain
 \begin{align}
 \epsilon 	\sim {\langle k \rangle}^{1/(3 - \gamma)-1}.
 \label{Eq19g}
 \end{align}
 For networks with $\gamma > 3$ we have a non-vanishing threshold, $k_0$, therefore it is convenient to introduce a new small parameter, $ \langle k \rangle - k_0 \ll 1$. Returning to Eq. \eqref{Eq19e} we find
\begin{align}
\frac{\xi}{k_0} (\langle k \rangle - k_0)=&(\gamma -3) \sum_{k=1}^{\infty} \frac{(-1)^{k} (  k_0)^{k} \epsilon^k}{(k +1)!(k+3 - \gamma) \xi^k}\nonumber \\
& -\bigg ( \frac{ k_0}{\xi}\bigg )^{\gamma- 2}\Gamma(3-\gamma) \epsilon^{\gamma -3}.
\label{Eq19f}
\end{align}
 Considering only the leading terms as $\epsilon \rightarrow 0$, we obtain
 \begin{align}
\epsilon  \sim  \left \{
	\begin{array} {ll}
		(\langle k \rangle - k_0)^{1/(\gamma -3)}, &  3 < \gamma   <4 \\
		1, &  \gamma   > 4
	\end{array}
	\right .
	\label{EqZ1}
\end{align} 

Returning  to the size of the giant component we find that $S \approx  \epsilon \langle k \rangle$. In conjunction with \eqref{Eq19g} and \eqref{EqZ1} this yields
 \begin{align}
S  \sim  \left \{
	\begin{array} {ll}
	 {\langle k \rangle}^{1/(3 - \gamma)},&  2 < \gamma   < 3  \\
		 (\langle k \rangle - k_0)^{1/(\gamma -3)}, &  3 < \gamma   <4 \\
		\langle k \rangle, &  \gamma   > 4
	\end{array}
	\right .
	\label{Eq33c} 
\end{align} 
 
To compare our results with those known for the percolation phase transitions, we use the relation
\begin{align}
z_2 - z_1 \sim q - q_c,	
\end{align}
where $q = 1-p$, $q_c$ is the threshold, and $p$ denotes the fraction of nodes (and their links)  removed from the network \cite{RPMR,RAS}. 

The size of the giant component near the critical point behaves as $S \sim ( q - q_c)^\delta$ \cite{RPMR}. Employing Eq. \eqref{EqZ} in combination with Eq. \eqref{Eq33c}, we find that the critical exponent, $\delta$, is given by
\begin{align}
\delta  =  \left \{
	\begin{array} {ll}
	{1/(3 - \gamma)},&  2 < \gamma   < 3  \\
		{1/(\gamma -3)}, &  3 < \gamma   <4 \\
		1, &  \gamma   > 4
	\end{array}
	\right .
	\label{Eq33k} 
\end{align} 
This is in agreement with the results reported in \cite{RAS}.

 \subsubsection*{High temperatures} 

 In this case, in order to to derive the size of the giant component, we employ Eq.\eqref{Eq21} for the generating function,
\begin{align} 
G_0(z)    = e^{- (1-z)\langle k \rangle}.
\label{Eq34}    
\end{align}    
Proceeding as above, we find that $ S $ satisfies the functional equation:
\begin{align}
1 -S =G_0(1-S)  = e^{-\langle k \rangle S}.
    \label{Eq35}
\end{align}
In this limit $S \simeq 1$, and the size of the giant component can be estimated as follows:
\begin{align}
     S\approx 1- e^{- \langle k \rangle}.
     \label{Eq36}
\end{align}
When $T \rightarrow \infty$ we obtain $S=e^{-N/2}$. Thus,  for $T \gg T_c$, almost all nodes 
of the network belong to the giant component. However, since $S < 1$, the giant component 
does not fill the entire graph. Moreover, only one giant component can be formed in the 
network, in agreement with the known results for graphs with purely power-law distributions \cite{NMSW}.

\section*{Conclusion} 

We demonstrated some diverse critical effects and phenomena occurring in networks, which significantly differ from those in lattices. We have shown how to treat random and scale-free networks within the conventional statistical physics approach and elucidated the role of network temperature. The temperature is a parameter that controls the average node degree in the whole network and \hmb{governs} the transition from unconnected to power-degree (scale-free) and random graphs. With increasing temperature, the degree distribution is changed from power-degree, for lower temperatures, to a Poisson-like distribution for high temperatures. Temperature can act \textit{contra} common sense in networks; for instance, increasing the sparse scale-free networks' temperature results in a high \hmb{connection degree}. For dense networks, the opposite is valid.

We introduced a configuration network model with hidden variables and found a finite-temperature phase transition for an asymptotically sparse network. The phase transition leads to fundamental structural changes in the network topology. The low-temperature phase yields a wholly disconnected graph. Above the critical temperature, the graph becomes connected, and a giant component appears. Near the critical temperature, the size of the giant component $S \ll 1$. This implies many vertices with a low degree and a small number with a high degree in the network. For high temperatures, almost all nodes of the network belong to the giant component. However, it turns out that the giant component does not fill the entire graph, even for the network's infinite temperature.

Our results suggest that a network temperature might be an inalienable property of real networks placing conditions on degree distribution, the topology of networks, and spreading information across these systems. We believe that our approach provides a unified statistical description of real networks' properties, from their scale-free and random graphs behavior to community structure and topology change.

\acknowledgements
The authors acknowledge the support by the CONACYT.

\appendix
 {
\section{Clustering coefficients}

In this Appendix, we obtain the analytic expressions for the global clustering coefficients $C_{1,2}$ for type A networks in the limit of low temperatures, $\beta \mu\gg \beta_c \mu \gg1$. We assume that the chemical potential $\mu = \text {const}$. The small parameters in the asymptotic expansions been used are: $ \alpha  = T \beta_c (\gamma -1 ) $ and $\epsilon = e^{-\beta \mu}$. As one can see $\epsilon \ll\alpha$, therefore, when it is applied, we keep the term $\mathcal{O}(\alpha)$
and omit $\mathcal{O}(\epsilon)$.

The  global clustering coefficients, $C_{1,2}$, are given by 
\begin{align}
   & C_1= \int d\varepsilon\rho(\varepsilon)  \frac{\iint  p(\varepsilon, \varepsilon'  ) p(\varepsilon', \varepsilon''  ) p(\varepsilon, \varepsilon''  )\rho(\varepsilon' )\rho(\varepsilon'')  d \varepsilon' d \varepsilon'' }{\big (\int p(\varepsilon, \varepsilon' ) \rho(\varepsilon' ) d \varepsilon' \big )^2  },
    \label{C1A} \\
   & C_2= \frac{\iiint  p(\varepsilon, \varepsilon'  ) p(\varepsilon', \varepsilon''  ) p(\varepsilon, \varepsilon''  )\rho(\varepsilon) \rho(\varepsilon' )\rho(\varepsilon'') d \varepsilon d \varepsilon' d \varepsilon'' }{\int \big (\int p(\varepsilon, \varepsilon' ) \rho(\varepsilon' ) d \varepsilon' \big )^2 \rho(\varepsilon )d \varepsilon }.
    \label{C2A}
\end{align}
 where $\rho (\varepsilon) =B\alpha \beta e^{\alpha\beta (\varepsilon - \mu/2 )}$, $B=  1/(2\sinh(\alpha\beta \mu /2))$, 
\begin{align}
    p(\varepsilon, \varepsilon') = \frac{1}{e^{\beta \left(\varepsilon + \varepsilon' -\mu\right) }+1},
    \label{EqPA}
\end{align}
and so on. 

It is convenient to introduce a new variable $x =e^{\alpha\beta (\varepsilon - \mu/2 )}$ and  recast \eqref{C1A} and \eqref{C2A}, as
\begin{widetext}
	\begin{align}
	C_1 &=  1-\int_\epsilon^{1/\epsilon} \frac{ I_1(x) }{ I_0(x) } \rho(x) d x,
	\label{EqAC1} \\
	 C_2&= \frac{1}{I}\iiint_\epsilon^{1/\epsilon}  p(x, x'  ) p(x', x''  ) p(x, x''  )\rho(x) \rho(x' )\rho(x'') d x d x' d x'' ,
	\label{EqAC2}
\end{align}
\end{widetext}
where $\epsilon = e^{-\beta\mu/2}$, $\rho (x) =B\alpha x^{\alpha -1}$ and
\begin{align}
    p(x,x') = \frac{1}{1 + x x'}.
 \end{align}
In Eqs. \eqref{EqAC1}  and \eqref{EqAC2}, we set
  \begin{align}   \label{EqA1a}
  & I=\int_\epsilon^{1/\epsilon} I_0(x)  \rho(x ) d x, \\
   & I_0(x)= \Big (\int_\epsilon^{1/\epsilon} p(x, x' ) \rho(x' ) d x' \Big )^2.
   \label{EqA1b}
   \end{align}
   \begin{align} 
  &  I_1(x)=    \int_\epsilon^{1/\epsilon} x' p(x, x'  ) Z(x, x'  ) \rho(x' ) d x', \label{EqA1c}\\
   &Z (x, x')  = \int_\epsilon^{1/\epsilon} x''  p(x', x''  ) p(x, x''  )\rho(x'')  d x''.
    \label{EqA1d}
    \end{align}

The computation of $I_0$ in the low-temperature limit yields $ I_0(x) = B^2 x^{-2\alpha} $. Next, substituting $I_0(x)$ in \eqref{EqA1a}, we obtain $I=B^2$. The computation of $I_0$ in the low-temperature limit yields $ I_0(x) = B^2 x^{-2\alpha} $. Next, substituting $I_0(x)$ in \eqref{EqA1a}, we obtain $I=B^2$. To calculate $Z$ we take into account that the leading term is obtained when $\epsilon =0$. Then taking the integral \eqref{EqA1d}, we obtain
\begin{align}
  Z(x,x')  = \alpha B\int^\infty_0 \frac{z^{\alpha} dz}{ (xz+1)(x' z +1)} = \frac{\alpha B\ln(x/x')}{x-x'} .
    \label{EqA2}
\end{align}
Substitution of $Z$ in \eqref{EqA1c}  yields
\begin{align}
  I_1 (x)=(\alpha B)^2 \int_0^\infty\frac{x'^{\alpha }\ln(x'/x) dx'}{(1+  xx') (x' - x)}.    
 \label{EqE1}
\end{align}
The computation of the integral  with help of the formula from the table of integrals \cite{PBM1},
\begin{equation}
\int_{0}^{\infty} \frac{\ln ^{n} x d x}{(x+a)(x+b)}=I_{n} \quad[\operatorname{Im} a=\operatorname{Im} b=0]
\label{EqI1}
\end{equation}
where
\begin{align} 	
I_{0}&=\frac{1}{a-b} \ln \left|\frac{a}{b}\right|, \\
I_{1}&=\frac{1}{2(a-b)}\left(\ln ^{2}|a|-\ln ^{2}|b|+A \operatorname{sgn} a\right) \nonumber \\
&[A=0 \, \text { if } a b>0 ;\, A=\pi^{2} \, \text {if}  \,a b<0],
 \end{align}
yields 
\begin{align}
  I_1 (x)=\frac{2(\alpha B)^2}{1+  x^2} \Big( \ln^2 x +\frac{\pi^2}{4} \Big ) + {\mathcal O}(\alpha ^2) .
 \label{EqE1}
\end{align}

\subsection*{Global clustering coefficient $C_1$}

In the low-temperature regime, the clustering coefficient's behavior $C_1$ is described by the following Proposition.

\begin{prop}
In the low-temperature limit, $\beta \mu \gg \beta_c \mu\gg 1$, the global clustering coefficient $C_1$ behaves as 
	 $$C_1=	1-  e^{-\alpha\beta \mu/2}  +{\mathcal O}(\alpha ^2).
$$  
where $\alpha  = T \beta_c (\gamma -1 ) $.
\end{prop}

\begin{proof}
The leading contribution in $C_1$ is determined by  
\begin{align}
C_1 &=  1-\int_0^{\infty} \frac{ I_1(x) }{ I_0(x) } \rho(x) d x.	
\label{EqC}
\end{align}
Substitution of $I_0(x)$ and  $I_1(x)$ yields
	\begin{align}
	C_1 =  1- e^{-\alpha\beta \mu/2}  A(\alpha) +{\mathcal O}(\alpha ^2)   ,
	\label{BC1b}
\end{align}
where
\begin{align}
A(\alpha) = 2\alpha^3 \int_{0}^{\infty}\frac{ d x x^{3\alpha -1} }{1+x^2}\Big( \ln^2 x +\frac{\pi^2}{4} \Big )  .
\end{align}
Making a change of variables, $x = \sqrt{y}$, we obtain
\begin{align}
A(\alpha) = \frac{\alpha^3 }{2}\int_{0}^{1}\frac{ d y( y^{3\alpha/2 -1} + y^{-3\alpha/2 -1 }) }{1+y}\Big( \ln^2 y +\frac{\pi^2}{4} \Big )  .
\end{align}

To proceed further we employ the formula from the table of integrals  \cite{PBM1}
\begin{align}
\int_{0}^{1}\frac{ x^{\alpha -1} \ln^n x \,dx}{1+x}  = \beta^{(n)}(\alpha),
\end{align}
where 
\begin{align}
	\beta(z)=\frac{1}{2}\left[\psi\left(\frac{z+1}{2}\right)- \psi\left(\frac{z}{2}\right)\right],
\end{align}
and $\psi (z)$ denotes the digamma function \cite{NIST}. The computation yields
\begin{align}
	A (\alpha) = 1+ \frac{\pi^2}{4} \alpha^2.
\end{align}
Substitution of this result in \eqref{BC1b} yields
\begin{align}
C_1=	1-   e^{-\alpha\beta \mu/2}  +{\mathcal O}(\alpha ^2).
\end{align}

\end{proof}

\subsection*{Global clustering coefficient $C_2$}

In the low temperature regime, the clustering coefficient's behavior $C_2$ is described by the following Proposition.

\begin{prop}
In the low-temperature limit, $\beta \mu \gg \beta_c \mu\gg 1$ the global clustering coefficient $C_2$ behaves as 
	 $$C_2 \propto  e^{- \alpha\beta \mu/2} + {\mathcal O}(\alpha^2) .$$  
where $\alpha  = (T/T_c )(\gamma -1 ) $.
\end{prop}

\begin{proof}
Substituting $\rho(x)$ and $I$ in Eq. \eqref{EqAC2}, we obtain
\begin{align}
C_2 = e^{- \alpha\beta \mu/2} Q(\alpha, \epsilon)	.
\end{align}
where $ Q(\alpha, \epsilon)= \alpha^3	\iiint_\epsilon^{1/\epsilon} F(x,x',x'') d x d x' d x''$ and 
\begin{align}
F(x,x',x'')=\frac{ x^{\alpha -1} x'^{\alpha -1} x''^{\alpha -1}}{ (1 + x x') (1 + x x'')(1 + x' x'')}.
\end{align}
The leading contribution in $C_2$ is determined by the term $Q_0 e^{- \alpha\beta \mu/2}$, where
\begin{align}
Q_0  = \lim_{\alpha\rightarrow 0 }  \alpha^3	\iiint_0^{\infty} F(x,x',x'') d x d x' d x''.
\label{EqC2b}
\end{align}
The limit exists, since $ F(x,x',x'') \sim  x^{\alpha -1} x'^{\alpha -1} x''^{\alpha -1}$
as $x \rightarrow 0$. Using this result, we obtain 
\begin{align}
 	C_2 =  Q_0e^{- \alpha\beta \mu/2}   + {\mathcal O}(\alpha^2).	 	
 	\end{align}
Thus, in low-temperature regime 
\begin{align}
C_2 \propto  e^{- \alpha\beta \mu/2} + {\mathcal O}(\alpha^2) .	
\end{align}

\end{proof}

}

\section{Supplementary Material}

\begin{widetext}

 In the Supplementary Material (SM), we calculate the critical exponent and find the logarithmic 
 corrections to the leading power laws that govern the order parameter and thermodynamic 
 potentials as a phase transition point is approached. First, we consider the behavior of the chemical potential near the critical temperature. The 
dependence of the chemical potential on temperature is defined by the equation, $ \langle k 
\rangle = 2L/N $,  where  $ L $ is the expected number of links.  Assuming $N >> 1$ and 
substituting $\langle k  \rangle = N\nu e^{-\beta_c \mu}$ in the equation,
\begin{align}
&L=  \frac{N(N-1) }{8\sinh^2(\alpha\beta\mu/2)} \Big ( e^{\alpha \beta \mu} {}_{3}F_{2} \big 
(1, \alpha, \alpha ; 1+\alpha , 1+\alpha  ;- e^{ \beta \mu } \big )  - 2 {}_{3}F_{2} \big (1, 
\alpha, \alpha ; 1+\alpha , 1+\alpha  ;- 1\big ) \nonumber \\
&	+ e^{ -\alpha \beta \mu} {}_{3}F_{2} \big (1, \alpha, \alpha ; 1+\alpha , 1+\alpha  ;- e^{ 
-\beta \mu} \big )\Big ),
	\end{align}
 we obtain
\begin{align}
 \nu \equiv\bigg (\frac{\gamma -1}{\gamma -2}  \bigg )^2 =  &\frac{e^{\beta_c 
 \mu}}{4\sinh^2(a\beta\mu/2)} \Big ( e^{a \beta \mu} {}_{3}F_{2} 
\big (1, a, a ; 
1+a , 1+a  ;- e^{ \beta \mu } \big )  - 2 {}_{3}F_{2} \big (1, a, a ; 1+a , 1+a  ;- 1\big ) 
\nonumber \\
&	+ e^{ -a \beta \mu} {}_{3}F_{2} \big (1, a, a ; 1+a , 1+a  ;- e^{ -\beta \mu} \big )\Big ). 	
\label{SM1a}
\end{align}
The solution of this equation yields the  chemical potential as a function of temperature. 
\begin{figure}[tbh]
      \includegraphics[width=0.24\linewidth]{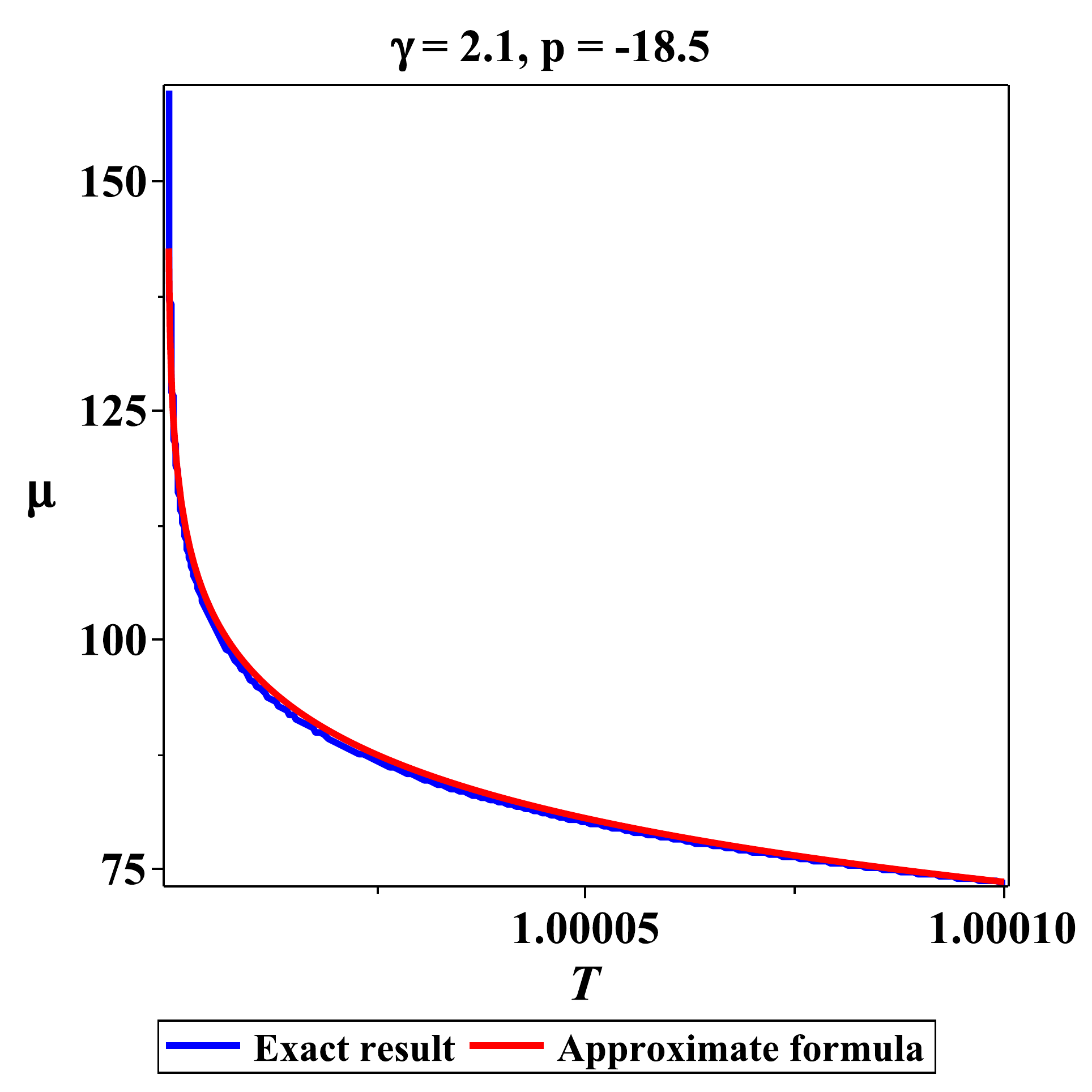} 
       \includegraphics[width=0.24\linewidth]{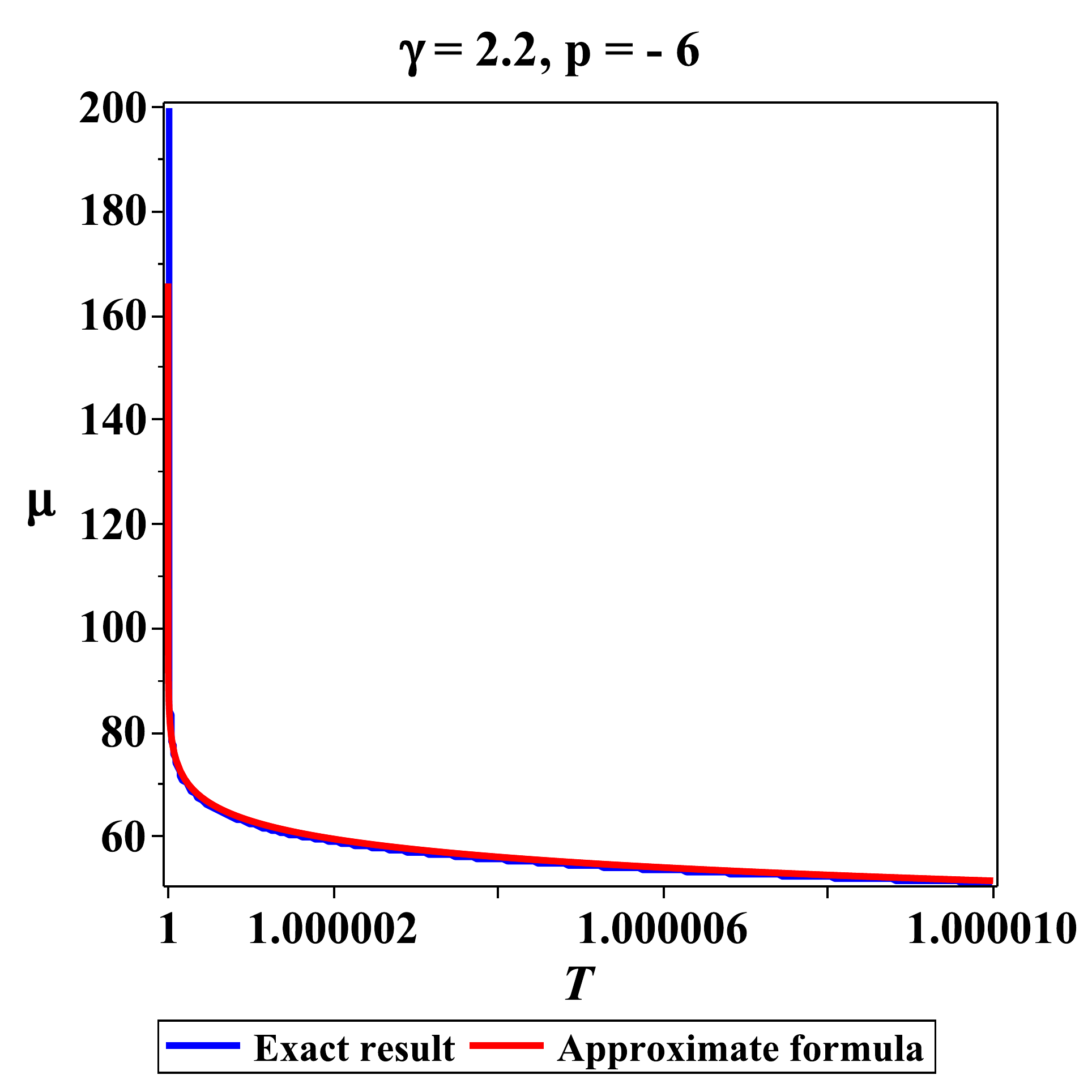} 
        \includegraphics[width=0.24\linewidth]{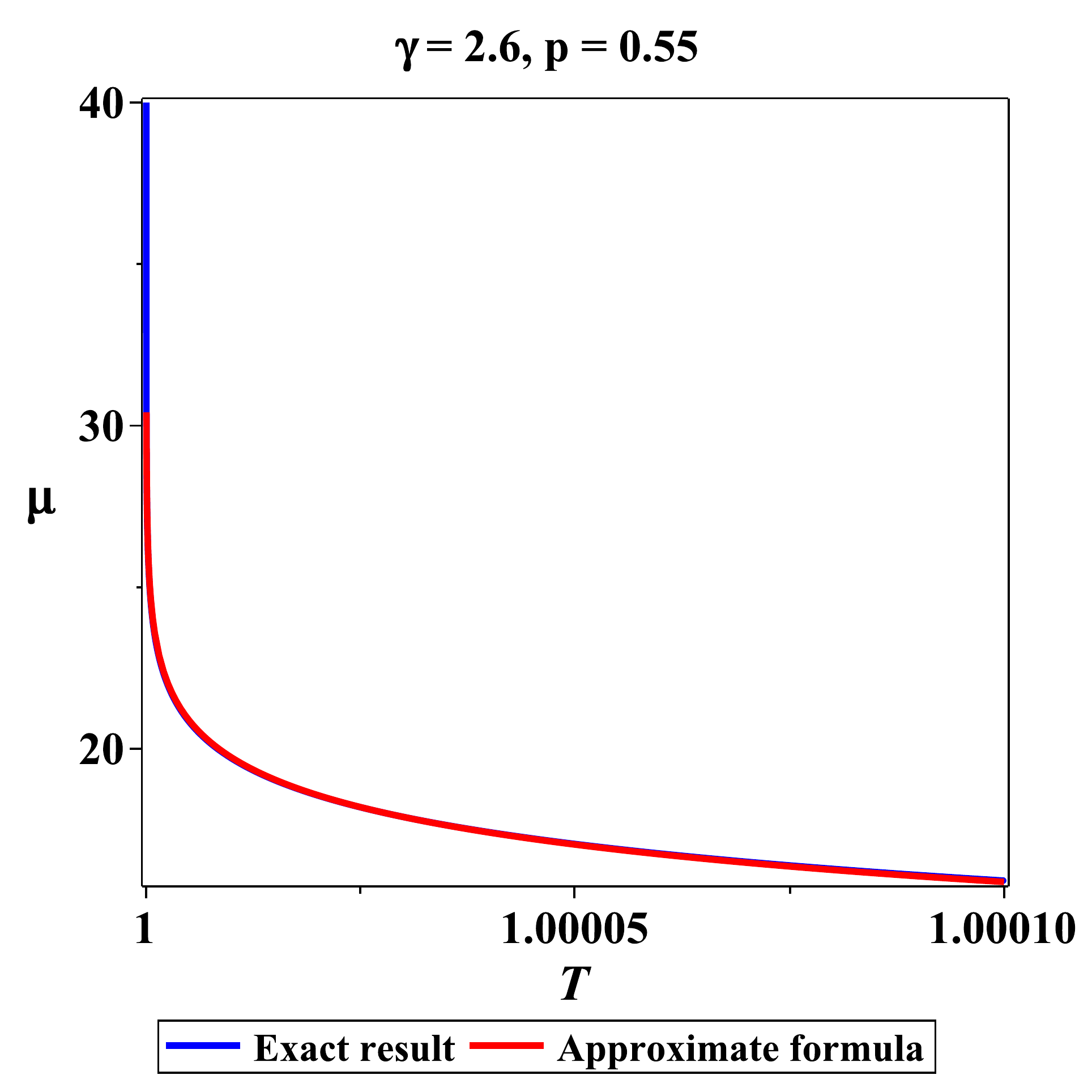} 
         \includegraphics[width=0.24\linewidth]{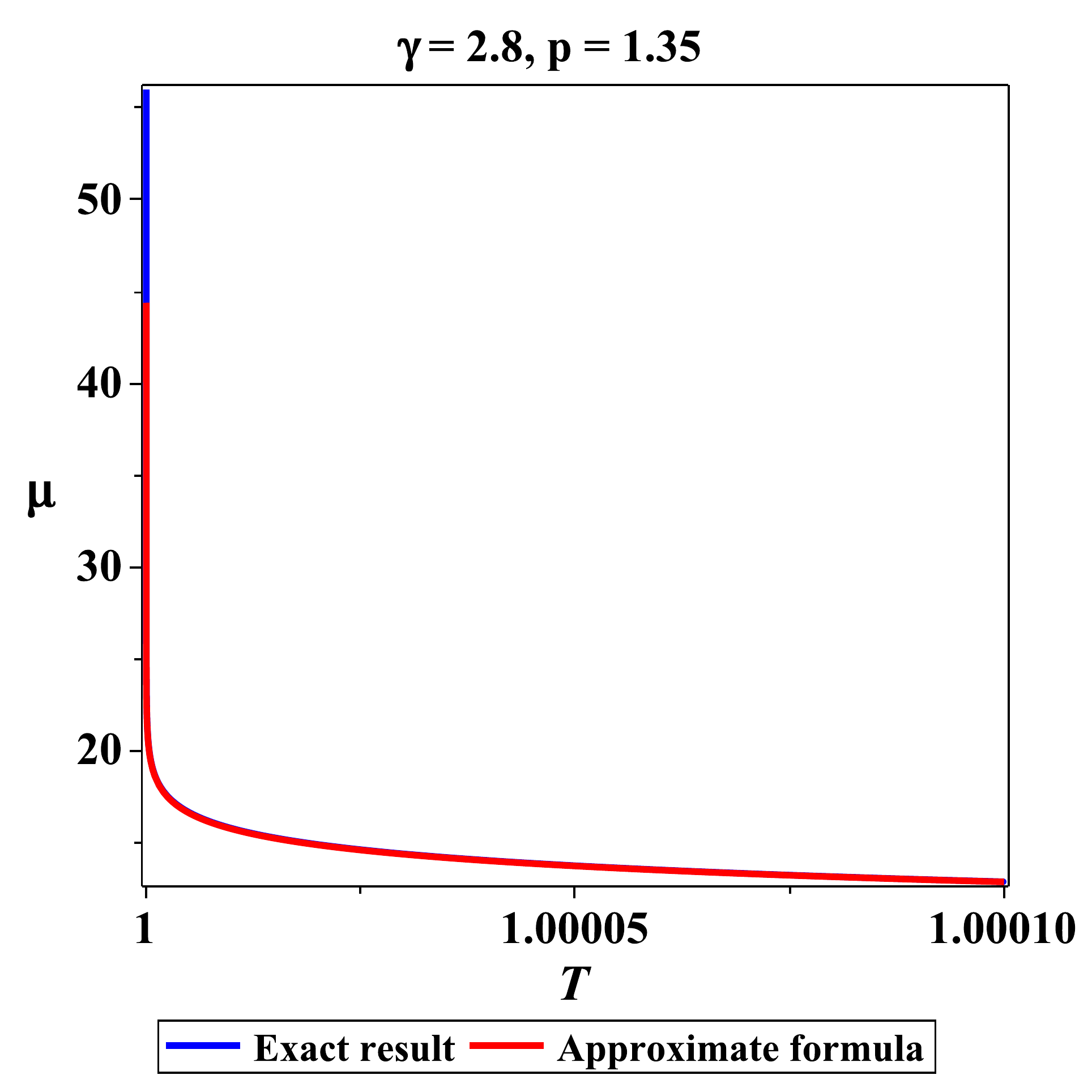} 
      \includegraphics[width=0.235\linewidth]{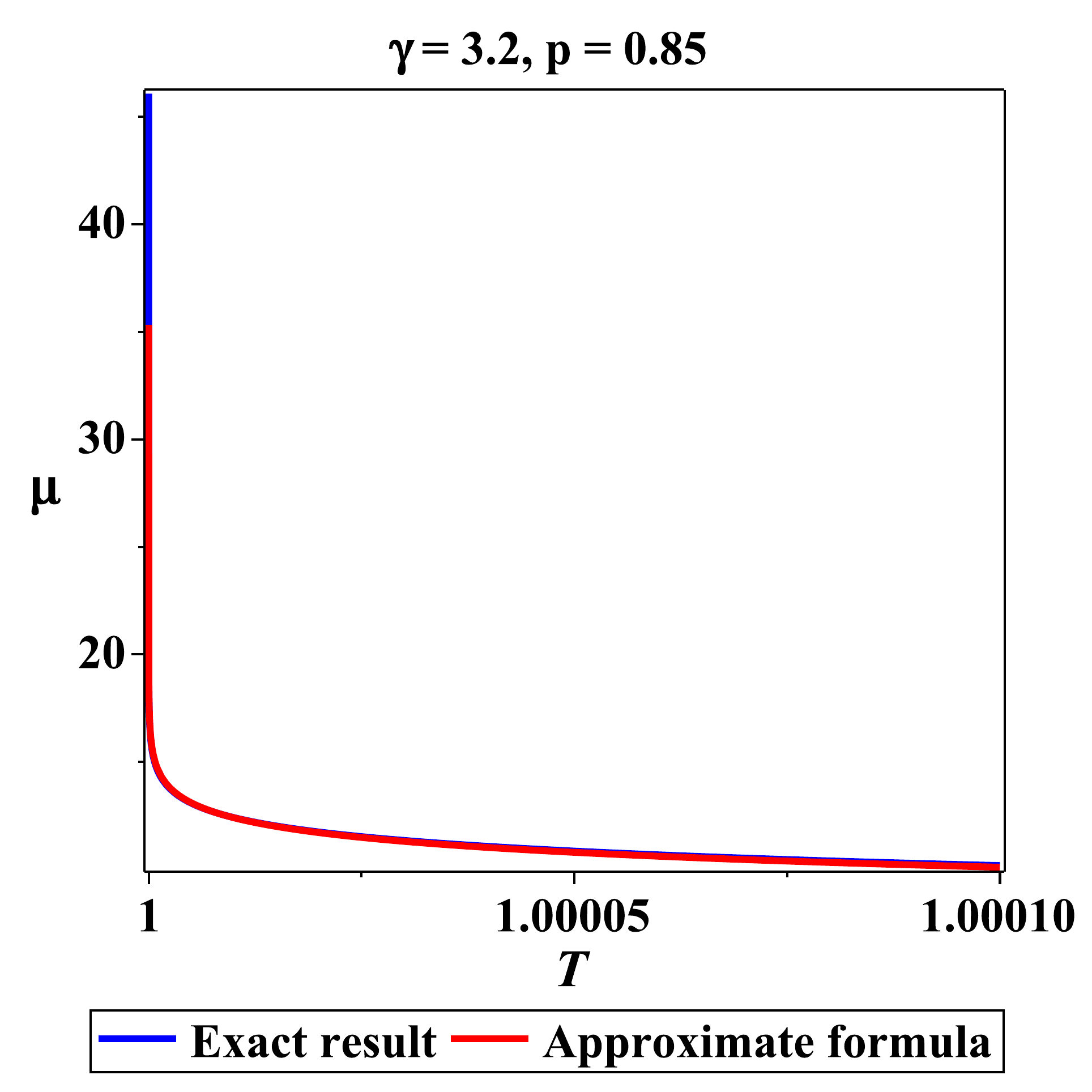} 
       \includegraphics[width=0.24\linewidth]{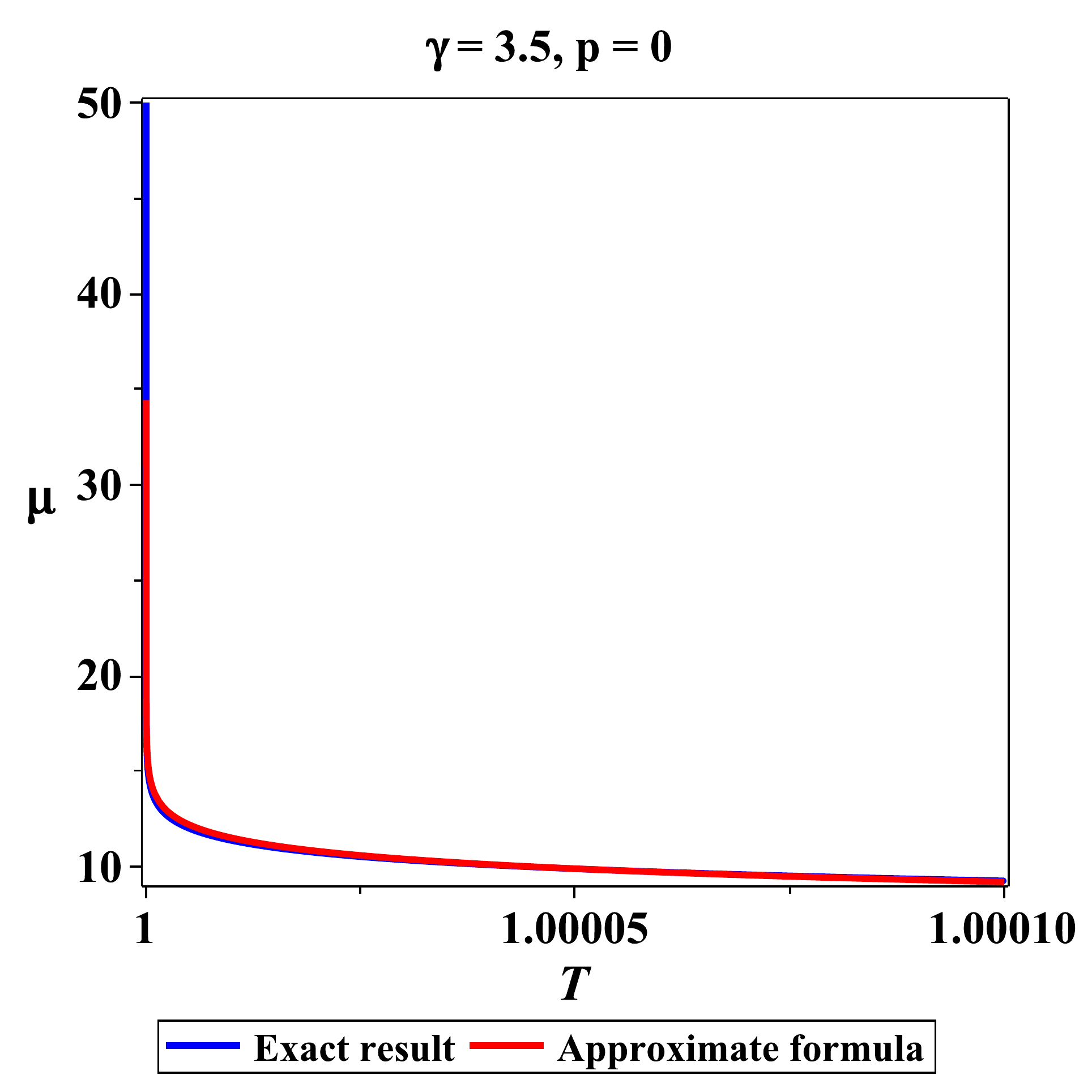} 
        \includegraphics[width=0.24\linewidth]{Fig4.eps} 
         \includegraphics[width=0.24\linewidth]{Fig5.eps} 
      	\caption{ Dependence of the chemical potential on  temperature. Blue curves present the 
      	exact result obtained as the solution of the Eq. \eqref{SM1a}. Red lines show the 
      	approximate formula.}
      	\label{figSM1}
      \end{figure}     
\begin{figure}[tbh]
      \includegraphics[width=0.7\linewidth]{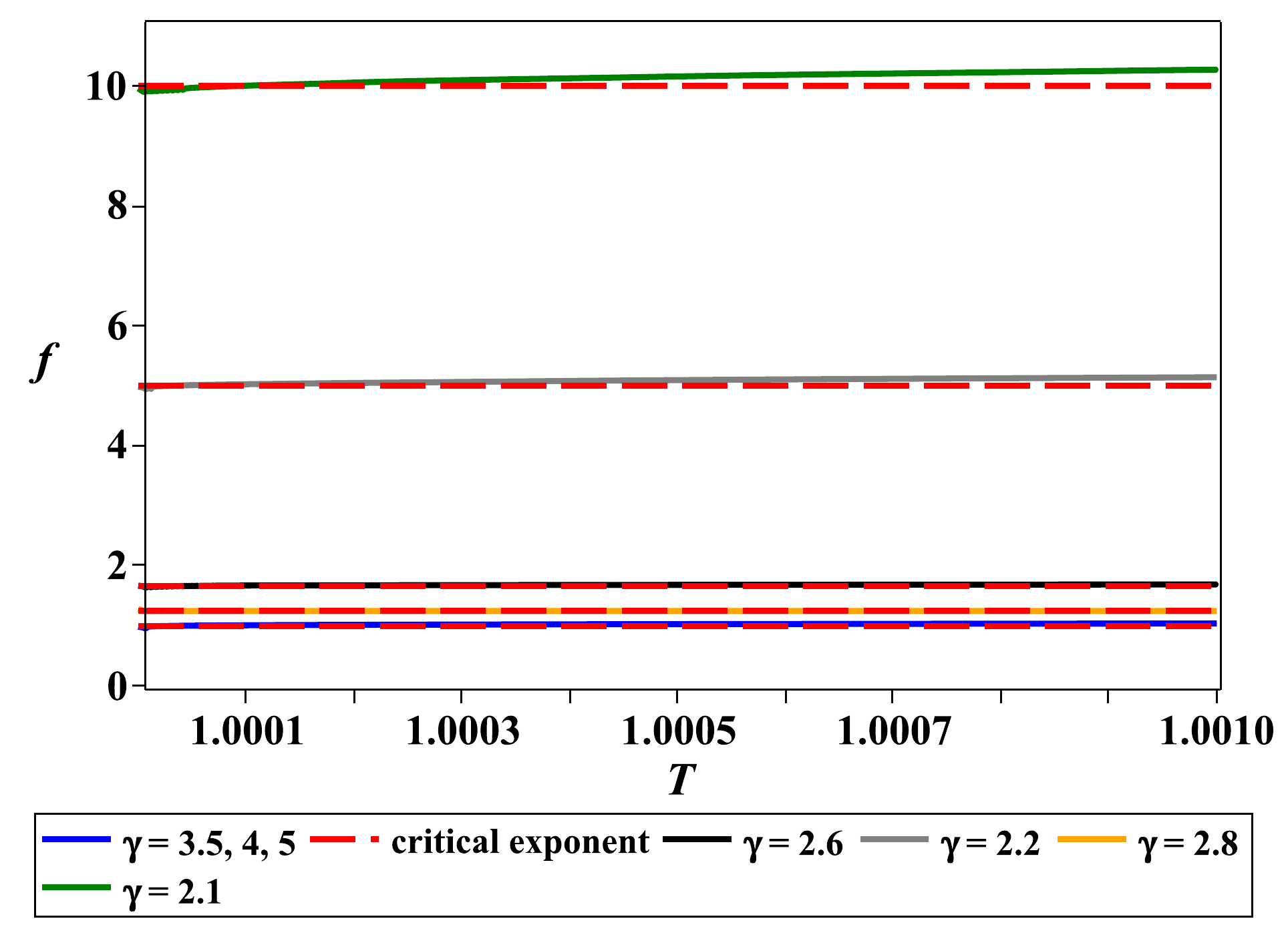} 
      	\caption{ Dependence of the function $f(\tau)$ on temperature. Blue curves depict the 
      	exact result obtained as the solution of Eq. \eqref{SM1a}. Red dashed lines present the 
      	critical exponent, $\lambda$.}
      	\label{figSM2}
      \end{figure}  

To describe the behavior of the chemical potential near the critical point, we use the trial 
function
\begin{align}
	\beta_c\mu \approx -\lambda\ln \tau +p, 
	\label{EqCP}
	\end{align}
where $\tau=(T-T_c)/T_c$ is the reduced temperature, and $p$ is a fitting constant.  As shown 
	in Fig.\ref{figSM1}, formula \eqref{EqCP} approximates  the exact solution of the Eq. 
\eqref{SM1a} well, with the choice of $\lambda$  being
\begin{align}
	\lambda = \left \{
	\begin{array} {ll}
		1/(\gamma - 2), &  2 < \gamma   <3 \\
		1, &  \gamma   \geq 3
	\end{array}
	\right .
	\label{EqL}
\end{align} 
Further it is convenient to introduce a new function $f(\tau) =(p -\beta_c\mu)/\ln \tau$. Now 
we can find $\lambda$ as a limit of $f(\tau)$, when $\tau \rightarrow 0$, i.e. $\lambda 
=\lim_{\tau \rightarrow 0}	 f(\tau)$.  After the substitution $\beta_c\mu = p -f(\tau) \ln 
\tau$ in Eq. \eqref{SM1a}, it becomes the equation that implicitly defines the dependence of the 
function $f(\tau)$ on temperature. In Fig. \ref{figSM2} we compare the results of our numerical 
solution of Eq. \eqref{SM1a} obtained for $f(\tau)$ with the critical exponent given by Eq. 
\eqref{EqL}. One can see a good agreement between both the approximate and exact results.	 


The critical exponent associated with the order parameter, $\eta = 2\langle k \rangle/N$, is 
defined as
\begin{align}
\lambda \stackrel {\rm def}{=} \lim_{\tau \rightarrow 0}	 \frac{\ln \eta}{|\ln \tau|}  .
\end{align}
In our model the order parameter is determined by the chemical potential, $\eta = 2 \nu 
e^{-\beta_c \mu}$. Substituting $\mu$ from Eq. \eqref{EqCP}, we find $\eta \propto 
\tau^\lambda$,  where the critical exponent, $\lambda$, is given by Eq. \eqref{EqL}. 

To obtain the Landau free energy dependence on the order parameter, we 
start with the exact expression for $\Omega$.
\begin{align}		
	\Omega= & -\frac{N^2}{2} \bigg(\frac{ \nu e^{-\beta_c \mu}}{\alpha  \beta} + \frac{1 }{4 
	\beta\sinh^2(\alpha \beta\mu/2)} \Big ( e^{a \beta  
				\mu}\ln \big (1 + e^{-\beta \mu} \big ) -2\ln 2 +e^{-a \beta\mu} \ln 
		\big (1 + e^{\beta \mu} \big ) + e^{a \beta\mu} \Phi \big(-e^{\beta \mu},1,\alpha \big 
		)\nonumber \\
		&-2 \Phi (-1,1,\alpha   ) + e^{-\alpha  \beta  \mu}\Phi \big (-e^{\beta \mu},1,\alpha  
		\big)\Big)\bigg).
				 \label{A11}
	\end{align}
Near the critical point one can approximate it as
\begin{align}		
	\Omega\approx & -\frac{N^2}{2 \beta_c} \big( A+\tau \beta_c \mu + e^{-\tau\beta_c \mu} 
	  {}_{2}F_{1} \big (1, a ; 1+a ;-  e^{\beta_c \mu} \big ) \Big)  e^{-\beta_c \mu} ,
				 \label{A11b}
	\end{align}
where $a \equiv\alpha |_{T_c} =\gamma -1$,
\begin{align}
	A= \frac{\gamma^2 - 3\gamma +3}{(\gamma -2)^2},
\end{align}
In derivation of Eq. \eqref{A11b} we use the relation \cite{AEWM,NIST,abr,PBM3},  
\begin{align}
{}_{q+1 }F_{q}\left(\begin{array}{c}
1, a, \ldots, a ; z \\
a+1, \ldots, a+1
\end{array}\right)=a^{q} \Phi(z, q, a).
\label{A3}
\end{align}
to replace the Lerch transcendent by the hypergeometric function  ${}_{2}F_{1} \big (1, a ; 1+a 
;-  e^{\beta_c \mu} \big )$. 

To proceed further, we use the linear transformation formulas of 
the hypergeometric function, yielding	
	\begin{align}
		{}_{2}F_{1} \big (1, a ; 1+a ;-z\big )  =& \frac{a z^{-1}}{a-1} {}_{2}F_{1} \big (1, 1-a ; 2-a 
		;-z^{-1}\big ) + \frac{\pi a z^{-a}}{\sin\pi a}, \quad a\neq 0, \pm 1,\pm 2,\dots  .
		\label{A6}
	\end{align}
For $|z| \gg 1$ we get
\begin{align}
  {}_{2}F_{1} \big (1, a ; 1+a ;- z\big ) 	= \frac{a z^{-1}}{a-1}+ \frac{\pi a  z^{ -a }}{\sin \pi a} 
  +  \mathcal O (z^{  - 2} ),
   \label{A7}
   \end{align}
and after substituting \eqref{A7} in Eq.\eqref{A11b} with $z=  e^{\beta_c \mu} $,  we obtain
\begin{align}		
	\Omega\approx  -\frac{N^2}{2 \beta_c} \Big( A+ \tau \beta_c \mu +  e^{-\tau \beta_c 
	\mu}\Big 
	(\frac{a  e^{-\beta_c 
	\mu}}{a-1}+ 	\frac{\pi a   e^{-a\beta_c \mu}}{\sin \pi a} \Big) \Big) e^{-\beta_c \mu}.
				 \label{A11b}
	\end{align}

The limiting case with $\gamma $ being integer, $\gamma = 3,4, \dots$, can be covered by 
using the linear transformation of variables \cite{abr}
\begin{equation}
{}_{2}F_{1}(a, b ; c ; -z)=(1+z)^{-a} {}_{2}F_{1}\left(a, c-b ; c ; \frac{z}{z+1}\right),
\end{equation}
followed by the hypergeometric series \cite{AEWM}
\begin{align}
	&{}_{2}F_{1}(a, b ; a+b+l ; z) =\frac{\Gamma(l) \Gamma(a+b+l)}{\Gamma(a+l) \Gamma(b+l)} 
	\sum_{n=0}^{l-1} \frac{(a)_{n}(b)_{n}}{(1-l)_{n} n !}(1-z)^{n} \nonumber \\
&+(1-z)^l(-1)^{l} \frac{\Gamma(a+b+l)}{\Gamma(a) \Gamma(b)} 
 \times  \sum_{n=0}^{\infty} \frac{(a+l)_{n}(b+l)_{n}}{n !(n+l) !}\left[k_{n}-\log 
 (1-z)\right](1-z)^{n},
\end{align}
where $l=0,1,2 \dots$ and
\begin{equation}
k_{n}=\psi(n+1)+\psi(n+1+l)-\psi(a+n+l)-\psi(b+n+l)
\end{equation}
The computation yields
\begin{align}
  {}_{2}F_{1} \big (1, a ; 1+a ;- z\big ) 	\approx \frac{a z^{-1}}{a-1} + (-1)^{a -1} k_0 z^{ -a },
   \label{A7a}
   \end{align}
where $k_0 = \psi(1)-\psi(a)$, and $\psi(z) $ denotes the Psi-function \cite{NIST}. Using 
\eqref{A7a} in Eq.\eqref{A11b}  we obtain
\begin{align}
\Omega\approx  -\frac{N^2}{2 \beta_c} \Big( A+ \tau \beta_c \mu +  e^{-\tau \beta_c 
	\mu}\Big 
	(\frac{a  e^{-\beta_c 
	\mu}}{a-1}+  (-1)^{a -1} k_0 e^{-\beta_c \mu}\Big) \Big) e^{-\beta_c \mu}. 
	\label{A8a}
		\end{align}

To complete our analysis, we substitute $e^{-\beta_c \mu} = \eta/2\nu$ in Eqs. \eqref{A11b} 
and 
\eqref{A8a}. Then leaving only the leading terms, we obtain
\begin{align}		
	\Omega  \approx -\frac{N^2}{4\nu\beta_c} \Big ( A - \tau \ln\Big ( 
	\frac{\eta}{2\nu}\Big ) \Big ) \eta\label{LFEa}.
	\end{align}
 Similar consideration for the Helmholtz free energy leads to
\begin{align}		
	F \approx -\frac{N^2}{4\nu\beta_c} \Big ( A + (1-\tau) \ln\Big ( \frac{\eta}{2\nu}\Big ) 
	\Big )\eta .
	\label{HFEb}
	\end{align}

Substituting $\eta = 2 \nu e^{-\beta_c \mu}$ in Eqs. \eqref{LFEa}, \eqref{HFEb}, we get
\begin{align}		\label{LFE3}
	\Omega & \approx -\frac{N^2}{2\beta_c} \big ( A + \tau \beta_c \mu \big )   e^{-\beta_c 
	\mu}  \\
	F& \approx -\frac{N^2}{2\beta_c} \big ( A - (1-\tau) \beta_c \mu \big )   e^{-\beta_c 
	\mu} 	\label{LFE3a}.
	\end{align}
	Next, using Eqs. \eqref{LFE3} - \eqref{LFE3a} and relations $ S  =- {\partial 
	\Omega}/{\partial T}\big |_{\mu} $ and $E=F+TS$, we obtain
\begin{align}		\label{LFE5}
	S& \approx \frac{N^2}{2} \beta_c \mu e^{-\beta_c \mu} , \\
		C_N& \approx \frac{N^2}{2} (1-\beta_c \mu )e^{-\beta_c \mu} \frac{d\mu}{dT} , \\
	E& \approx -\frac{N^2}{2\beta_c} \big ( A - (2-\tau) \beta_c \mu \big )   e^{-\beta_c 
	\mu} 	\label{HFE3}.
	\end{align}

\end{widetext}

%


\begin{thebibliography}{45}%
\makeatletter
\providecommand \@ifxundefined [1]{%
 \@ifx{#1\undefined}
}%
\providecommand \@ifnum [1]{%
 \ifnum #1\expandafter \@firstoftwo
 \else \expandafter \@secondoftwo
 \fi
}%
\providecommand \@ifx [1]{%
 \ifx #1\expandafter \@firstoftwo
 \else \expandafter \@secondoftwo
 \fi
}%
\providecommand \natexlab [1]{#1}%
\providecommand \enquote  [1]{``#1''}%
\providecommand \bibnamefont  [1]{#1}%
\providecommand \bibfnamefont [1]{#1}%
\providecommand \citenamefont [1]{#1}%
\providecommand \href@noop [0]{\@secondoftwo}%
\providecommand \href [0]{\begingroup \@sanitize@url \@href}%
\providecommand \@href[1]{\@@startlink{#1}\@@href}%
\providecommand \@@href[1]{\endgroup#1\@@endlink}%
\providecommand \@sanitize@url [0]{\catcode `\\12\catcode `\$12\catcode
  `\&12\catcode `\#12\catcode `\^12\catcode `\_12\catcode `\%12\relax}%
\providecommand \@@startlink[1]{}%
\providecommand \@@endlink[0]{}%
\providecommand \url  [0]{\begingroup\@sanitize@url \@url }%
\providecommand \@url [1]{\endgroup\@href {#1}{\urlprefix }}%
\providecommand \urlprefix  [0]{URL }%
\providecommand \Eprint [0]{\href }%
\providecommand \doibase [0]{http://dx.doi.org/}%
\providecommand \selectlanguage [0]{\@gobble}%
\providecommand \bibinfo  [0]{\@secondoftwo}%
\providecommand \bibfield  [0]{\@secondoftwo}%
\providecommand \translation [1]{[#1]}%
\providecommand \BibitemOpen [0]{}%
\providecommand \bibitemStop [0]{}%
\providecommand \bibitemNoStop [0]{.\EOS\space}%
\providecommand \EOS [0]{\spacefactor3000\relax}%
\providecommand \BibitemShut  [1]{\csname bibitem#1\endcsname}%
\let\auto@bib@innerbib\@empty
\bibitem [{\citenamefont {Bollob\'as}(2001)}]{BB1}%
  \BibitemOpen
  \bibfield  {author} {\bibinfo {author} {\bibfnamefont {B.}\ \bibnamefont
  {Bollob\'as}},\ }\href@noop {} {\emph {\bibinfo {title} {{Random Graphs}}}}\
  (\bibinfo  {publisher} {Cambridge University Press, Cambridge},\ \bibinfo
  {year} {2001})\BibitemShut {NoStop}%
\bibitem [{\citenamefont {{ S. N. Dorogovtsev and J. F. F.
  Mendes}}(2003)}]{DSMF}%
  \BibitemOpen
  \bibfield  {author} {\bibinfo {author} {\bibnamefont {{ S. N. Dorogovtsev and
  J. F. F. Mendes}}},\ }\href@noop {} {\emph {\bibinfo {title} {Evolution of
  Networks: From Biological Nets to the Internet and WWW}}}\ (\bibinfo
  {publisher} {Oxford University Press, Oxford},\ \bibinfo {year}
  {2003})\BibitemShut {NoStop}%
\bibitem [{\citenamefont {Caldarelli}(2007)}]{GCal}%
  \BibitemOpen
  \bibfield  {author} {\bibinfo {author} {\bibfnamefont {G.}\ \bibnamefont
  {Caldarelli}},\ }\href@noop {} {\emph {\bibinfo {title} {Scale-Free Networks:
  Complex Webs in Nature, and Technology}}}\ (\bibinfo  {publisher} {Oxford
  University Press, Oxford},\ \bibinfo {year} {2007})\BibitemShut {NoStop}%
\bibitem [{\citenamefont {Barrat}\ \emph {et~al.}(2008)\citenamefont {Barrat},
  \citenamefont {Barthelemy},\ and\ \citenamefont {Vespignani}}]{BABM}%
  \BibitemOpen
  \bibfield  {author} {\bibinfo {author} {\bibfnamefont {A.}~\bibnamefont
  {Barrat}}, \bibinfo {author} {\bibfnamefont {M.}~\bibnamefont {Barthelemy}},
  \ and\ \bibinfo {author} {\bibfnamefont {A.}~\bibnamefont {Vespignani}},\
  }\href@noop {} {\emph {\bibinfo {title} {{Dynamical Processes on Complex
  Networks}}}}\ (\bibinfo  {publisher} {{Cambridge University Press,
  Cambridge}},\ \bibinfo {year} {2008})\BibitemShut {NoStop}%
\bibitem [{\citenamefont {Barab\'asi}(2016)}]{BAL}%
  \BibitemOpen
  \bibfield  {author} {\bibinfo {author} {\bibfnamefont {A.-L.}\
  \bibnamefont {Barab\'asi}},\ }
  {\emph {\bibinfo {title} {Network Science}}}\ (\bibinfo  {publisher}
  {Cambridge University Press, Cambridge},\ \bibinfo {year} {2016})\BibitemShut {NoStop}%
\bibitem [{\citenamefont {Newman}(2018)}]{MN2018}%
  \BibitemOpen
  \bibfield  {author} {\bibinfo {author} {\bibfnamefont {M.}\ \bibnamefont
  {Newman}},\ }\href@noop {} {\emph {\bibinfo {title} {Networks}}}\ (\bibinfo
  {publisher} {Oxford University Press, Oxford},\ \bibinfo {year}
  {2018})\BibitemShut {NoStop}%
\bibitem [{\citenamefont {Albert}\ and\ \citenamefont
  {Barab\'asi}(2002)}]{ARB}%
  \BibitemOpen
  \bibfield  {author} {\bibinfo {author} {\bibfnamefont {R.}\ \bibnamefont
  {Albert}}\ and\ \bibinfo {author} {\bibfnamefont {A.-L.}\
  \bibnamefont {Barab\'asi}},\ }\bibfield  {title} {\enquote {\bibinfo {title}
  {Statistical mechanics of complex networks},}\ }\href@noop {} {\bibfield
  {journal} {\bibinfo  {journal} {Rev. Mod. Phys.}\ }\textbf {\bibinfo {volume}
  {74}},\ \bibinfo {pages} {47--97} (\bibinfo {year} {2002})}\BibitemShut
  {NoStop}%
\bibitem [{\citenamefont {Watts}\ and\ \citenamefont {Strogatz}(1998)}]{WDSS}%
  \BibitemOpen
  \bibfield  {author} {\bibinfo {author} {\bibfnamefont {Duncan~J.}\
  \bibnamefont {Watts}}\ and\ \bibinfo {author} {\bibfnamefont {Steven~H.}\
  \bibnamefont {Strogatz}},\ }\bibfield  {title} {\enquote {\bibinfo {title}
  {Collective dynamics of `small-world' networks},}\ }\href@noop {} {\bibfield
  {journal} {\bibinfo  {journal} {Nature}\ }\textbf {\bibinfo {volume} {393}},\
  \bibinfo {pages} {440 -- 442} (\bibinfo {year} {1998})}\BibitemShut {NoStop}%
\bibitem [{\citenamefont {Kwapie\'{n}}\ and\ \citenamefont
  {Drozdz}(2012)}]{KJSD}%
  \BibitemOpen
  \bibfield  {author} {\bibinfo {author} {\bibfnamefont {J.}\
  \bibnamefont {Kwapie\'{n}}}\ and\ \bibinfo {author} {\bibfnamefont
  {S.}\ \bibnamefont {Drozdz}},\ }\bibfield  {title} {\enquote {\bibinfo
  {title} {Physical approach to complex systems},}\ }\href {\doibase
  https://doi.org/10.1016/j.physrep.2012.01.007} {\bibfield  {journal}
  {\bibinfo  {journal} {Physics Reports}\ }\textbf {\bibinfo {volume} {515}},\
  \bibinfo {pages} {115 -- 226} (\bibinfo {year} {2012})}\BibitemShut {NoStop}%
\bibitem [{\citenamefont {Barab{\'a}si}(2009)}]{BAL1}%
  \BibitemOpen
  \bibfield  {author} {\bibinfo {author} {\bibfnamefont
  {Albert-L{\'a}szl{\'o}}\ \bibnamefont {Barab{\'a}si}},\ }\bibfield  {title}
  {\enquote {\bibinfo {title} {Scale-free networks: A decade and beyond},}\
  }\href {\doibase 10.1126/science.1173299} {\bibfield  {journal} {\bibinfo
  {journal} {Science}\ }\textbf {\bibinfo {volume} {325}},\ \bibinfo {pages}
  {412--413} (\bibinfo {year} {2009})}\BibitemShut {NoStop}%
\bibitem [{\citenamefont {Girvan}\ and\ \citenamefont {Newman}(2002)}]{GMNM}%
  \BibitemOpen
  \bibfield  {author} {\bibinfo {author} {\bibfnamefont {M.}~\bibnamefont
  {Girvan}}\ and\ \bibinfo {author} {\bibfnamefont {M.~E.~J.}\ \bibnamefont
  {Newman}},\ }\bibfield  {title} {\enquote {\bibinfo {title} {Community
  structure in social and biological networks},}\ }\href {\doibase
  10.1073/pnas.122653799} {\bibfield  {journal} {\bibinfo  {journal}
  {Proceedings of the National Academy of Sciences}\ }\textbf {\bibinfo
  {volume} {99}},\ \bibinfo {pages} {7821--7826} (\bibinfo {year}
  {2002})}\BibitemShut {NoStop}%
\bibitem [{\citenamefont {Voitalov}\ \emph {et~al.}(2019)\citenamefont
  {Voitalov}, \citenamefont {van~der Hoorn}, \citenamefont {van~der Hofstad},\
  and\ \citenamefont {Krioukov}}]{VIHP}%
  \BibitemOpen
  \bibfield  {author} {\bibinfo {author} {\bibfnamefont {I.}\ \bibnamefont
  {Voitalov}}, \bibinfo {author} {\bibfnamefont {P.}\ \bibnamefont {van~der
  Hoorn}}, \bibinfo {author} {\bibfnamefont {R.}\ \bibnamefont {van~der
  Hofstad}}, \ and\ \bibinfo {author} {\bibfnamefont {Dmitri}\ \bibnamefont
  {Krioukov}},\ }\bibfield  {title} {\enquote {\bibinfo {title} {Scale-free
  networks well done},}\ }\href {\doibase 10.1103/PhysRevResearch.1.033034}
  {\bibfield  {journal} {\bibinfo  {journal} {Phys. Rev. Research}\ }\textbf
  {\bibinfo {volume} {1}},\ \bibinfo {pages} {033034} (\bibinfo {year}
  {2019})}\BibitemShut {NoStop}%
\bibitem [{\citenamefont {Newman}\ \emph {et~al.}(2001)\citenamefont {Newman},
  \citenamefont {Strogatz},\ and\ \citenamefont {Watts}}]{NMSW}%
  \BibitemOpen
  \bibfield  {author} {\bibinfo {author} {\bibfnamefont {M.~E.~J.}\
  \bibnamefont {Newman}}, \bibinfo {author} {\bibfnamefont {S.~H.}\
  \bibnamefont {Strogatz}}, \ and\ \bibinfo {author} {\bibfnamefont {D.~J.}\
  \bibnamefont {Watts}},\ }\bibfield  {title} {\enquote {\bibinfo {title}
  {Random graphs with arbitrary degree distributions and their applications},}\
  }\href@noop {} {\bibfield  {journal} {\bibinfo  {journal} {Phys. Rev. E}\
  }\textbf {\bibinfo {volume} {64}},\ \bibinfo {pages} {026118} (\bibinfo
  {year} {2001})}\BibitemShut {NoStop}%
\bibitem [{\citenamefont {Newman}(2003)}]{MEJN1}%
  \BibitemOpen
  \bibfield  {author} {\bibinfo {author} {\bibfnamefont {M.~E.~J.}\
  \bibnamefont {Newman}},\ }\bibfield  {title} {\enquote {\bibinfo {title} {The
  structure and function of complex networks},}\ }\href@noop {} {\bibfield
  {journal} {\bibinfo  {journal} {SIAM Review}\ }\textbf {\bibinfo {volume}
  {45}},\ \bibinfo {pages} {167--256} (\bibinfo {year} {2003})}\BibitemShut
  {NoStop}%
\bibitem [{\citenamefont {Romualdo Pastor-Satorras}(2003)}]{RPMR}%
  \BibitemOpen
  \bibfield  {author} {\bibinfo {author} {\bibfnamefont {A.
  Diaz-Guilera~(eds.)}\ \bibnamefont {R. Pastor-Satorras}, \bibfnamefont
  {M. Rubi}},\ }\href@noop {} {\emph {\bibinfo {title} {Statistical
  Mechanics of Complex Networks}}},\ Lecture Notes in Physics 625\ (\bibinfo
  {publisher} {Springer-Verlag, New York},\ \bibinfo {year} {2003})\BibitemShut
  {NoStop}%
\bibitem [{\citenamefont {Park}\ and\ \citenamefont {Newman}(2004)}]{PJNM}%
  \BibitemOpen
  \bibfield  {author} {\bibinfo {author} {\bibfnamefont {J.}\ \bibnamefont
  {Park}}\ and\ \bibinfo {author} {\bibfnamefont {M.~E.~J.}\ \bibnamefont
  {Newman}},\ }\bibfield  {title} {\enquote {\bibinfo {title} {Statistical
  mechanics of networks},}\ }\href@noop {} {\bibfield  {journal} {\bibinfo
  {journal} {Phys. Rev. E}\ }\textbf {\bibinfo {volume} {70}},\ \bibinfo
  {pages} {066117} (\bibinfo {year} {2004})}\BibitemShut {NoStop}%
\bibitem [{\citenamefont {Garlaschelli}\ and\ \citenamefont
  {Loffredo}(2008)}]{CDLM}%
  \BibitemOpen
  \bibfield  {author} {\bibinfo {author} {\bibfnamefont {D.}\ \bibnamefont
  {Garlaschelli}}\ and\ \bibinfo {author} {\bibfnamefont {M.~I.}\
  \bibnamefont {Loffredo}},\ }\bibfield  {title} {\enquote {\bibinfo {title}
  {Maximum likelihood: Extracting unbiased information from complex
  networks},}\ }\href {\doibase 10.1103/PhysRevE.78.015101} {\bibfield
  {journal} {\bibinfo  {journal} {Phys. Rev. E}\ }\textbf {\bibinfo {volume}
  {78}},\ \bibinfo {pages} {015101(R)} (\bibinfo {year} {2008})}\BibitemShut
  {NoStop}%
\bibitem [{\citenamefont {Cimini}\ \emph {et~al.}(2019)\citenamefont {Cimini},
  \citenamefont {Squartini}, \citenamefont {Saracco}, \citenamefont
  {Garlaschelli}, \citenamefont {Gabrielli},\ and\ \citenamefont
  {Caldarelli}}]{CGST}%
  \BibitemOpen
  \bibfield  {author} {\bibinfo {author} {\bibfnamefont {G.}\ \bibnamefont
  {Cimini}}, \bibinfo {author} {\bibfnamefont {T.}\ \bibnamefont
  {Squartini}}, \bibinfo {author} {\bibfnamefont {F.}\ \bibnamefont
  {Saracco}}, \bibinfo {author} {\bibfnamefont {D.}\ \bibnamefont
  {Garlaschelli}}, \bibinfo {author} {\bibfnamefont {A.}\ \bibnamefont
  {Gabrielli}}, \ and\ \bibinfo {author} {\bibfnamefont {G.}\ \bibnamefont
  {Caldarelli}},\ }\bibfield  {title} {\enquote {\bibinfo {title} {The
  statistical physics of real-world networks},}\ }\href {\doibase
  10.1038/s42254-018-0002-6} {\bibfield  {journal} {\bibinfo  {journal} {Nature
  Reviews Physics}\ }\textbf {\bibinfo {volume} {1}},\ \bibinfo {pages}
  {58--71} (\bibinfo {year} {2019})}\BibitemShut {NoStop}%
\bibitem [{\citenamefont {Mikulecky}(2001)}]{MDC}%
  \BibitemOpen
  \bibfield  {author} {\bibinfo {author} {\bibfnamefont {D.~C.}\
  \bibnamefont {Mikulecky}},\ }\bibfield  {title} {\enquote {\bibinfo {title}
  {Network thermodynamics and complexity: a transition to relational systems
  theory},}\ }\href {\doibase https://doi.org/10.1016/S0097-8485(01)00072-9}
  {\bibfield  {journal} {\bibinfo  {journal} {Computers $\&$ Chemistry}\
  }\textbf {\bibinfo {volume} {25}},\ \bibinfo {pages} {369 -- 391} (\bibinfo
  {year} {2001})}\BibitemShut {NoStop}%
\bibitem [{\citenamefont {Garlaschelli}\ \emph {et~al.}(2013)\citenamefont
  {Garlaschelli}, \citenamefont {Ahnert}, \citenamefont {Fink},\ and\
  \citenamefont {Caldarelli}}]{CDAS}%
  \BibitemOpen
  \bibfield  {author} {\bibinfo {author} {\bibfnamefont {D.}\ \bibnamefont
  {Garlaschelli}}, \bibinfo {author} {\bibfnamefont {S.~E.}\
  \bibnamefont {Ahnert}}, \bibinfo {author} {\bibfnamefont {T. M.~A.}\
  \bibnamefont {Fink}}, \ and\ \bibinfo {author} {\bibfnamefont {G.}\
  \bibnamefont {Caldarelli}},\ }\bibfield  {title} {\enquote {\bibinfo {title}
  {Low-temperature behaviour of social and economic networks},}\ }\href
  {\doibase 10.3390/e15083238} {\bibfield  {journal} {\bibinfo  {journal}
  {Entropy}\ }\textbf {\bibinfo {volume} {15}},\ \bibinfo {pages} {3148--3169}
  (\bibinfo {year} {2013})}\BibitemShut {NoStop}%
\bibitem [{\citenamefont {Krioukov}\ \emph {et~al.}(2009)\citenamefont
  {Krioukov}, \citenamefont {Papadopoulos}, \citenamefont {Vahdat},\ and\
  \citenamefont {Bogu\~n\'a}}]{KDPF1}%
  \BibitemOpen
  \bibfield  {author} {\bibinfo {author} {\bibfnamefont {D.}\ \bibnamefont
  {Krioukov}}, \bibinfo {author} {\bibfnamefont {F.}\ \bibnamefont
  {Papadopoulos}}, \bibinfo {author} {\bibfnamefont {A.}\ \bibnamefont
  {Vahdat}}, \ and\ \bibinfo {author} {\bibfnamefont {M.}\ \bibnamefont
  {Bogu\~n\'a}},\ }\bibfield  {title} {\enquote {\bibinfo {title} {Curvature
  and temperature of complex networks},}\ }\href@noop {} {\bibfield  {journal}
  {\bibinfo  {journal} {Phys. Rev. E}\ }\textbf {\bibinfo {volume} {80}},\
  \bibinfo {pages} {035101(R)} (\bibinfo {year} {2009})}\BibitemShut {NoStop}%
\bibitem [{\citenamefont {Krioukov}\ \emph {et~al.}(2010)\citenamefont
  {Krioukov}, \citenamefont {Papadopoulos}, \citenamefont {Kitsak},
  \citenamefont {Vahdat},\ and\ \citenamefont {Bogu\~n\'a}}]{KDPF2}%
  \BibitemOpen
  \bibfield  {author} {\bibinfo {author} {\bibfnamefont {D.}\ \bibnamefont
  {Krioukov}}, \bibinfo {author} {\bibfnamefont {F.}\ \bibnamefont
  {Papadopoulos}}, \bibinfo {author} {\bibfnamefont {M.}\ \bibnamefont
  {Kitsak}}, \bibinfo {author} {\bibfnamefont {A.}\ \bibnamefont {Vahdat}}, \
  and\ \bibinfo {author} {\bibfnamefont {M.}\ \bibnamefont
  {Bogu\~n\'a}},\ }\bibfield  {title} {\enquote {\bibinfo {title} {Hyperbolic
  geometry of complex networks},}\ }\href@noop {} {\bibfield  {journal}
  {\bibinfo  {journal} {Phys. Rev. E}\ }\textbf {\bibinfo {volume} {82}},\
  \bibinfo {pages} {036106} (\bibinfo {year} {2010})}\BibitemShut {NoStop}%
\bibitem [{\citenamefont {Nesterov}\ and\ \citenamefont
  {Mata~Villafuerte}(2020)}]{ANHM}%
  \BibitemOpen
  \bibfield  {author} {\bibinfo {author} {\bibfnamefont {A.~I.}\
  \bibnamefont {Nesterov}}\ and\ \bibinfo {author} {\bibfnamefont
  {P.~H.}\ \bibnamefont {Mata~Villafuerte}},\ }\bibfield  {title}
  {\enquote {\bibinfo {title} {Complex networks in the framework of
  nonassociative geometry},}\ }\href {\doibase 10.1103/PhysRevE.101.032302}
  {\bibfield  {journal} {\bibinfo  {journal} {Phys. Rev. E}\ }\textbf {\bibinfo
  {volume} {101}},\ \bibinfo {pages} {032302} (\bibinfo {year}
  {2020})}\BibitemShut {NoStop}%
\bibitem [{\citenamefont {Callaway}\ \emph {et~al.}(2000)\citenamefont
  {Callaway}, \citenamefont {Newman}, \citenamefont {Strogatz},\ and\
  \citenamefont {Watts}}]{DCNMS}%
  \BibitemOpen
  \bibfield  {author} {\bibinfo {author} {\bibfnamefont {D.~S.}\
  \bibnamefont {Callaway}}, \bibinfo {author} {\bibfnamefont {M.~E.~J.}\
  \bibnamefont {Newman}}, \bibinfo {author} {\bibfnamefont {S.~H.}\
  \bibnamefont {Strogatz}}, \ and\ \bibinfo {author} {\bibfnamefont
  {Duncan~J.}\ \bibnamefont {Watts}},\ }\bibfield  {title} {\enquote {\bibinfo
  {title} {Network robustness and fragility: Percolation on random graphs},}\
  }\href {\doibase 10.1103/PhysRevLett.85.5468} {\bibfield  {journal} {\bibinfo
   {journal} {Phys. Rev. Lett.}\ }\textbf {\bibinfo {volume} {85}},\ \bibinfo
  {pages} {5468--5471} (\bibinfo {year} {2000})}\BibitemShut {NoStop}%
\bibitem [{\citenamefont {Dorogovtsev}\ \emph {et~al.}(2008)\citenamefont
  {Dorogovtsev}, \citenamefont {Goltsev},\ and\ \citenamefont
  {Mendes}}]{DSGAM}%
  \BibitemOpen
  \bibfield  {author} {\bibinfo {author} {\bibfnamefont {S.~N.}\ \bibnamefont
  {Dorogovtsev}}, \bibinfo {author} {\bibfnamefont {A.~V.}\ \bibnamefont
  {Goltsev}}, \ and\ \bibinfo {author} {\bibfnamefont {J.~F.~F.}\ \bibnamefont
  {Mendes}},\ }\bibfield  {title} {\enquote {\bibinfo {title} {Critical
  phenomena in complex networks},}\ }\href {\doibase
  10.1103/RevModPhys.80.1275} {\bibfield  {journal} {\bibinfo  {journal} {Rev.
  Mod. Phys.}\ }\textbf {\bibinfo {volume} {80}},\ \bibinfo {pages}
  {1275--1335} (\bibinfo {year} {2008})}\BibitemShut {NoStop}%
\bibitem [{\citenamefont {Bogu\~n\'a}\ and\ \citenamefont
  {Pastor-Satorras}(2003)}]{BMPSR}%
  \BibitemOpen
  \bibfield  {author} {\bibinfo {author} {\bibfnamefont {M.}\
  \bibnamefont {Bogu\~n\'a}}\ and\ \bibinfo {author} {\bibfnamefont {R.}\
  \bibnamefont {Pastor-Satorras}},\ }\bibfield  {title} {\enquote {\bibinfo
  {title} {Class of correlated random networks with hidden variables},}\ }\href
  {\doibase 10.1103/PhysRevE.68.036112} {\bibfield  {journal} {\bibinfo
  {journal} {Phys. Rev. E}\ }\textbf {\bibinfo {volume} {68}},\ \bibinfo
  {pages} {036112} (\bibinfo {year} {2003})}\BibitemShut {NoStop}%
\bibitem [{\citenamefont {Serrano}\ \emph {et~al.}(2008)\citenamefont
  {Serrano}, \citenamefont {Krioukov},\ and\ \citenamefont
  {Bogu\~n\'a}}]{SMKD1}%
  \BibitemOpen
  \bibfield  {author} {\bibinfo {author} {\bibfnamefont {M.~\'A.}\
  \bibnamefont {Serrano}}, \bibinfo {author} {\bibfnamefont {D.}\
  \bibnamefont {Krioukov}}, \ and\ \bibinfo {author} {\bibfnamefont {Mari\'an}\
  \bibnamefont {Bogu\~n\'a}},\ }\bibfield  {title} {\enquote {\bibinfo {title}
  {Self-similarity of complex networks and hidden metric spaces},}\ }\href
  {\doibase 10.1103/PhysRevLett.100.078701} {\bibfield  {journal} {\bibinfo
  {journal} {Phys. Rev. Lett.}\ }\textbf {\bibinfo {volume} {100}},\ \bibinfo
  {pages} {078701} (\bibinfo {year} {2008})}\BibitemShut {NoStop}%
\bibitem [{\citenamefont {Boccaletti}\ \emph {et~al.}(2006)\citenamefont
  {Boccaletti}, \citenamefont {Latora}, \citenamefont {Moreno}, \citenamefont
  {Chavez},\ and\ \citenamefont {Hwang}}]{BSLV}%
  \BibitemOpen
  \bibfield  {author} {\bibinfo {author} {\bibfnamefont {S.}~\bibnamefont
  {Boccaletti}}, \bibinfo {author} {\bibfnamefont {V.}~\bibnamefont {Latora}},
  \bibinfo {author} {\bibfnamefont {Y.}~\bibnamefont {Moreno}}, \bibinfo
  {author} {\bibfnamefont {M.}~\bibnamefont {Chavez}}, \ and\ \bibinfo {author}
  {\bibfnamefont {D.-U.}\ \bibnamefont {Hwang}},\ }\bibfield  {title} {\enquote
  {\bibinfo {title} {Complex networks: Structure and dynamics},}\ }\href
  {\doibase https://doi.org/10.1016/j.physrep.2005.10.009} {\bibfield
  {journal} {\bibinfo  {journal} {Physics Reports}\ }\textbf {\bibinfo {volume}
  {424}},\ \bibinfo {pages} {175 -- 308} (\bibinfo {year} {2006})}\BibitemShut
  {NoStop}%
\bibitem [{\citenamefont {Squartini}\ and\ \citenamefont
  {Garlaschelli}(2017)}]{STCD}%
  \BibitemOpen
  \bibfield  {author} {\bibinfo {author} {\bibfnamefont {T.}\ \bibnamefont
  {Squartini}}\ and\ \bibinfo {author} {\bibfnamefont {D.}\ \bibnamefont
  {Garlaschelli}},\ }\href@noop {} {\emph {\bibinfo {title} {{Maximum-Entropy
  Networks: Pattern Detection, Network Reconstruction and Graph
  Combinatorics}}}}\ (\bibinfo  {publisher} {Springer International
  Publishing},\ \bibinfo {address} {Cham},\ \bibinfo {year} {2017})\BibitemShut
  {NoStop}%
\bibitem [{\citenamefont {van~der Hoorn}\ \emph {et~al.}(2018)\citenamefont
  {van~der Hoorn}, \citenamefont {Lippner},\ and\ \citenamefont
  {Krioukov}}]{HPLG}%
  \BibitemOpen
  \bibfield  {author} {\bibinfo {author} {\bibfnamefont {P.}\ \bibnamefont
  {van~der Hoorn}}, \bibinfo {author} {\bibfnamefont {G.}\ \bibnamefont
  {Lippner}}, \ and\ \bibinfo {author} {\bibfnamefont {D.}\ \bibnamefont
  {Krioukov}},\ }\bibfield  {title} {\enquote {\bibinfo {title} {{Sparse
  Maximum-Entropy Random Graphs with a Given Power-Law Degree Distribution}},}\
  }\href {\doibase 10.1007/s10955-017-1887-7} {\bibfield  {journal} {\bibinfo
  {journal} {Journal of Statistical Physics}\ }\textbf {\bibinfo {volume}
  {173}},\ \bibinfo {pages} {806--844} (\bibinfo {year} {2018})}\BibitemShut
  {NoStop}%
\bibitem [{\citenamefont {Anand}\ \emph {et~al.}(2014)\citenamefont {Anand},
  \citenamefont {Krioukov},\ and\ \citenamefont {Bianconi}}]{KAKD}%
  \BibitemOpen
  \bibfield  {author} {\bibinfo {author} {\bibfnamefont {K.}\ \bibnamefont
  {Anand}}, \bibinfo {author} {\bibfnamefont {D.}\ \bibnamefont
  {Krioukov}}, \ and\ \bibinfo {author} {\bibfnamefont {G.}\ \bibnamefont
  {Bianconi}},\ }\bibfield  {title} {\enquote {\bibinfo {title} {Entropy
  distribution and condensation in random networks with a given degree
  distribution},}\ }\href {\doibase 10.1103/PhysRevE.89.062807} {\bibfield
  {journal} {\bibinfo  {journal} {Phys. Rev. E}\ }\textbf {\bibinfo {volume}
  {89}},\ \bibinfo {pages} {062807} (\bibinfo {year} {2014})}\BibitemShut
  {NoStop}%
\bibitem [{\citenamefont {Voitalov}\ \emph {et~al.}(2020)\citenamefont
  {Voitalov}, \citenamefont {van~der Hoorn}, \citenamefont {Kitsak},
  \citenamefont {Papadopoulos},\ and\ \citenamefont {Krioukov}}]{VIHPKM}%
  \BibitemOpen
  \bibfield  {author} {\bibinfo {author} {\bibfnamefont {I.}\ \bibnamefont
  {Voitalov}}, \bibinfo {author} {\bibfnamefont {P.}\ \bibnamefont {van~der
  Hoorn}}, \bibinfo {author} {\bibfnamefont {M.}\ \bibnamefont {Kitsak}},
  \bibinfo {author} {\bibfnamefont {F.}\ \bibnamefont {Papadopoulos}},
  \ and\ \bibinfo {author} {\bibfnamefont {D.}\ \bibnamefont {Krioukov}},\
  }\bibfield  {title} {\enquote {\bibinfo {title} {Weighted hypersoft
  configuration model},}\ }\href {\doibase 10.1103/PhysRevResearch.2.043157}
  {\bibfield  {journal} {\bibinfo  {journal} {Phys. Rev. Research}\ }\textbf
  {\bibinfo {volume} {2}},\ \bibinfo {pages} {043157} (\bibinfo {year}
  {2020})}\BibitemShut {NoStop}%
\bibitem [{\citenamefont {Barrat}\ and\ \citenamefont {Weigt}(2000)}]{BAWM}%
  \BibitemOpen
  \bibfield  {author} {\bibinfo {author} {\bibfnamefont {A.}~\bibnamefont
  {Barrat}}\ and\ \bibinfo {author} {\bibfnamefont {M.}~\bibnamefont {Weigt}},\
  }\bibfield  {title} {\enquote {\bibinfo {title} {On the properties of
  small-world network models},}\ }\href {\doibase 10.1007/s100510050067}
  {\bibfield  {journal} {\bibinfo  {journal} {The European Physical Journal B -
  Condensed Matter and Complex Systems}\ }\textbf {\bibinfo {volume} {13}},\
  \bibinfo {pages} {547--560} (\bibinfo {year} {2000})}\BibitemShut {NoStop}%
\bibitem [{\citenamefont {Bollob\'as}\ and\ \citenamefont
  {Riordan}(2003)}]{BBOR}%
  \BibitemOpen
  \bibfield  {author} {\bibinfo {author} {\bibfnamefont {B.}\ \bibnamefont
  {Bollob\'as}}\ and\ \bibinfo {author} {\bibfnamefont {O.~M.}\
  \bibnamefont {Riordan}},\ }\enquote {\bibinfo {title} {Mathematical results
  on scale-free random graphs},}\ in\ \href@noop {} {\emph {\bibinfo
  {booktitle} {{Handbook of graphs and networks: from the genome to the
  Internet}}}},\ \bibinfo {editor} {edited by\ \bibinfo {editor} {\bibfnamefont
  {Stefan}\ \bibnamefont {Bornholdt}}\ and\ \bibinfo {editor} {\bibfnamefont
  {Heinz~Georg}\ \bibnamefont {Schuster}}}\ (\bibinfo  {publisher} {Wiley-VCH,
  Weinheim},\ \bibinfo {year} {2003})\BibitemShut {NoStop}%
\bibitem [{\citenamefont {Wilf}(2006)}]{WHS}%
  \BibitemOpen
  \bibfield  {author} {\bibinfo {author} {\bibfnamefont {H.~S.}\
  \bibnamefont {Wilf}},\ }{\emph {\bibinfo {title} {Generatingfunctionology}}},\ \bibinfo {edition}
  {3rd}\ ed.\ (\bibinfo  {publisher} {A K Peters, Wellesley, MA},\ \bibinfo
  {year} {2006})\BibitemShut {NoStop}%
\bibitem [{\citenamefont {Chung}\ and\ \citenamefont
  {Lu}(2002{\natexlab{a}})}]{CFLL}%
  \BibitemOpen
  \bibfield  {author} {\bibinfo {author} {\bibfnamefont {F.}\ \bibnamefont
  {Chung}}\ and\ \bibinfo {author} {\bibfnamefont {L.}\ \bibnamefont
  {Lu}},\ }\bibfield  {title} {\enquote {\bibinfo {title} {{Connected
  Components in Random Graphs with Given Expected Degree Sequences}},}\
  }\href@noop {} {\bibfield  {journal} {\bibinfo  {journal} {Annals of
  Combinatorics}\ }\textbf {\bibinfo {volume} {6}},\ \bibinfo {pages}
  {125--145} (\bibinfo {year} {2002}{\natexlab{a}})}\BibitemShut {NoStop}%
\bibitem [{\citenamefont {Chung}\ and\ \citenamefont
  {Lu}(2002{\natexlab{b}})}]{CFLL1}%
  \BibitemOpen
  \bibfield  {author} {\bibinfo {author} {\bibfnamefont {F.}\ \bibnamefont
  {Chung}}\ and\ \bibinfo {author} {\bibfnamefont {Linyuan}\ \bibnamefont
  {Lu}},\ }\bibfield  {title} {\enquote {\bibinfo {title} {The average
  distances in random graphs with given expected degrees},}\ }\href@noop {}
  {\bibfield  {journal} {\bibinfo  {journal} {Proceedings of the National
  Academy of Sciences}\ }\textbf {\bibinfo {volume} {99}},\ \bibinfo {pages}
  {15879} (\bibinfo {year} {2002}{\natexlab{b}})}\BibitemShut {NoStop}%
\bibitem [{\citenamefont {Caldarelli}\ \emph {et~al.}(2002)\citenamefont
  {Caldarelli}, \citenamefont {Capocci}, \citenamefont {De~Los~Rios},\ and\
  \citenamefont {Mu\~noz}}]{CGCA}%
  \BibitemOpen
  \bibfield  {author} {\bibinfo {author} {\bibfnamefont {G.}~\bibnamefont
  {Caldarelli}}, \bibinfo {author} {\bibfnamefont {A.}~\bibnamefont {Capocci}},
  \bibinfo {author} {\bibfnamefont {P.}~\bibnamefont {De Los Rios}}, \ and\
  \bibinfo {author} {\bibfnamefont {M.~A.}\ \bibnamefont {Mu\~noz}},\
  }\bibfield  {title} {\enquote {\bibinfo {title} {Scale-free networks from
  varying vertex intrinsic fitness},}\ }\href {\doibase
  10.1103/PhysRevLett.89.258702} {\bibfield  {journal} {\bibinfo  {journal}
  {Phys. Rev. Lett.}\ }\textbf {\bibinfo {volume} {89}},\ \bibinfo {pages}
  {258702} (\bibinfo {year} {2002})}\BibitemShut {NoStop}%
\bibitem [{\citenamefont {Servedio}\ \emph {et~al.}(2004)\citenamefont
  {Servedio}, \citenamefont {Caldarelli},\ and\ \citenamefont
  {Butt\`a}}]{SVCG}%
  \BibitemOpen
  \bibfield  {author} {\bibinfo {author} {\bibfnamefont {V. D.~P.}\
  \bibnamefont {Servedio}}, \bibinfo {author} {\bibfnamefont {G.}\
  \bibnamefont {Caldarelli}}, \ and\ \bibinfo {author} {\bibfnamefont {P.}\
  \bibnamefont {Butt\`a}},\ }\bibfield  {title} {\enquote {\bibinfo {title}
  {Vertex intrinsic fitness: How to produce arbitrary scale-free networks},}\
  }\href {\doibase 10.1103/PhysRevE.70.056126} {\bibfield  {journal} {\bibinfo
  {journal} {Phys. Rev. E}\ }\textbf {\bibinfo {volume} {70}},\ \bibinfo
  {pages} {056126(R)} (\bibinfo {year} {2004})}\BibitemShut {NoStop}%
\bibitem [{\citenamefont {Squartini}\ and\ \citenamefont
  {Garlaschelli}(2011)}]{TSGD}%
  \BibitemOpen
  \bibfield  {author} {\bibinfo {author} {\bibfnamefont {T.}\ \bibnamefont
  {Squartini}}\ and\ \bibinfo {author} {\bibfnamefont {D.}\ \bibnamefont
  {Garlaschelli}},\ }\bibfield  {title} {\enquote {\bibinfo {title} {Analytical
  maximum-likelihood method to detect patterns in real networks},}\ }\href
  {\doibase 10.1088/1367-2630/13/8/083001} {\bibfield  {journal} {\bibinfo
  {journal} {New Journal of Physics}\ }\textbf {\bibinfo {volume} {13}},\
  \bibinfo {pages} {083001} (\bibinfo {year} {2011})}\BibitemShut {NoStop}%
\bibitem [{\citenamefont {Park}\ and\ \citenamefont {Newman}(2003)}]{PYNM2}%
  \BibitemOpen
  \bibfield  {author} {\bibinfo {author} {\bibfnamefont {J.}\ \bibnamefont
  {Park}}\ and\ \bibinfo {author} {\bibfnamefont {M.~E.~J.}\ \bibnamefont
  {Newman}},\ }\bibfield  {title} {\enquote {\bibinfo {title} {{Origin of
  degree correlations in the Internet and other networks}},}\ }\href {\doibase
  10.1103/PhysRevE.68.026112} {\bibfield  {journal} {\bibinfo  {journal} {Phys.
  Rev. E}\ }\textbf {\bibinfo {volume} {68}},\ \bibinfo {pages} {026112}
  (\bibinfo {year} {2003})}\BibitemShut {NoStop}%
  \bibitem [{\citenamefont {Frank W. J.~Olver}(2010)}]{NIST}%
  \BibitemOpen
  \bibfield  {author} {\bibinfo {author} {\bibfnamefont {R. F. Boisvert, 
  C. W.~Clark},\ \bibnamefont {F. W. J.~Olver}, \bibfnamefont {D.
  W.~Lozier}},\ }\href@noop {} {\emph {\bibinfo {title} {{NIST Handbook of
  Mathematical Functions}}}}\ (\bibinfo  {publisher} {Cambridge University
  Press, Cambridge},\ \bibinfo {year} {2010})\BibitemShut {NoStop}%
  \bibitem [{\citenamefont {A.}\ \emph {et~al.}(1953)\citenamefont {A.},
  \citenamefont {W.},\ and\ \citenamefont {F.}}]{AEWM}%
  \BibitemOpen
  \bibfield  {author} {\bibinfo {author} {\bibfnamefont {E.}\
  \bibnamefont {A.}}, \bibinfo {author} {\bibfnamefont {Magnus}\ \bibnamefont
  {W.}}, \ and\ \bibinfo {author} {\bibfnamefont {Oberhettinger}\ \bibnamefont
  {F.}},\ }\href@noop {} {\emph {\bibinfo {title} {{ Higher Transcendental
  Functions, Vol. I.}}}}\ (\bibinfo  {publisher} {McGraw-Hill},\ \bibinfo
  {address} {New York, NY, USA},\ \bibinfo {year} {1953})\BibitemShut {NoStop}%
    \bibitem [{\citenamefont {Prudnikov}\ \emph {et~al.}(2002)\citenamefont
  {Prudnikov}, \citenamefont {Brychkov},\ and\ \citenamefont
  {Marichev}}]{PBM3}%
  \BibitemOpen
  \bibfield  {author} {\bibinfo {author} {\bibfnamefont {A.~P.}\ \bibnamefont
  {Prudnikov}}, \bibinfo {author} {\bibfnamefont {Yu.~A.}\ \bibnamefont
  {Brychkov}}, \ and\ \bibinfo {author} {\bibfnamefont {O.~I.}\ \bibnamefont
  {Marichev}},\ }\href@noop {} {\emph {\bibinfo {title} {Integrals and series.
  3, More special functions}}}\ (\bibinfo  {publisher} {Gordon and Breach
  Science Publishers},\ \bibinfo {year} {2002})\BibitemShut {NoStop}%
\bibitem [{\citenamefont {Abramowitz}\ and\ \citenamefont
  {Stegun}(1965)}]{abr}%
  \BibitemOpen
  \bibinfo {editor} {\bibfnamefont {M.}~\bibnamefont {Abramowitz}}\ and\
  \bibinfo {editor} {\bibfnamefont {I.~A.}\ \bibnamefont {Stegun}},\ eds.,\
  \href@noop {} {\emph {\bibinfo {title} {{Handbook of Mathematical
  Functions}}}}\ (\bibinfo  {publisher} {Dover},\ \bibinfo {address} {New
  York},\ \bibinfo {year} {1965})\BibitemShut {NoStop}%
\bibitem [{\citenamefont {Cohen}\ \emph {et~al.}(2002)\citenamefont {Cohen},
  \citenamefont {ben-Avraham},\ and\ \citenamefont {Havlin}}]{RAS}%
  \BibitemOpen
  \bibfield  {author} {\bibinfo {author} {\bibfnamefont {R.}\ \bibnamefont
  {Cohen}}, \bibinfo {author} {\bibfnamefont {D.}\ \bibnamefont {ben-Avraham}}, \ and\ \bibinfo {author} {\bibfnamefont {S.}\ \bibnamefont
  {Havlin}},\ }\bibfield  {title} {\enquote {\bibinfo {title} {Percolation
  critical exponents in scale-free networks},}\ }\href {\doibase
  10.1103/PhysRevE.66.036113} {\bibfield  {journal} {\bibinfo  {journal} {Phys.
  Rev. E}\ }\textbf {\bibinfo {volume} {66}},\ \bibinfo {pages} {036113}
  (\bibinfo {year} {2002})}\BibitemShut {NoStop}%
  \bibitem [{\citenamefont {Prudnikov}\ \emph {et~al.}(2002)\citenamefont
  {Prudnikov}, \citenamefont {Brychkov},\ and\ \citenamefont
  {Marichev}}]{PBM1}%
  \BibitemOpen
  \bibfield  {author} {\bibinfo {author} {\bibfnamefont {A.~P.}\ \bibnamefont
  {Prudnikov}}, \bibinfo {author} {\bibfnamefont {Yu.~A.}\ \bibnamefont
  {Brychkov}},  \bibinfo {author} {\bibfnamefont {O.~I.}\ \bibnamefont
  {Marichev}},\ and\ \bibinfo {author} {\bibfnamefont {N.~M.}\ \bibnamefont
  {Queen (Translator)}},\ }\href@noop {} {\emph {\bibinfo {title} {	Integrals and Series: elementary functions}}}\ (\bibinfo  {publisher} {CRC},\ \bibinfo {year} {1998})\BibitemShut {NoStop}%
\end{thebibliography}

\end{document}